\documentclass[a4paper,11pt]{article}

\usepackage{amsthm}
\usepackage{fullpage}%
\usepackage{graphicx,tikz,enumerate}
\usepackage[colorlinks=true]{hyperref}
\usepackage[title]{appendix}
\usepackage{amssymb, amsfonts,mathtools,stmaryrd}
\usepackage{xcolor}
\usepackage{enumitem}
\usepackage{etoolbox}

\usepackage{stmaryrd}

\usepackage{algorithm,algpseudocode}

\usepackage{dsfont}
\usepackage{bbm}
\newtheorem{definition}{\textbf{Definition}}
\newtheorem{theorem}[definition]{\textbf{Theorem}}
\newtheorem{lemma}[definition]{\textbf{Lemma}}
\newtheorem{remark}[definition]{\textbf{Remark}}
\newtheorem{proposition}[definition]{\textbf{Proposition}}
\newtheorem{problem}[definition]{\textbf{Problem}}

\newtheorem*{proof*}{\textbf{Proof}}
\newcommand{\brck}[1]{\llbracket #1 \rrbracket}
\newcommand{\msf}[1]{\mathsf{#1}}
\newcommand{\Ag}{\msf{Ag}}
\newcommand{\fanBr}[1]{\langle\!\langle#1\rangle\!\rangle\!}
\newcommand{\fanBrOp}[1]{\langle\!\langle#1\rangle\!\rangle}

\newcommand{\N}{\mathbb{N}}

\newcommand{\B}{\mathbb{B}}

\newcommand{\SAT}{\mathsf{SAT}}

\newcommand{\post}{\mathrm{post}}

\newcommand{\prop}{\mathcal{P}}

\newcommand{\propformula}{\Omega}
\DeclareMathOperator{\lF}{\mathbf{F}}
\DeclareMathOperator{\lG}{\mathbf{G}}
\DeclareMathOperator{\lU}{\mathbf{U}}
\DeclareMathOperator{\lX}{\mathbf{X}}

\DeclareMathOperator{\false}{\mathit{false}}
\DeclareMathOperator{\true}{\mathit{true}}
\DeclareMathOperator{\lA}{\mathbf{A}\!}
\DeclareMathOperator{\lE}{\mathbf{E}\!}

\newcommand{\placeholder}{?}

\newcommand{\Sat}{\mathtt{SAT}_{C}}
\newcommand{\sample}{\mathcal{S}}

\usepackage{subcaption}
\usepackage{pifont}
\usepackage{authblk}
\usepackage{graphicx}
\usepackage{multirow}

\usepackage{pgfplots}
\usepackage{wrapfig}
\usepackage{booktabs}

\author[1,2]{Benjamin Bordais}
\author[1,2]{Daniel Neider}
\author[3]{Rajarshi Roy}
\affil[1]{TU Dortmund University, Dortmund, Germany}
\affil[2]{Center for Trustworthy Data Science and Security, University Alliance Ruhr, Dortmund, Germany}
\affil[3]{Max Planck Institute for Software Systems, Kaiserslautern, Germany}
\date{}

\begin{document}
	\title{Learning Branching-Time Properties in CTL and ATL via Constraint Solving}
	\maketitle	
	
\begin{abstract}
	We address the problem of learning temporal properties from the branching-time behavior of systems.
	%In recent years, learning temporal properties from system behavior has crystallized as an effective approach to synthesizing specifications for formal verification and explaining the temporal behavior of systems.
	Existing research in this field has mostly focused on learning linear temporal properties specified using popular logics, such as Linear Temporal Logic (LTL) and Signal Temporal Logic (STL).
	Branching-time logics such as Computation Tree Logic (CTL) and Alternating-time Temporal Logic (ATL), despite being extensively used in specifying and verifying distributed and multi-agent systems, have not received adequate attention.
	Thus, in this paper, we investigate the problem of learning CTL and ATL formulas from examples of system behavior.
	As input to the learning problems, we rely on the typical representations of branching behavior as Kripke structures and concurrent game structures, respectively.
	Given a sample of structures, we learn concise formulas by encoding the learning problem into a satisfiability problem, most notably %using The crux of this reduction involves 
	by symbolically encoding both the search for prospective formulas and their fixed-point based model checking algorithms.
	We %additionally 
	also study the decision problem of checking the existence of prospective %CTL and 
	ATL formulas for a given sample.
	We implement our algorithms in an Python prototype and have evaluated them to extract several common CTL and ATL formulas used in practical applications.
\end{abstract}

%!TEX root = LearnATL.tex

\section{Introduction}

Formal verification relies on the fact that formal specifications, which are precise descriptions of the design requirements, are either readily available or can be constructed easily.
This assumption, however, often proves to be unrealistic as constructing specifications manually is not only tedious but also prone to errors.
As a result, for years, the availability of precise, functional, and usable specifications has been one of the biggest bottlenecks of formal methods~\cite{AmmonsBL02,BjornerH14,Rozier16}.

To tackle this serious limitation, recent research has concentrated on automatically generating specifications, especially in temporal logics.
There is a large body of works targeted towards learning specifications in linear-time logics such as Linear Temporal Logic (LTL)~\cite{flie,CamachoM19}, Metric Temporal Logic (MTL)~\cite{DBLP:conf/vmcai/RahaRFNP24}, Signal Temporal Logic (STL)~\cite{dtbombara,MohammadinejadD20}, etc.
These approaches not only generate reliable specifications but can also be used to infer interpretable descriptions for complex temporal behaviors.

Along with linear-time logics, branching-time logics have had a significant impact on formal verification.
Computation Tree Logic (CTL)~\cite{DBLP:conf/lop/ClarkeE81}, which combines temporal operators such as $\lX$ (next), $\lF$ (finally), $\lG$ (globally) with the branching quantifiers $\lE$ (exists) and $\lA$ (all), is a specification language of choice for numerous verification tools~\cite{DBLP:conf/hybrid/BengtssonLLPY95,DBLP:journals/tse/Holzmann97,DBLP:conf/cav/CimattiCGR99}.
Alternating-time Temporal Logic (ATL)~\cite{AlurATL}, which augments CTL with a ``cooperation'' quantifier $\fanBrOp{\cdot}$ to reason about interaction of multiple agents, is popular in specifying properties for distributed and multi-agent systems~\cite{DBLP:conf/cav/AlurHMQRT98,DBLP:conf/tacas/LomuscioR06} and has several applications in AI domains~\cite{DBLP:journals/sLogica/HoekW03a,DBLP:journals/tocl/MogaveroMPV14}.
%For further details on these logics, we refer to Section~\ref{sec:prelim}.

Despite the significance of branching-time logics, learning such specifications has received considerably less attention.
The few existing works handle the problem of completing user-defined queries~\cite{Chan00}, i.e., specifications with missing parts, or searching for specifications based on few restricted templates~\cite{WasylkowskiZ11} such as $\lA \lF \placeholder$, $\lA \lG (\placeholder_1 \rightarrow \lF\placeholder_2)$, etc.
These works are limited in their generality and require one to handcraft queries/templates suitable for learning.
%works include  is by Chan, which addresses the problem of completing queries,  in CTL.
%Another similar work is by Wasylkowski and Zeller~\cite{WasylkowskiZ11}, which mines invariants in CTL for Java programs using queries of few restricted forms
%such as .

Towards learning arbitrary branching-time properties, we consider the \emph{passive learning} problem for both CTL and ATL. This problem requires, given a sample $\sample$ of positive (or desirable) and negative (or undesirable) structures, to infer a minimal CTL/ATL formula that is consistent with $\sample$. 
Passive learning is widely studied in the literature~\cite{BiermannF72,Gold78,flie} and forms a significant part of many learning frameworks~\cite{CamachoM19,ltl-from-positive-only} (see Section~\ref{sec:passive-learning} for elaboration).
In our learning problem, we consider, as input, structures typical for describing branching behavior: Kripke structures (KSs) and concurrent game structures (CGSs), which model single and multi-agent systems, respectively.

%The structures that we consider here are
%As input to the problem, we consider standard models describing branching behavior--- 
%Kripke structures (KSs) and concurrent game structures (CGSs), which represent single and multi-agent systems, respectively.
%Given a sample $\sample$ of positive (or desirable) and negative (or undesirable) structures, the passive learning problem asks to infer a minimal CTL/ATL formula that is consistent with $\sample$, see 
%We formally introduce the problem in 
%Problem~\ref{prob:pass-learning} in Section~\ref{sec:passive-learning} .

%\paragraph{Our contributions} %for space reasons
To address the passive learning 
problem, we design algorithms for %both 
CTL and ATL that employ constraint-solving techniques to learn prospective formulas.
%Specifically, f
Following Neider and Gavran~\cite{flie}, our algorithms search for prospective formulas of size $n$ for a given sample $\sample$, by encoding the problem into the satisfiability of a propositional formula $\propformula^{\sample}_{n}$. 
This formula is then solved using an off-the-shelf SAT solver to obtain a CTL/ATL formula of size $n$ if one exists.
The crux of the SAT encoding %formula $\propformula^{\sample}_{n}$
lies in symbolically encoding both the structure of the formula and %, concurrently, encoding 
the standard fixed-point based model-checking for the symbolic formula.

To present the technical details of the encoding, we focus on passive learning of ATL formulas from CGSs, as ATL and CGSs generalize CTL and KSs for multi-agent settings, respectively. In particular, restricting the number of agents to one simply reduces ATL to CTL and CGSs to KSs.
Nonetheless, we highlight aspects of the learning algorithm that improve %when considering 
in the case of CTL passive learning% of CTL
.

%Alongside the learning algorithms, w
We also initiate the theoretical study of the ATL passive learning problem%for ATL
.
%We show that the problem of learning a CTL/ATL formula of size $n$ is in $\NP$ due to the polynomial size (in $n$ and $|\sample|$) of our encoding $\propformula$. 
%Moreover, w
We study the decision problem of whether there is an ATL formula consistent with $\sample$. We extend already existing results for CTL to the case of ATL: we show that the decision problem for full ATL can be solved in polynomial time (Thm.~\ref{thm:decide_polynomial_time}, extending  \cite[Thm. 3.2, 3.9]{DBLP:journals/tcs/BrowneCG88}). We also show that,
for any fragment of ATL, the same decision problem can be decided in exponential time (Thm.~\ref{thm:bound_size_formulas}, extending \cite[Thm. 3]{DBLP:journals/pacmpl/KrogmeierM23}). In the same theorem, we exhibit an exponential bound on the size of the formulae that need 
to be considered to find a consistent one, regardless of the fragment considered (parallelizing \cite[Coro. 1]{DBLP:journals/corr/abs-2402-06366}). %We give more details in Section~\ref{subsec:deciding_sep}.

We have implemented our learning algorithms in an open-source prototype that can access an array of SAT solvers.
We evaluate the prototype on synthetic benchmarks consisting of samples of KSs and CGSs that reflect typical branching-time properties.
We observed that our algorithms display the ability to learn formulas from samples of varying sizes.
Further, we demonstrated improvements to the SAT encoding for enhanced runtime performance.
%We present all of the experimental results in Section~\ref{sec:implem}% and conclude with a discussion in Section~\ref{}
%.

\subsubsection{Related Work}
As alluded to in the introduction, the majority of works in inferring temporal logics focus on linear-time logics.
For LTL, many works consider learning based on handcrafted templates or queries, which are incomplete formulas~\cite{LiDS11,ShahKSL18,KimMSAS19}.
Few others learn formulas of arbitrary syntactic structure in LTL (or its important fragments) either by exploiting constraint solving~\cite{flie,CamachoM19,Riener19} or efficient enumerative search~\cite{scarlet,DBLP:journals/corr/abs-2402-12373}.
Some recent works rely on neuro-symbolic approaches to learn LTL formulas from noisy data~\cite{DBLP:conf/aaai/LuoLDWPZ22,DBLP:conf/aaai/WanLDLYP24}. For STL, most works focus on learning formulas of particular syntactic structure~\cite{dtbombara,dtmethod} or searching time intervals for STL formulas of known structure~\cite{asarin,rpstl1,rpstl2}.
A handful of works consider learning STL formulas of arbitrary structure~\cite{MohammadinejadD20,genetic}. 
Apart from LTL and STL, there are works on learning several other logics such as Metric Temporal Logic~\cite{DBLP:conf/vmcai/RahaRFNP24}, Past LTL~\cite{ArifLERCT20}, Property Specification Language~\cite{0002FN20}, etc.

In contrast to linear-time logics, research on learning branching-time properties remains relatively sparse.
Chan~\cite{Chan00} considers the problem of completing CTL queries---incomplete CTL formulas with missing (Boolean combinations of) atomic propositions.
A related work by Wasylkowski and Zeller~\cite{WasylkowskiZ11} considers inferring operational preconditions for Java methods in CTL.
Both of these works are limited in their ability to search through large number of CTL formulas of arbitrary syntactic structure.
As a result, they resort to user-defined queries or handcrafted templates to reduce the search space of specifications.

A recent preprint by Pommellet et al.~\cite{DBLP:journals/corr/abs-2402-06366}, yet to be published, addresses the problem of learning CTL formulas from a sample of Kripke structures (KSs).
Their learning algorithm follows a SAT-based paradigm
and uses an encoding similar to ours.
% similar to ours, utilizing an encoding in propositional logic with a similar flavor.
Our encoding for CTL was developed independently~\cite{DBLP:journals/corr/abs-2310-13778}.
In this paper, we study the more general problem of learning ATL formulas from CGSs, which conceptually subsumes the problem of learning CTL formulas from KSs.
%Our work devises an encoding for learning ATL, studying the existence of prospective ATL formulas, and empirically evaluating the ATL learning algorithm.

Another work that devises a similar encoding is the one by Bertrand et al.~\cite{BertrandFS12}.
Their SMT encoding, albeit similar, is tailored towards solving a different problem of synthesizing small models for probabilistic CTL (PCTL) formulas.

%!TEX root = LearnATL.tex

\section{Preliminaries}
\label{sec:prelim}

We refer to the set of positive integers by $\N_1$. For $n \in \N_1$, we let $\brck{1,n} \subseteq \N_1$ denote the set $\{ 1,\ldots,n \}$. Furthermore, given any non-empty set $Q$, we let $Q^*,Q^+,Q^\omega$ denote the set of finite, non-empty finite, and infinite sequences of elements in $Q$, respectively. For all $\bullet \in \{ *,+,\omega \}$, $\rho \in Q^\bullet$, and $i \in \N$, if $\rho$ has at least $i+1$ elements, we let $\rho[i] \in Q$ denote the $(i+1)$-th element in $\rho$, $\rho[:i] \in Q^+$ denote the finite sequence $\rho_0 \cdots \rho_i \in Q^+$, and $\rho[i:] \in Q^\bullet$ denote the sequence $\rho_i \cdots \in Q^\bullet$.
%\begin{itemize}
%	\item $\rho[i] \in Q$ denote the $(i+1)$-th element in $\rho$;
%	\item $\rho[:i] \in Q^+$ denote the non-empty finite sequence $\rho_0 \cdots \rho_i \in Q^+$;
%	\item $\rho[i:] \in Q^\bullet$ denote the non-empty sequence $\rho_i \cdots \in Q^\bullet$.
%\end{itemize}

\subsection{Concurrent game structure (CGS) and Kripke structure}
We model multi-agent systems with concurrent game structures %(CGS)% that we define below.
defined below. 
\begin{definition}
	\label{def:CGS}
	A concurrent game structure (CGS for short) is a tuple $C = \langle Q,I,k,\prop, \pi, d,\delta \rangle$ where $Q$ is the finite set of states, $I \subseteq Q$ is the set of initial states, $k \in \N_1$ is the number of agents, we denote by $\Ag := \brck{1,k}$ the set of $k$ agents. Furthermore, $\prop$ is the finite set of propositions (or observations)% or what we call propositions
	, $\pi: Q\mapsto 2^{\prop}$ maps each state $q\in Q$ to the set of propositions $\pi(q) \subseteq \prop$ that hold in $q$. Finally, $d: Q \times \Ag \rightarrow \N_1$ maps each state and agent to the number of actions available to that agent at that state, and $\delta: Q_\msf{Act} \rightarrow Q$ is the function mapping every state and tuple of one action per agent to the next state, where $Q_\msf{Act} := \cup_{q \in Q} Q_\msf{Act}(q)$ with $Q_\msf{Act}(q) := \{ (q,\alpha_1,\ldots,\alpha_k) \mid \forall a \in \Ag,\; \alpha_a \in \brck{1,d(q,a)} \}$ representing the set of tuples of actions available to the players at state $q$. 
	\iffalse
	\begin{itemize}
		\item $Q$ is the finite set of states;
		\item $I \subseteq Q$ is the set of initial states;
		\item $k \in \N_1$ is the number of agents, we denote by $\Ag := \brck{1,k}$ the set of $k$ agents;
		\item $\prop$ is the finite set of observations or what we call propositions;
		\item $\pi: Q\mapsto 2^{\mathcal{P}}$ %is a function that 
		maps each state $q\in Q$ to the set of propositions $\mathcal{P}'\subseteq \mathcal{P}$ in $q$;
		%\item $\sigma: Q\mapsto \{1,\cdots k\}$ is a function that maps each state $s\in Q$ to a agent $a\in \{1,\cdots, k\}$ that is supposed to play in $s$;
		\item $d: Q \times \brck{1,k} \rightarrow \mathbb{N}^+$ maps each state and agent to the number of actions available to that agent at that state;
		\item $\delta: Q_\msf{Act} \rightarrow Q$ is the function mapping every state and tuple of one action per agent to the next state, where $Q_\msf{Act} := \cup_{q \in Q} Q_\msf{Act}(q)$ with $Q_\msf{Act}(q) := \{ (q,\alpha_1,\ldots,\alpha_i) \mid \forall a \in \Ag,\; \alpha_a \in \brck{1,d(q,a)} \}$. 
	\end{itemize}
	\fi
	
	For all $q \in Q$ and $A \subseteq \Ag$, we let $\msf{Act}_A(q) := \{ \alpha = (\alpha_a)_{a \in A} \in \prod_{a \in A} %\mid \forall a \in \Ag,\; \alpha_a \in 
	\{a\} \times \brck{1,d(q,a)} \}$. Then, for all tuple $\alpha% = (\alpha_a)_{a \in A} 
	\in \msf{Act}_A(q)$ of one action per agent in $A$, we let: 
	\begin{equation*}
	\msf{Succ}(q,\alpha) := \{ q' \in Q \mid \exists \alpha' %= (\alpha'_a)_{a \in \Ag \setminus A} 
	\in \msf{Act}_{\Ag \setminus A}(q),\; \delta(q,(\alpha,\alpha')) = q' \}
	\end{equation*}
	
	When $k = 1$, the %concurrent game structure 
	CGS $C$ is called a \emph{Kripke structure}. In that case, for all states $q \in Q$, we have the set $\msf{Succ}(q) \subseteq Q$ of the successor states of $Q$. 
	
	Finally, we define the size $|C|$ of the structure $C$ %. Then, the size $|C|$ of the structure $C$ is equal to:
	by $|C| = |Q_\msf{Act}| + |\prop| + |\Ag|$. 
\end{definition}
Unless otherwise stated, a CGS $C$ will always refer to the tuple $C = \langle Q,I,k,\prop, \pi, d,\delta \rangle$.

In a CGS, a strategy for an agent is a function that prescribes to the agent what to do as a function of the history of the game, i.e., the finite sequence of states seen so far. Moreover, given a coalition of agents and a tuple of one strategy per agent in the coalition, we define the set of infinite sequences of states that can occur with this tuple of strategies. Formally, we define this as follows.
\begin{definition}
	Consider a CGS $C$ and an agent $a \in \Ag$. A \emph{strategy} for Agent $a$ is a function $s_a: Q^+ \rightarrow \N_1$ such that, for all $\rho = \rho_0 \dots \rho_n \in Q^+$, we have $s_a(\rho) \leq d(\rho_n,a)$. We let $\msf{S}_a$ denote the set of strategies available to Agent $a% \in \brck{1,k}
	$.
	
	Given a coalition (or subset) of agents $A \subseteq \Ag$, a \emph{strategy profile} for the coalition $A$ is a tuple $s = (s_a)_{a \in A}$ of one strategy per agent in $A$. We denote by $\msf{S}_A$ the set of strategy profiles for the coalition $A$. %Given any such strategy profile 
	For all $s \in \msf{S}_A$ %, for all states 
	and $q \in Q$, we let $\msf{Out}^Q(q,s) \subseteq Q^\omega$ denote the set of infinite paths $\rho$ compatible with %the strategy profile 
	$s$ from $q$:
	\begin{equation*}
	\msf{Out}^Q(q,s) := \{ \rho \in Q^\omega \mid \rho[0] = q,\; \forall i \in \N, \rho[i+1] \in \msf{Succ}(\rho[i],(s_a(\rho[:i]))_{a \in A}) \}
	\end{equation*}
	%Then, we let:
	%\begin{equation*}
	%	\msf{Out}^\prop(q,s) := \{ \pi^\omega(\rho) \in (2^\prop)^\omega \mid \rho \in \msf{Out}^Q(q,s) \}
	%\end{equation*}
	%denote the set of infinite sequences of observations that are compatible with the strategy profile $s$ from $q$.
\end{definition}

\subsection{Alternating-time Temporal Logic}
Alternating-time Temporal Logic (ATL) is a temporal logic that takes into accounts strategic behavior of the agents. It can be seen as a generalization of Computation Tree Logic (CTL) with more than one agent. %Simiarly to Computation Tree Logic (CTL) formulas, we must introduce two types of formulas, state and path formulas. 
There are two different kinds of ATL formulas: state formulas ---  where propositions and strategic operators occur --- and path formulas -- where temporal operators occur. %Intuitively, state formulas express properties of states --- this is where  --- whereas path formulas express properties that can hold at different indices of paths --- that is where temporal operators appear. 
To avoid confusion, we denote state formulas and path formulas with Greek capital letters and Greek lowercase letters, respectively. ATL state formulas over a set of propositions $\prop$ and a set of agents $\Ag$ are given by the grammar:
\begin{equation*}
\Phi \Coloneqq p \mid \neg \Phi \mid \Phi \wedge \Phi \mid \fanBrOp{A} \varphi,
\end{equation*}
where $p\in \prop$ is a proposition, $A \subseteq \Ag$ is a coalition of agents and $\varphi$ is a path formula. We include the Boolean constants $\true$ and $\false$ and other operators such as $\Phi \lor \Phi_2$ and $\Phi_1\Rightarrow\Phi_2$.
Next, ATL path formulas are given by the grammar:
\begin{equation*}
\varphi \Coloneqq \lX \Phi \mid \Phi \lU \Phi \mid \lG \Phi
\end{equation*}
where $\lX$ is the neXt operator, $\lU$ is the Until operator, and $\lG$ is the Globally operator. As syntactic sugar, we allow standard temporal operators $\lF$, the Finally operator, which is defined in the usual manner: for any coalition of agents $A \subseteq \Ag$: $\fanBr{A} \lF\Phi\coloneqq \fanBrOp{A}(\true \lU \Phi)$. 

A CTL formula is an ATL formula on a single agent $\Ag = \{1\}$.
In particular, the path quantifiers of CTL can be obtained as follows: $\lE\ \!\equiv\fanBrOp{1}$ and $\lA\ \!\equiv\fanBrOp{}$.

The size $|\Phi|$ of an ATL formula $\Phi$ is then defined as size of the set of sub-formulas: $|\Phi| := |\msf{SubF}(\Phi)|$, which is defined inductively %We define the set of subformulas $\msf{SubF}(\Phi)$ of an ATL formula $\Phi$ 
as follows:
\begin{itemize}
	\item $\msf{SubF}(p) := \{ p \}$ for all $p \in \prop$;
	\item $\msf{SubF}(\neg \Phi) := \{ \neg \Phi \} \cup \msf{SubF}(\Phi)$;
	\item $\msf{SubF}(\Phi_1 \wedge \Phi_2) := \{ \Phi_1 \wedge \Phi_2 \} \cup \msf{SubF}(\Phi_1) \cup \msf{SubF}(\Phi_2)$;
	\item $\msf{SubF}(\fanBr{A} \bullet \Phi) := \{ \fanBr{A} \bullet \Phi \} \cup \msf{SubF}(\Phi)$ for $\bullet \in \{ \lX,\lG \}$ and $A \subseteq \Ag$;
	\item $\msf{SubF}(\fanBrOp{A} (\Phi_1 \lU \Phi_2)) := \{ \fanBrOp{A} (\Phi_1 \lU \Phi_2) \} \cup \msf{SubF}(\Phi_1) \cup \msf{SubF}(\Phi_2)$ for $A \subseteq \Ag$.
	%\item $\msf{SubF}(\fanBr{A} \lG \Phi) := \{ \fanBr{A} \lG \Phi \} \cup \msf{SubF}(\Phi)$.
\end{itemize}
%The size $|\Phi|$ of an ATL formula $\Phi$ is then defined as 
%its number of sub-formulas: $|\Phi| := |\msf{SubF}(\Phi)|$. 

We interpret ATL formulas over CGSs $C$ using the standard definitions~\cite{DBLP:journals/jacm/AlurHK02}.
Given a state $q \in Q$, we define when a state formula~$\Phi$ holds in state $q$ --- denoted by $q\models \Phi$ --- inductively as follows:
\begin{align*}
q \models p & \text{ iif } p \in \pi(q), \\
q \models \neg \Phi & \text{ iif } q \not\models \Phi, \\
q\models \Phi_1\wedge\Phi_2 & \text{ iif } q\models\Phi_1 \text{ and } q\models\Phi_2, \\
q\models \fanBrOp{A}\varphi & \text{ iif } \exists s \in \msf{S}_A,\; \forall \rho\in \msf{Out}^Q(q,s),\; \rho\models\varphi
\end{align*}
Similarly, given a path $\rho \in Q^\omega$ and a path formula $\varphi$, we define when $\varphi$ holds on path $\rho$, denoted $\rho \models \varphi$ as above, inductively as follows:
\begin{align*}
\rho\models \lX\Phi &\text{ iif } \rho[1:] \models \Phi \\
\rho\models \Phi_1 \lU \Phi_2 & \text{ iif } \exists j \in \N,\; \rho[j] \models \Phi_2 \text{, and }\forall k < j, \rho[k:]\models\Phi_1 \\
\rho \models \lG \Phi & \text{ iif } \forall j \in \N,\; \rho[j:] \models \Phi
\end{align*}
We say that an ATL formula $\Phi$ accepts (resp. rejects) a state $q$ if $q \models \Phi$ (resp. $q \not\models \Phi$). We say that it distinguishes two states $q,q'$ if it accepts one and rejects the other. Finally, we then say that the ATL formula $\Phi$ accepts a CGS $C$, denoted by $C\models \Phi$, if $\Phi$ accepts all initial states of $C$%$q\models\Phi$ for all initial states $q\in I$ of $C$. 
.

\begin{remark}
	When evaluated on turn-based game structures (i.e., where, at each state, at most one player has more than one action available), the formulas $\fanBrOp{A}\varphi$ and $\neg \fanBrOp{\Ag \setminus A}\neg \varphi$ are equivalent
	%but not on arbitrary CGSs.
	%.
	(However, it is not the case when they are evaluated on arbitrary CGSs.)
\end{remark}

\iffalse
\subsection{Observations}

The inference of ATL formulas relies on the observations of behavior of agents.
We formally represent behavior of agents as functions $O_a: Q\mapsto \{1,\cdots,d_a\}$.
For our inference, we collect observations of a set of agents $A$ in a set $O_A= \{O_a~|~ a\in A\}$.
In collect a number of such observations in a sample $\mathcal{S} = \{O^1_A, O^2_A,\cdots\}$.
\fi

%!TEX root = LearnATLarXiv.tex

\section{Passive Learning for ATL}\label{sec:passive-learning}

In this problem, we are given a sample $\sample=(P,N)$ consisting of a set $P$ of positive structures and a set $N$ of negative structures. 
The goal is to find a minimal formula $\Phi$ that is \emph{consistent} with $\sample$, i.e., $\Phi$ must hold on all positive structures and must not hold on any negative structure. 
We are specifically searching for a minimal formula and the reason for this is two-fold: (1) the prospective formula will be more interpretable, and (2) it will not overfit the sample~\cite{flie,0002FN20}. Formally:

%The size $|\sample|$ of the sample $\sample$ is equal to $|\sample| := \sum_{C \in P \cup N} |C|$.
%We are, howevr not looking for any consistent formula but for a minimal size one% formula separating the positive and negative instances
%An additional very common assumption in passive learning is that there is a bound on the size of the prospective formula. 
%We now present the formal problem statement as follows:
\begin{problem}[Passive learning of ATL]\label{prob:pass-learning}
	Given a sample $\sample=(P,N)$ consisting of two finite sets $P$ and $N$ of concurrent game structures (CGSs) with the same set $\Ag$ of agents, find a minimal size ATL formula $\Phi$ on $\Ag$ that is consistent with $\sample$.
\end{problem}
The passive learning problem of CTL can be obtained by simplifying Problem~\ref{prob:pass-learning} by use of a single agent $\Ag=\{1\}$, which reduces CGSs to KSs and ATL to CTL.
%Restricting the set of agents to one renders the passive learning problem for CTL.
%n that case, we simply 

Before describing our solution to Problem~\ref{prob:pass-learning}, we briefly discuss the source of the positive and negative structures.
Passive learning, among several applications, constitutes a critical subroutine of certain learning frameworks.
Active learning~\cite{DBLP:journals/iandc/Angluin87}, which involves learning black-box systems by interacting with a teacher, often involves repeated passive learning on the counter-example models, i.e., the feedback, received from the teacher~\cite{CamachoM19}.
Furthermore, one-class classification, or learning from positive examples, leverages passive learning to derive candidate formulas~\cite{AvellanedaP18,ltl-from-positive-only}.
These formulas facilitate the generation of negative examples, which help refine the search for more concise and descriptive formulas.

A concrete application of Problem~\ref{prob:pass-learning} would be to provide contrastive explanations~\cite{KimMSAS19} in a multi-agent setting.
Consider a multi-agent system that is deemed to have some ``good'' positions and some ``bad'' positions. 
This system would yield positive and negative CGSs corresponding to the good and bad positions.
To explain the dichotomy between these good and bad positions, one can learn an ATL formula that accepts the positive CGSs and rejects the negative ones.

\subsection{SAT-based Learning Algorithm}
\label{sec:sat-based}

Our approach to solving Problem~\ref{prob:pass-learning} is by reducing it to satisfiability problems in propositional logic. We thus provide a brief introduction to propositional logic.

\paragraph{Propositional Logic.} Let $\mathsf{Var}$ be a set of propositional variables that can be set to Boolean values from $\B = \{0,1\}$ (0 representing $\false$ and 1 representing $\true$).
Formulas in Propositional Logic are inductively constructed as follows:
\[
\propformula \Coloneqq \true \mid \false \mid x \in \mathsf{Var} \mid \lnot\propformula \mid  \propformula \lor \propformula \mid \propformula \wedge \propformula \mid \propformula \Rightarrow \propformula \mid \propformula \Leftrightarrow \propformula
\]

%As syntactic sugar, we allow the standard Boolean formulas $\true$, $\false$, $\propformula_1 \land \propformula_2$, $\propformula_1 \Rightarrow \propformula_2$, and $\propformula_1 \Leftrightarrow \propformula_2$.
%Note that t
To avoid confusion with ATL formulas% and propositional formulas
, we will be exclusively using the letter $\propformula$ (along with its variants) to denote propositional formulas. 

To assign values to propositional variables, we rely on a valuation function $v: \mathsf{Var} \to \mathbb{B}$. We exploit the valuation function $v$ to define the satisfaction $v \models \propformula$ of a propositional formula $\propformula$; we use standard definitions for this.
%The semantics of 
%We now define the semantics of 
%propositional logic using a satisfaction relation $\models$ that is defined inductively as follows: $v \models x$ if and only if
%$v(x) = 1$, $v \models \lnot\propformula$ if and only if $v \nvDash \propformula$, and $v \models \propformula_1 \lor \propformula_2$ if and only if $v \models \propformula_1$ or $v \models \propformula_2$.
When $v \models \propformula$, we say that $v$ satisfies $\propformula$ and call it a \emph{satisfying valuation} of $\propformula$.
A %propositional 
formula $\propformula$ is \emph{satisfiable} if there exists a satisfying valuation $v$ of $\Phi$.
%The \emph{size} of a formula is the number of its subformulas (as is standard).
The satisfiability (SAT) problem for propositional logic is a well-known $\mathsf{NP}$-complete problem, which asks if a propositional formula given as input is satisfiable. 
To handle SAT, numerous optimized decision procedures have been designed in recent years~\cite{MouraB08,AudemardS18,BarbosaBBKLMMMN22}.

We now describe a reduction of % solution to 
Problem~\ref{prob:pass-learning} to SAT%using satisfiability problems in propositional logic
, inspired by~\cite{flie,CamachoM19}.
Following their work, we design a propositional formula $\propformula^{\sample}_n$ that enables the search for an ATL formula of size at most $n$ that is consistent with a sample $\sample$. The formula $\propformula^{\sample}_n$ has the following properties:
\begin{enumerate}
	\item $\propformula^{\sample}_n$ is satisfiable if and only if there exists an ATL formula of size at most $n$ that is consistent with $\sample$; and
	\item from a satisfying valuation of $\propformula^{\sample}_n$, one can easily extract a suitable ATL formula of size at most $n$.
\end{enumerate}

%One can then design an iterative algorithm to search for a 

%One can then find a minimal, consistent ATL formula: 
One can then iteratively search for a minimal,  consistent formula: for increasing values of $n$, %check the satisfiability of 
if $\propformula^{\sample}_n$ is satisfiable, extract a consistent ATL formula.

%Based on the formula $\propformula^{\sample}_n$, one can design an iterative algorithm to search for a minimal, suitable ATL formula: for increasing values of $n$, check the satisfiability of $\propformula^{\sample}_n$ and if satisfiable, extract an ATL formula.

The formula $\propformula^{\sample}_n$ is defined as a conjunction of subformulas with distinct roles: $\propformula^{\sample}_n \coloneqq \propformula^{\mathsf{str}}_n \wedge \propformula^{\mathsf{sem}}_n \wedge \propformula^{\mathsf{con}}_n$.
The subformula $\propformula^{\mathsf{str}}$ encodes the structure of the prospective ATL formula $\Phi$.
The subformula $\propformula^{\mathsf{sem}}$ encodes that the correct semantics of ATL is used to interpret the prospective ATL formula on the given CGSs. Finally, the subformula $\propformula^{\mathsf{con}}$ ensures that the prospective ATL formula holds on the models in $P$ and not in the models in $N$.
The formula used for CTL learning has an identical high-level structure, with similar subformulas.

We now describe the subformulas of $\propformula^{\sample}_n$ in detail.

\paragraph{Encoding the structure of ATL formulas.}
The structure of an ATL formula is symbolically encoded as a \emph{syntax DAG}. A syntax DAG of an ATL formula is simply a syntax tree
in which the common nodes are merged. Figure~\ref{fig:syntax-dag} depicts an example of a syntax tree and %syntax 
DAG of an ATL formula.

\begin{figure}
	\centering
	\scalebox{0.9}{
	\centering
	%	\begin{subfigure}[b]{0.3\textwidth}
	%		\centering
	%		\begin{tikzpicture}
	%		\node (1) at (0, 0) {$\lor$};
	%		\node (2) at (.7, -0.9) {$\fanBr{1}\lU$};
	%		\node (3) at (-.7, -0.9) {$\fanBr{2}\lX$};
	%		\node (4) at (0, -1.8) {$p$};
	%		\node (5) at (1.4, -1.8) {$\fanBr{1,\!3}\lG$};
	%		\node (6) at (1.4, -2.7) {$q$};
	%		\node (7) at (-.7, -1.8) {$p$};
	%		\draw[->] (1) -- (2); 
	%		\draw[->] (1) -- (3);
	%		\draw[->] (2) -- (4);
	%		\draw[->] (2) -- (5);
	%		\draw[->] (3) -- (7);
	%		\draw[->] (5) -- (6);
	%		\end{tikzpicture}
	%		\caption{Syntax tree}
	%	\end{subfigure}
	%	\begin{subfigure}[b]{0.3\textwidth}
	\centering
	\begin{tikzpicture}
	\node (1) at (0, 0) {$\lor$};
	\node (2) at (.7, -0.9) {$\fanBr{1}\lU$};
	\node (3) at (-.7, -0.9) {$\fanBr{2}\lX$};
	\node (4) at (0, -1.8) {$p$};
	\node (5) at (1.4, -1.8) {$\fanBr{1,\!3}\lG$};
	\node (6) at (1.4, -2.7) {$q$};
	\draw[->] (1) -- (2); 
	\draw[->] (1) -- (3);
	\draw[->] (2) -- (4);
	\draw[->] (2) -- (5);
	\draw[->] (3) -- (4);
	\draw[->] (5) -- (6);
	\node[font=\scriptsize] at (0.2,0.2) {6};
	\node[font=\scriptsize] at (1.23,-.65) {4};
	\node[font=\scriptsize] at (-.17,-.67) {5};
	\node[font=\scriptsize] at (0,-1.5)  {3};
	\node[font=\scriptsize] at (2.1,-1.6) {2};
	\node[font=\scriptsize] at (1.6,-2.5) {1};
	\end{tikzpicture}}
	%	\end{subfigure}
	%	\begin{subfigure}[b]{0.3\textwidth}
	%		\centering
	%		\begin{tikzpicture}
	%		\node (1) at (0, 0) {6};
	%		\node (2) at (.7, -0.9) {4};
	%		\node (3) at (-.7, -0.9) {5};
	%		\node (4) at (0, -1.8) {3};
	%		\node (5) at (1.4, -1.8) {2};
	%		\node (6) at (1.4, -2.7) {1};
	%		\draw[->] (1) -- (2); 
	%		\draw[->] (1) -- (3);
	%		\draw[->] (2) -- (4);
	%		\draw[->] (2) -- (5);
	%		\draw[->] (3) -- (4);
	%		\draw[->] (5) -- (6);
	%		\end{tikzpicture}
	%		\caption{Identifiers}
	%	\end{subfigure}
	\caption{Syntax DAG with identifiers (indicated above nodes) of $\fanBr{2}\lX p \lor \fanBrOp{1}(p \lU \fanBr{1,\!3}\lG q )$.}
	\label{fig:syntax-dag}
\end{figure}

To conveniently encode the syntax DAG of an ATL formula, we first fix a naming convention for its nodes. For a formula of size at most $n$, we assign to each of its nodes an identifier in $\brck{1,n}$ such that the identifier of each node is larger than those of its children, if applicable. Note that such a naming convention may not be unique.
 %Based on these identifiers, 
We then denote the sub-formula of $\Phi$ rooted at Node~$i$ as $\Phi[i]$. Thus, $\Phi[n]$ denotes the entire formula $\Phi$.

Next, to encode a syntax DAG symbolically, we introduce the following propositional variables: (i) $x_{i,\lambda}$ for $i\in\brck{1,n}$ and $\lambda\in \prop\cup \Lambda$ with $\Lambda := \{\neg,\wedge,\fanBr{\cdot}\lX,\fanBr{\cdot}\lG,\fanBr{\cdot}\lU \}$; (ii) $A_{i,a}$ for $i\in\brck{1,n}$ and $a\in\Ag$; and (iii) $l_{i,j}$ and $r_{i,j}$ for $i\in\brck{1,n}$ and $j\in\brck{1,i-1}$. The variable $x_{i,\lambda}$ tracks the operator labeled in Node~$i$, meaning, $x_{i,\lambda}$ is set to true if and only if Node~$i$ is labeled with $\lambda$. 
The variable $A_{i,a}$ is relevant only if $x_{i,\fanBrOp{\cdot}\bullet}$ is set to true for some temporal operator $\bullet \in \{\lX,\lG,\lU\}$. In such a case, the variables $(A_{i,a})_{a \in \Ag}$ track which agents are in the coalition $\fanBrOp{\cdot}$ at Node $i$. 
The variable $l_{i,j}$ (resp., $r_{i,j}$) tracks the left (resp., right) child of Node~$i$, meaning, $l_{i,j}$ (resp., $r_{i,j}$) is set to true if and only if the left (resp., right) child of Node~$i$ is Node~$j$. 
%For unary operators (in $\Lambda_{\mathsf{Un}}$), we only consider the left child.%Not necessary constraint

%Finally, note that the variable $\mathsf{rel}_{i}$ is set to true whenever there is no $j < i$ such that $x_{j,\bot}$ is true. Informally, the indices $i$ for which $x_{i,\bot}$ is true  are irrelevant in the sense that the node is not part of the ATL. The reason for adding $\bot$ comes from the fact that we do not know a priori how many $\fanBr{\cdot}$ operators, of size $|\Ag|$, will be used. Hence, since we want to only consider formulas of size (at most) $n$, potentially some nodes should not appear in the formula, as to make it too large. We discuss this further below.

To ensure that these variables encode a valid syntax DAG, we impose structural constraints similar to the ones proposed by Neider and Gavran~\cite{flie}.
For instance, the constraint below ensures that each node is labeled with a unique operator:
\[\Big[ \bigwedge_{i\in\brck{1,n}} \bigvee_{\lambda \in \Lambda} x_{i,\lambda} \Big] \land \Big[\bigwedge_{i\in\brck{1,n}} \bigwedge_{\lambda \neq \lambda' \in \Lambda} \lnot x_{i,\lambda} \lor \lnot x_{i,\lambda'}  \Big]\]
We impose similar constraints to ensure that each node has a unique left and right child.
We detail all the structural constraints in Appendix~\ref{appen:structural_cond}.
The formula $\Omega^{\mathsf{str}}_n$ is simply the conjunction of all such structural constraints.

Based on a satisfying valuation $v$ of $\Omega^{\mathsf{str}}_n$, one can construct a unique ATL formula: label each Node~$i$ with the operator $\lambda$ for which $v(x_{i,\lambda})=1$, include players in a coalition for which $v(A_{i,a})=1$, and mark the left (resp., right) child with Nodes $j$ (resp., $j'$) for which $v(l_{i,j})=1$ (resp., $v(r_{i,j'})=1$).

For encoding the structure of (the existential fragment of) CTL, the only difference from that of ATL is that the set of operators is $\Lambda\coloneqq \{\neg,\wedge,\lE\lX,\lE\lG,\lE\lU \}$.

%These constraints are expressed with the logical formulas $\Omega_\Lambda$ (for the variables $x_{i,\lambda}$ with $\lambda \in \Lambda$), $\Omega_l$ (for the variables $l_{i,j}$), $\Omega_r$ (for the variables $r_{i,j}$), and $\Omega_\prop$ (for the variables $x_{1,p}$ with $p \in \prop$). Due to a lack of space, their definitions, along with a brief explanation of how to construct the unique ATL formula that a satisfying valuation of the variables induce,  postponed to Appendix~\ref{appen:structural_cond}. 

%The formula $\Omega^{\mathsf{str}}$ is then simply defined as the conjunction of all these formulas: $\Omega^{str}_n := \Omega_\Lambda \wedge \Omega_l \wedge \Omega_r \wedge \Omega_\prop$.

\paragraph{Encoding the semantics of ATL formulas.}
To symbolically encode the semantics of the prospective ATL formula $\Phi$ for a given CGS $C$, we rely on encoding the ATL model-checking procedure developed in~\cite[Section 4]{DBLP:journals/jacm/AlurHK02}. The procedure involves calculating, for each sub-formula $\Phi'$ of $\Phi$, the set $\Sat(\Phi') = \{q\in Q \mid q\models \Phi'\}$ of the states of $C$ where $\Phi'$ holds. Since we consider ATL formulas, we need to handle the strategic operators $\fanBrOp{A}$ for $A \subseteq \Ag$. To do so, given a coalition of agents $A \subseteq \Ag$ and some subset of states $S \subseteq Q$, we let $\mathsf{Pre}_A(S) \subseteq Q$ denote the set of states from which the coalition $A$ has a strategy to enforce reaching the set $S$ in one step. That is:
$\mathsf{Pre}_A(S) := \{ q \in Q \mid \exists \alpha \in \mathsf{Act}_A(q),\; \mathsf{Succ}(q,\alpha) \subseteq S \}$.

%The computation of $\Sat(\Phi')$ is done recursively on the structure of $\Phi$ based on fixed-point computations.
%For a given concurrent game structure $C$, t
%
We can now describe how to compute the set $\Sat(\Phi)$. It is done inductively on the structure of the ATL formula $\Phi$ as follows:% for the different possible ATL formulas are stated as follows:
\begin{align}
&\Sat(p) = \{ q \in Q~|~p \in \pi(q) \}, \text{ for any }p\in \prop,\\
&\Sat(\Phi \wedge \Psi) = \Sat(\Phi) \cap \Sat(\Psi),\\
&\Sat(\neg\Phi) = Q \setminus \Sat(\Phi),\\
&\Sat(\fanBr{A}\lX\Phi) = \mathsf{Pre}_A(\Sat(\Phi)),\\
&\Sat(\fanBrOp{A}(\Phi_1 \lU \Phi_2)) \text{ is the smallest }T \subseteq Q, \text{ such that }\nonumber\\
&\hspace{1cm}(1)\ \Sat(\Phi_2) \subseteq T\text{ and }
(2)\ \Sat(\Phi_1)\cap\mathsf{Pre}_A(T)\subseteq T
%\ q \in \Sat(\Phi)\text{ and }q \in \mathsf{Pre}_A(T)\text{ implies }q \in T
,\label{eq:sat-EU}\\
&\Sat(\fanBr{A}\lG\Phi)\text{ is the largest }T \subseteq Q,\text{ such that }\nonumber\\
&\hspace{1cm}(1)\ T \subseteq \Sat(\Phi)\text{ and }(2)\ %\ q \in T\text{ implies }q \in \mathsf{Pre}_A(T)\label{eq:sat-EG}
T\subseteq\mathsf{Pre}_A(T)\label{eq:sat-EG}
\end{align}

Our goal is to symbolically encode the above-described computation. To do so, %intuitively, our approach is to 
we introduce the propositional variables $y^{C}_{i,q}$ for each $i\in\brck{1,n}$, $q\in Q$ that track whether a state $q \in Q$ %of a concurrent game structure $C$ 
belongs to $\Sat(\Phi[i])$ for a sub-formula $\Phi[i]$ of $\Phi$%Formally, we introduce propositional variables $y^{C}_{i,q}$ for each $i\in\brck{1,n}$, $q\in \{1,\ldots,|Q|\}$% and $C\in P\cup N$
. That is, 
%The variable $y^{C}_{i,q}$ tracks whether the sub-formula $\Phi[i]$
%of $\Phi$ holds in state $q$, meaning, 
$y^{C}_{i,q}$ is set to true if and only if $\Phi[i]$ holds in state $q$, i.e. $q\in\Sat(\Phi[i])$.

Before defining the propositional formulas that ensure the desired meaning of the variables $y^{C}_{i,q}$, we introduce formulas to keep track of whether a state belongs to the set $\mathsf{Pre}_A(S)$. Formally, for all $q \in Q$ and $i \in \brck{1,n}$, given any predicate $y_S: Q \Rightarrow \B$ (encoding a set $S = (y_S)^{-1}[true] \subseteq Q$), we define the formula $\Omega_{q,i}^{\mathsf{Pre}}(y_S)$
%we have used the formula $\Omega_{q,i}^{\mathsf{Pre}}(y)$ %$\mathsf{Pre}_i(y)$ 
%that we have not yet defined. It is 
%defined for any predicate $y: Q \Rightarrow \B$. It
that encodes the fact that $q \in \mathsf{Pre}_{A_i}(S)$, %i.e. The fact that $q \in \mathsf{Pre}_i(y)$ means 
where $A_i$ is the coalition of agents %$A_i$, d
defined by the variables $(A_{i,a})_{a \in \Ag}$%, has a strategy to ensure, from $q$, to reach in one step a state satisfying the predicate $y$
. This formula is defined as follows: %we express the fact that $q \in \mathsf{Pre}_i(y)$ below:
\begin{equation*}
	\Omega_{q,i}^{\mathsf{Pre},C}(y_S) := \bigvee_{\alpha \in \mathsf{Act}_\Ag(q)} \bigwedge_{\alpha' \in \mathsf{Act}_\Ag(q)} \big[(\bigwedge_{a \in \Ag} A_{i,a} \Rightarrow (\alpha_a = \alpha'_a)) \Rightarrow y_S(\delta(q,\alpha')) \big]
\end{equation*}

This formula can be informally read as follows: there exists an action tuple for the coalition $A$ (the disjunction), such that for all action tuples for the opposing coalition (the conjunction), the corresponding state is in the set $S$ (indicated by $y_S(\delta(q,\alpha'))$). Since we do not know, a priori, what the coalition $A_i$ is, %, since it is defined by propositional variable, 
we quantify over action tuples for all the agents both in the disjunction and the conjunction. However, the rightmost implication of the formula ensures that the only relevant tuples of actions $\alpha' \in \mathsf{Act}_\Ag(q)$ are those for which the action for the agents in the coalition $A_i$ are given by the tuple of actions $\alpha \in \mathsf{Act}_\Ag(q)$.

\iffalse
Informally, the disjunction quantifies over action tuples at $q$ where the only relevant actions are those of the coalition $A$. %(Note that the action taken by any agent not in the coalition is irrelevant in the formula.)
Then, the  conjunction quantifies over action tuples at $q$ that stands for the response of the agents not in the coalition. The implication can be read as follows: if the later tuple of actions coincide with the former on the actions chosen by the coalition%(i.e. if it indeed corresponds to the actions chosen by the coalition corresponds to a response)
, then %it should be that 
the corresponding state $\delta(q,\alpha')$ satisfies the predicate $y$, i.e. that $y_{\delta(q,\alpha')}$ holds.
%The second conjunction checks that only those tuple of actions for which what coalition plays match what is prescribed by $\alpha$ is considered. Finally, we check that indeed the variable $y_i^C$ is satisfied in the corresponding successor state.
\fi

We define the formulas ensuring the intended meaning of the variables $y_{i,q}$%, we consider the following formulas:
.
\begin{align*}
\Omega^{\mathsf{sem}}_{\prop,C} & := 
\bigwedge_{p \in \prop} \bigwedge_{i \in \brck{1,n}} \Big[ x_{i,p} \Rightarrow \bigwedge_{q\in Q,\; p\in \pi(q)} y^{C}_{i,q} \; \wedge \bigwedge_{q\in Q,\; p\notin \pi(q)} \neg y^{C}_{i,q} \Big]%\label{eq:prop}
\\
\Omega^{\mathsf{sem}}_{\wedge,C} & := 
\bigwedge_{\substack{{i \in \brck{1,n}}\\{j,j' \in \brck{1,i-1}}}} \Big[ [x_{i,\wedge}\wedge l_{i,j}\wedge r_{i,j'}] \Rightarrow \bigwedge_{q\in Q}\big[ y^{C}_{i,q} \Leftrightarrow [y^{C}_{j,q} \wedge y^{C}_{j',q}]\big]\Big]%\label{eq:and}
\\
\Omega^{\mathsf{sem}}_{\neg,C} & := 
\bigwedge_{\substack{{i \in \brck{1,n}}\\{j \in \brck{1,i-1}}}} \Big[ [x_{i,\neg}\wedge l_{i,j}] \Rightarrow \bigwedge_{q\in Q}\big[ y^{C}_{i,q} \Leftrightarrow \neg y^{C}_{j,q}\big]\Big]%\label{eq:not}
\\
\Omega^{\mathsf{sem}}_{\lX,C} & := 
\bigwedge_{\substack{{i \in \brck{1,n}}\\{j \in \brck{1,i-1}}}} \Big[ [x_{i,\fanBr{\cdot}\lX}\wedge l_{i,j}] \Rightarrow \bigwedge_{q\in Q}\big[ y^{C}_{i,q} \Leftrightarrow \Omega_{q,i}^{\mathsf{Pre},C}(y^C_{j,\cdot}) %q \in \mathsf{Pre}_i(y^C_j) 
\big]\Big]%\label{eq:EX}
\end{align*}
The above formulas encode, via a straightforward translation, the $\Sat$ computation for the propositions, Boolean operators and the $\fanBr{\cdot} \lX$ operator. 

The case of the $\fanBr{\cdot}\lU$ and $\fanBr{\cdot}\lG$ operators require some innovation.
Indeed, as can be seen in Equations~(\ref{eq:sat-EU}) and~(\ref{eq:sat-EG}), the $\Sat$ involves a 
%This is because they require the computation of the 
least and greatest fixed-point, respectively%, as can be seen from in Equations~(\ref{eq:sat-EU}) and~(\ref{eq:sat-EG})
.
To circumvent this difficulty, we mimic the steps of the fixed-point computation algorithm~\cite[Fig. 3]{DBLP:journals/jacm/AlurHK02} in propositional logic. Let us recall how they are computed. %Fix a concurrent game structure $C$ and a coalition of agents $A \subseteq \Ag$. 
Given an ATL formula $\Phi = \fanBrOp{A} \Phi_1 \lU \Phi_2$, the way $\Sat(\Phi)$ is computed from $\Sat(\Phi_1)$ and $\Sat(\Phi_2)$ is described in Algorithm~\ref{alg:algo_U}. Similarly, given an ATL formula $\Phi = \fanBr{A} \lG \Phi'$, the way $\Sat(\Phi)$ is computed from $\Sat(\Phi')$ is described in Algorithm~\ref{alg:algo_G}. Note that, 
instead of using \emph{while} loops, as in ~\cite[Fig. 3]{DBLP:journals/jacm/AlurHK02}, that are necessarily exited after at most $|Q|$ steps, we use \emph{for} loops. %This is possible thanks to the fact that the \emph{while} loops used is exited after at most $|Q|$ steps. Furthermore, this allows us to translate directly into our $\SAT$ encoding the computation of the \emph{for} loops.%we use both times a \emph{for} loop instead of \emph{while} loop to be closer to our $\SAT$ encoding.%,  in both cases, the \emph{while} loops are exited after at most $|Q|$ steps. 

\begin{algorithm}[t]
	\caption{Compute $\Sat(\Phi)$ for $\Phi = \fanBr{A} \; \Phi_1 \lU \Phi_2$
	a}\label{alg:algo_U}
	\textbf{Input}: CGS $C$, coalition $A$, $\Sat(\Phi_1)$ and $\Sat(\Phi_2)$
	\begin{algorithmic}[1]
		\State $S\coloneqq\Sat(\Phi_2)$
		\For{$1 \leq k \leq |Q|$}
		%\STATE 
		%\STATE $\mathsf{change}\coloneqq true$
		%\WHILE{$\mathsf{change}$}
		%\STATE $\mathsf{change} := false$
		%\STATE $\mathsf{New} \coloneqq \{ q \in \Sat(\Phi_1) \setminus S \mid q \in \mathsf{Pre}_A(S) \}$
		\State $S \leftarrow S \cup \{ q \in \Sat(\Phi_1) \cap \mathsf{Pre}_A(S) \}$
		%\IF{$\mathsf{New} \neq \emptyset$}
		%\STATE $\mathsf{change} := true$
		%\ENDIF
		\EndFor
		\Return $S$
	\end{algorithmic}
\end{algorithm}
\begin{algorithm}[t]
	\caption{Compute $\Sat(\Phi)$ for $\Phi = \fanBr{A} \; \lG \Phi'$}\label{alg:algo_G}
	\textbf{Input}: CGS $C$, coalition $A$ and $\Sat(\Phi')$
	\begin{algorithmic}[1]
		\State $S \coloneqq \Sat(\Phi')$
		%\STATE $\mathsf{change}\coloneqq true$
		%\WHILE{$\mathsf{change}$}
		\For{$1 \leq k \leq |Q|$}
		%\STATE $\mathsf{change} := false$
		%\STATE $\mathsf{Old} \coloneqq \{ q \in S \mid q \notin \mathsf{Pre}_A(S) \}$
		\State $S \leftarrow S \cap \mathsf{Pre}_A(S)$
		%\IF{$\mathsf{Old} \neq \emptyset$}
		%\STATE $\mathsf{change} := true$
		%\ENDIF
		\EndFor
		\Return $S$
	\end{algorithmic}
\end{algorithm}

As can be seen in both algorithms% for the operators $\fanBr{\cdot}\lU$ and $\fanBr{\cdot}\lG$
, the fixed-point computation algorithm internally maintains an estimate of the $\Sat$ set and updates it iteratively. Thus, to encode the fixed-point computation, we introduce propositional variables that encode whether a state $q$ of a CGS $C$ belongs to a particular estimate of $\Sat$. Formally, we introduce propositional variables $y^{C}_{i,q,k}$ for each $i\in\brck{1,n}$, $q\in A$, and $k\in\brck{0,|Q|}$%, and $C\in P\cup N$ 
%that are responsible for encoding encode estimates of $\Sat$
, where the parameter $k\in\brck{0,|Q|}$ %in the variable $y^{C}_{i,q,k}$ 
tracks which iterative step of the fixed-point computation the variable $y^{C}_{i,q,k}$ encodes. %Also, as the fixed-point algorithms need at most $|Q|+1$ iterative steps, $l$ ranges from $1$ to $|Q|+1$.
We define the formulas below %on 
to ensure the intended meaning of these introduced variables:
\begin{align*}
\Omega^{\mathsf{sem}}_{\lU,C} &:= 
\bigwedge_{\substack{{i \in \brck{1,n}}\\{j,j' \in \brck{1,n}}}} [x_{i,\fanBr{\cdot}\lU}\wedge l_{i,j}\wedge r_{i,j'}]  \Rightarrow 
\bigwedge_{q\in Q}\Big[
\big[y^{C}_{i,q,0} \Leftrightarrow y^{C}_{j',q}\big] \wedge \\
&\hspace{9mm}\bigwedge_{0\leq k\leq |Q|-1}
\big[
y^{C}_{i,q,k+1} \Leftrightarrow [y^{C}_{i,q,k} \vee [y^{C}_{j,q} \wedge \Omega_{q,i}^{\mathsf{Pre}}(y^C_{i,\cdot,k}) ]]
\big] \wedge \; [y^{C}_{i,q} \Leftrightarrow y^{C}_{i,q,|Q|}] \Big] \\ 
\Omega^{\mathsf{sem}}_{\lG,C} &:= 
\bigwedge_{\substack{{i \in \brck{1,n}}\\{j \in \brck{1,i-1}}}}  [x_{i,\fanBr{\cdot}\lG}\wedge l_{i,j}] \Rightarrow
\bigwedge_{q\in Q}\Big[
\big[y^{C}_{i,q,0} \Leftrightarrow y^{C}_{j,q}\big] \wedge \\ 
&\hspace{2cm}\bigwedge_{0\leq k\leq |Q|-1}
\big[
y^{C}_{i,q,k+1} \Leftrightarrow [y^{C}_{i,q,k} \wedge \Omega_{q,i}^{\mathsf{Pre}}(y^C_{i,\cdot,k}) %q \in \mathsf{Pre}_i(y^C_{i,\cdot,k})
] \big] \wedge \; [y^{C}_{i,q} \Leftrightarrow y^{C}_{i,q,|Q|}] \Big]
\end{align*}
\iffalse
with
\begin{align*}
\Omega^{\mathsf{sem}}_{\lU,C,i,j,j'} := \bigwedge_{q\in Q}\Big[
\big[y^{C}_{i,q,0} \Leftrightarrow y^{C}_{j',q}\big] & \wedge 
\bigwedge_{0\leq k\leq |Q|-1}
\big[
y^{C}_{i,q,k+1} \Leftrightarrow [y^{C}_{i,q,k} \vee [y^{C}_{j,q} \wedge \Omega_{q,i}^{\mathsf{Pre}}(y^C_{i,\cdot,k})% q \in \mathsf{Pre}_i(y^C_{i,\cdot,k})
]]
\big] %\wedge [y^{C}_{i,q} \Leftrightarrow y^{C}_{i,q,|Q|+1}]
\\
&\wedge \; [y^{C}_{i,q} \Leftrightarrow y^{C}_{i,q,|Q|}] \Big]%\label{eq:EU}\nonumber
\\
\Omega^{\mathsf{sem}}_{\lG,C,i,j} := \bigwedge_{q\in Q}\Big[
\big[y^{C}_{i,q,0} \Leftrightarrow y^{C}_{j,q}\big] & \wedge 
\bigwedge_{0\leq k\leq |Q|-1}
\big[
y^{C}_{i,q,k+1} \Leftrightarrow [y^{C}_{i,q,k} \wedge \Omega_{q,i}^{\mathsf{Pre}}(y^C_{i,\cdot,k}) %q \in \mathsf{Pre}_i(y^C_{i,\cdot,k})
] \big] \\
& \wedge \; [y^{C}_{i,q} \Leftrightarrow y^{C}_{i,q,|Q|}] \Big]%\label{eq:EG}\nonumber
\end{align*}
\fi 
The formula $\Omega^{\mathsf{sem}}_n$ is simply defined as the conjunction of all the %semantics constraints defined 
formulas above. 
%$\Omega^{\mathsf{sem}}_{n} := \wedge_{C \in P \cup N} \Omega^{\mathsf{sem}}_{n,C}$ where, for all $C \in P \cup N$,  $\Omega^{\mathsf{sem}}_{n,C} := \Omega^{\mathsf{sem}}_{\prop,C} \wedge \Omega^{\mathsf{sem}}_{\wedge,C} \wedge  \Omega^{\mathsf{sem}}_{\neg,C} \wedge \Omega^{\mathsf{sem}}_{\lX} \wedge \Omega^{\mathsf{sem}}_{\lU,C} \wedge \Omega^{\mathsf{sem}}_{\lG,C}$.

For the semantics of CTL, the main difference is that the formula $\Omega_{q,i}^{\mathsf{Pre}}(y_S)$ encoding the fact that $q\in\mathsf{Pre}(S)$ can be greatly simplified (for the quantifier $\lE$\ \!):
\begin{equation*}
 \hat{\Omega}_{q,i}^{\mathsf{Pre}}(y_S) := \bigvee_{q' \in \mathsf{Succ}(q)} y_S(q')
\end{equation*}

\paragraph{Encoding the consistency with the models}
Finally, to encode that the prospective formula is consistent with $\sample$, we have the following formula:
\begin{equation*}
	\Omega^{\mathsf{con}}_n := \big[\bigwedge_{C\in P}\bigwedge_{s\in I} y^{C}_{n,s}\big] \wedge \big[\bigwedge_{C\in N}\bigvee_{s\in I} \neg y^{C}_{n,s}\big]
\end{equation*}
%which encodes that the prospective formula must hold in all of the initial states for the models in $P$ and must not hold in some initial state for the models in $N$.

The size of the formula $\Omega^{\sample}_n$, and the number of variables involved in it, is polynomial in $n$ and the size of $\sample$, $|\sample|\coloneqq\sum_{C \in P \cup N} |C|$. To express this formally, let us first introduce notations for various quantities that will occur when expressing the size of the different formulas involved.
\begin{definition}
	Consider a sample $S = (P,N)$ of CGS. We let: $x := |\Lambda| = 5 + |\prop|$, $k := |\Ag|$, $r = \max_{q \in Q} |Q_{\mathsf{Act}}(q)|$, for all $C \in P \cup N$, $m_C := |Q_C|$ for $Q_C$ the of states of the structure $C$ and $m := |Q|$ for $Q := \cup_{C \in P \cup N} Q_C$. 
\end{definition} 

%In terms of number of variables, we have:
%\begin{itemize}
%	\item $|X_n^{\mathsf{Var}}| = O(n \cdot x + n \cdot k + n^2)$
%	\item For all $C \in P \cup N$, $|X_C^{\mathsf{Struct}}| = O(n \cdot m^2)$
%\end{itemize}

Now, by looking how the different formulas are defined, we obtain the following (asymptotic) bounds:
\begin{itemize}
	\item $|\Omega_{\Lambda}| = O(n \cdot x^2)$;
	\item $|\Omega_{l}| = O(n^3)$;
	\item $|\Omega_{r}| = O(n^3)$;
	\item $|\Omega_{\prop}| = O(x)$.
\end{itemize}
Hence, we obtain $|\Omega_{n}^{str}| = O(n^3 + n^2 \cdot x)$. In addition, for all $C \in P \cup N$, we have:
\begin{itemize}
	\item $|\Omega_{\prop,C}^{\mathsf{sem}}| = O(x \cdot n \cdot m_C)$;
	\item $|\Omega_{\wedge,C}^{\mathsf{sem}}| = O(n^3 \cdot m_C)$;
	\item $|\Omega_{\neg,C}^{\mathsf{sem}}| = O(n^2 \cdot m_C)$;
	\item For all states $q \in Q$ and $i \leq n$, we have $|\Omega_{q,i}^{\mathsf{Pre},C}(y)| = O(r^2 \cdot k)$;
	\item $|\Omega_{\lX,C}^{\mathsf{sem}}| = O(n^2 \cdot m_C \cdot r^2 \cdot k)$;
	\item $|\Omega_{\lU,C}^{\mathsf{sem}}| = O(n^3 \cdot m_C^2 \cdot r^2 \cdot k)$;
	\item $|\Omega_{\lG,C}^{\mathsf{sem}}| = O(n^2 \cdot m_C^2 \cdot r^2 \cdot k)$.
\end{itemize}
Overall, we obtain $|\Omega_{n}^{\mathsf{sem}}| = O(n^3 \cdot m^2 \cdot r^2 \cdot k + x \cdot n \cdot m)$. Finally, we have $|\Omega_{n}^{\mathsf{con}}| = O(m)$. We do obtain that the size of $\Omega_n^{\mathsf{S}}$ is (asymptotically) polynomial in $|S|$ and $n$.

Furthermore, we have the lemma below establishing the correctness of the encoding. A detailed proof of this proposition can be found in Appendix~\ref{subsec:appen_correctness}.
\begin{proposition}
	\label{prop:pass-learning-correctness}
	Let $\sample=(P,N)$ be a sample and $n \in \mathbb{N}\setminus \{0\}$
	%and $\Omega^{\sample}_n$ be the propositional formula from Equation~(\ref{eqn:formula})
	. Then:
	\begin{enumerate}
		\item If an ATL formula of size at most $n$ consistent with $\sample$ exists, then the propositional formula $\Omega^{\sample}_n$ is satisfiable.
		\item If a valuation $v$ is such that $v\models\propformula^{\sample}_n$, then there is an ATL formula $\Phi^v$ of size at most $n$ that is consistent with $\sample$.
	\end{enumerate}
\end{proposition}

 %It essentially consists in formally stating and proving that the various formulas that we have defined do encode the properties that they are supposed to.
%
%Additionally, note that 

\subsection{Deciding the separability}
\label{subsec:deciding_sep}
Given a sample $\sample$, by iteratively checking if $\Omega^{\sample}_0$ is satisfiable, if $\Omega^{\sample}_1$ is satisfiable, etc., we can find a minimal size formula consistent with $\sample$ if one exists. However, %it may be that 
if there is no such formula, %in which case 
the above iteration would not terminate. To circumvent that issue, there are two possibilities. We may first decide the separability of the sample, i.e. decide if there exists an ATL formula consistent with it; or we may exhibit a bound $B$, expressed as a function of %the size of the CGS in 
$\sample$, such that if there is a separating formula, there is one of size at most $B$.
%
%The way Problem~\ref{prob:pass-learning} is stated 
%
%With the encoding into a $\SAT$ formula that we have described above, we can describe an algorithm that finds, given a sample $\sample$ of positive and negative concurrent game structures, a smallest formula consistent with that sample. Indeed, it suffices to start from $n = 0$, and then increase $n$ until finding a formula $\SAT$ formula $\Omega^{\sample}_n$ that is satisfiable. However, as such, this algorithm may not terminate, if there is no separating formula. There are two ways to fix that:
In this subsection, we tackle both of these issues. 

It was shown in \cite[Section 3, Thm. 3.2, 3.9]{DBLP:journals/tcs/BrowneCG88} that the separability %of the sample 
can be decided in polynomial time for full CTL (i.e., all operators can be used) with Kripke structures. % if (at least) the operators $\neg,\wedge$ and $\lX$
Furthermore, in \cite[Coro. 1]{DBLP:journals/corr/abs-2402-06366}, %the authors use 
the results of \cite{DBLP:journals/tcs/BrowneCG88} are used to exhibit an exponential bound on the size of the CTL formulas to be considered. 

For fragments of CTL (i.e. CTL formulas that can use only some operators), it was shown in \cite[Thm. 3]{DBLP:journals/pacmpl/KrogmeierM23}, as a corollary of a \textquotedblleft meta theorem\textquotedblright{} with applications to various logic, that the separability %of the sample 
can be decided in exponential time. %This result was obtained as a corollary of a very general result (a \textquotedblleft meta theorem\textquotedblright{}) with applications for various logic, including CTL. This result holds for all fragments of CTL.

Here, we extend these results to ATL formulas. We consider two settings: full ATL and any fragment of ATL. In the first setting, %we extend these results to ATL and concurrent game structures. That is, 
we show that separability can also be decided in polynomial time. %We also show that, for all fragments of ATL, 
In the second setting, we show that separability can be decided in exponential time. We deduce an exponential bound on the size of the formulas that need to be considered, which hold regardless of the fragment considered (including full ATL).

\subsubsection{Separability for full ATL.}
%We consider the case of ATL formula where the operators $\neg, \wedge$ and $\fanBr{\cdot} \lX$ are allowed. 
For full ATL, our goal is to show the theorem below.
\begin{theorem}
	\label{thm:decide_polynomial_time}
	Given a sample $\sample = (P,N)$ of CGS, we can decide in time polynomial in $|\sample|$ if the sample $\sample$ is separable with (arbitrary) ATL formulas.
\end{theorem}

What we are actually going to prove is Lemma~\ref{lem:separate_ATL_wedge}, that we state below, establishing the result of Theorem~\ref{thm:decide_polynomial_time} for all fragments of ATL with the operators $\neg,\wedge$ and $\lX$. Let us first introduce formally these fragments.
\begin{definition}
	We let $\mathsf{OP} := \{ \neg,\wedge, \lX,\lG,\lU\}$ denote the set of operators that can be used in (regular) ATL formulas. 
	
	For all subsets of operators $\mathsf{O} \subseteq \mathsf{Op}$, we let $\mathbf{ATL}(\mathsf{O})$ denote the set of ATL formulas using only operators in $\mathsf{O}$. When $\mathsf{O} = \mathsf{OP}$, $\mathbf{ATL}(\mathsf{O})$ is simply denoted $\mathbf{ATL}$. 
	
	We let $\mathsf{OPn} := \{\neg,\wedge, \lX\}$ and we refer to formula in the fragment $\mathbf{ATL}(\mathsf{OPn})$ as ATL-$\lX$ formulas.
\end{definition}
For simplicity, in the following, we allow ATL-$\lX$ formulas to use the classical operator $\lor$ since $\Phi_1 \lor \Phi_2 \equiv \neg (\neg \Phi_1 \wedge \neg \Phi_2)$ for all ATL formulas $\Phi_1$ and $\Phi_2$.

We want to establish of this section the lemma below.
\begin{lemma}
	\label{lem:separate_ATL_wedge}
	Consider any subset of operators $\mathsf{OPn} \subseteq \mathsf{O} \subseteq \mathsf{OP}$, and a sample $S = (P,N)$ of CGS. We can decide in time $O(n^6 \cdot r^6 + n^4 \cdot k + n^2 \cdot p)$ if the sample $S$ is separable with $\mathbf{ATL}(\mathsf{O})$ formulas% that can use the  operators $\neg, \wedge,\fanBr{\cdot} \lX$
	.
\end{lemma}
The proof of this lemma uses ideas similar to what is presented \cite{DBLP:journals/tcs/BrowneCG88} with CTL formulas. %Due to space constraint, we only provide an informal explanation here. %What we present here is  except it is extended to ATL. 
%The idea is to explain informally how to prove this theorem on the case of CTL formulas and Kripke structures. First of all, let us introduce a notations for the fragment of CTL with the operators $\neg, \wedge,\lor,\lX$ (we have added the operator $\lor$ since it can obtained the operators $\wedge$ and $\neg$). 
%
Before we proceed to the proof of this lemma, we introduce some notations below.
\begin{definition}
	Consider a sample $S = (P,N)$ of CGS. For all CGS $C \in P \cup N$, we let $Q_C$ (resp. $I_C \subseteq Q_C$) denote the set of (resp. initial) states of the structure $C$. When the sample $S$ is clear from context, we let $Q := \cup_{C \in P \cup N} Q_C$ denote the set of all the states occurring in $S$. We also let $\Ag$ denote the set of agents occurring in all CGS of $S$ and we let $\prop$ denote the set of all the propositions occurring in at least one CGS on $S$. %When the sample $S$ is clear from context, we omit the $_S$ in the notations introduced above.
	
	Whenever a function in $\pi,d,\delta,\mathsf{Act},\mathsf{Succ}$ is applied to a state $q \in Q$, we always assume that it is the function from the CGS $C$ such that $q \in  Q_C$.
\end{definition} 

We also introduce below the notations on samples that we will use in the following to write complexity bounds. 
\begin{definition}
	Consider a sample $S = (P,N)$ of CGS. For all states $q \in Q$, we let:
	\begin{equation*}
		||q|| := |Q_{\mathsf{Act}}(q)|
	\end{equation*}
	Note that, for all CGS $C \in P \cup N$, we have $|(Q_C)_{\mathsf{Act}}| = \sum_{q \in Q_C} ||q||$. 
	
	We let $r := \max_{q \in Q} ||q||$, $M := |Q_{\mathsf{Act}}|$, $n := |Q|$, $p := |\prop|$ and $k := |\Ag|$.
	%
	%Note that $r +n + p + k = O(|S|)$.
\end{definition}

Now, the first step that we take is to prove that it is sufficient to consider only ATL-$\lX$ formulas.%, %i.e., ATL formulas whose only used temporal operator is $\lX$. 
\iffalse
\begin{definition}
	We let CTLn denote the fragment of CTL consisting of formulas that only use the operators $\neg,\wedge,\lor,\lX$.
\end{definition}
\fi
%
This may seem counterintuitive since given an ATL formula $\Phi$ using the operator $\lG$ (or the operator $\lU$), there does not exist ATL-$\lX$ formula equivalent to $\Phi$ because the number of states of the CGS on which $\Phi$ may be evaluated is arbitrarily large. However, given a sample $\sample$ of finitely many CGS, there is a bound on the number of states used in all the CGS of $\sample$. Hence, there is an ATL-$\lX$ formula equivalent to $\Phi$ on all the CGS of $\sample$. %The same holds for the operator $\lU$. %Hence, we have the lemma below.
%because the horizon of an ATL-$\lX$ formula is bounded, while it is not the case for the formulae 
%
%However, given an ATL formula $\Phi$ that may use the operators $\lG$ or $\lU$, there may not exist an ATL-$\lX$ formula that is equivalent to $\Phi$. Nonetheless, for all $n \in \N$, there does exist an ATL-$\lX$ formula equivalent to $\Phi$ on CGS with at most $n$ states. Hence, we have the lemma below, which justifies restricting ourselves to ATL-$\lX$ formulas.
\begin{lemma}
	\label{lem:ATLn_ok}
	Consider a sample $\sample = (P,N)$ of CGS and any subset of operators $\mathsf{OPn} \subseteq \mathsf{O} \subseteq \mathsf{OP}$. If there exists an $\mathbf{ATL}(\mathsf{O})$ formula consistent with $\sample$, then there is one that is an ATL-$\lX$ formula.
\end{lemma}
Note that this lemma is known for a long time in the context of CTL formulas and Kripke structures, see \cite[Theorem 3.2, 3.9]{DBLP:journals/tcs/BrowneCG88}. Furthermore, the proof of this lemma is not complicated, but quite cumbersome as it requires to handle all the operators $\lG,\lU$,etc, it is therefore postponed to Appendix~\ref{proof:lem_ATLn_ok}.

Let us now consider the set $\mathsf{Distinguish}(\sample) \subseteq Q^2$ of pairs of states that we can distinguish with an ATL-$\lX$ formula.
\begin{definition}
	Consider a sample $\sample$ of CGS. We let $Q$ denote the set of all the states occurring in all CGS of $\sample$. We let: 
	\begin{equation*}
	\mathsf{Distinguish}(\sample)\coloneqq\{ (q,q') \in Q^2~|~\text{there is ATL-$\lX$ formula } \Phi \text{ s.t. } q\models \Phi \Leftrightarrow q' \not\models \Phi \}
	\end{equation*}
\end{definition}

Given a sample $\sample$, we claim the following: 1) it is possible to compute in time polynomial in $|\sample|$ the set $\mathsf{Distinguish}(\sample)$, and 2) given $\mathsf{Distinguish}(\sample) \subseteq Q^2$, we can decide in polynomial time if there is an ATL-$\lX$ formula consistent with $\sample$.

The second claim is actually a straightforward consequence of the lemma below. 
%Interestingly, from the set $\mathsf{Distinguish}(S)$, we can decide if there is an ATL-$\lX$ formula consistent with a sample $S$.
\begin{lemma}
	\label{lem:decide_from_distinguish}
	Consider a sample $S = (P,N)$ of CGS. There is an ATL-$\lX$ formula consistent with $S$ if and only if for all $C_N \in N$, there some state $q_{C_N} \in I_{C_N}$ such that, for all $C_P \in P$ and states $q_{C_P} \in I_{C_P}$, we have $(q_{C_P},q_{C_N}) \in \mathsf{Distinguish}(S)$. 
\end{lemma}
\begin{proof}
	Assume that there is an ATL-$\lX$ formula $\Phi$ consistent with $S$. Then, for all $C_N \in N$, there is some state $q_{C_N} \in I_{C_N}$ such that $q_{C_N} \not\models \Phi$. Furthermore, for all $C_P \in P$ and states $q_{C_P} \in I_{C_P}$, we have $q_{C_P} \models \Phi$. Hence, $(q_{C_P},q_{C_N}) \in \mathsf{Distinguish}(S)$.
	
	On the other hand, for all $C_N \in N$, $C_P \in P$ and states $q_{C_P} \in I_{C_P}$, we let $\Phi_{q_{C_P},q_{C_N}}$ be an ATL formula that accepts $q_{C_P}$ and rejects $q_{C_N}$. Then, we let $\Phi := \lor_{C_P \in P} \lor_{q_{C_P} \in I_{C_P}} (\wedge_{C_N \in N} \Phi_{q_{C_P},q_{C_N}})$ be an ATL-$\lX$ formula. By construction, %Let us show that it 
	this ATL-$\lX$ formula $\Phi$ is consistent with $S$.%· Let $C_P \in P$ and $q_{C_P} \in I_{C_P}$. For all $C_N \in N$, we have 
\end{proof}
%The reason why it holds is the following: there is an ATL-$\lX$ formula consistent with $\sample$ iff, for all negative structures $C_N$ of $\sample$, there is a starting state $q_N \in I_{C_N}$ such that, for all starting states $q_P$ of all positive structures of $\sample$, we have $(q_P,q_N) \in \mathsf{Distinguish}(\sample)$, which can be checked in polynomial time. The \textquotedblleft only if\textquotedblright{} implication comes directly from the definition of a formula consistent with a sample. The \textquotedblleft if\textquotedblright{} implication is a consequence of the fact that ATL-$\lX$ formulas can use conjunctions, disjunctions, and negations. %Furthermore, the latter condition can indeed be checked in time polynomial in $|\S|$.

Let us now consider the first claim. Our goal is to compute the set $\mathsf{Distinguish}(S)$ in polynomial tine. To do so, we are going to %compute this set $\mathsf{Distinguish}(S)$ by 
iteratively add to the empty set pairs of states that can be distinguished by ATL-$\lX$ formulas. 

Before we define what sets of states we add at each state, let us define the notion of relevant agents, given any pair of states.
\begin{definition}
	Consider a sample $S = (P,N)$ of CGS. For all pairs of states $(q,q') \in Q^2$, we let $\mathsf{RelAg}(q,q') \subseteq \Ag$ be equal to:
	\begin{equation*}
		\mathsf{RelAg}(q,q') := \{ a \in \Ag \mid d(q,a) \cdot d(q',a) \geq 2 \}
	\end{equation*}
	Note that, for convenience, for the empty coalition $\emptyset$, we let $\mathsf{Act}_\emptyset(q) := \{ \epsilon \}$, for all states $q \in Q$, where $\epsilon$ is an empty tuple of actions.
\end{definition}

The way we add states at every step is now defined below. 
\begin{definition}
	\label{def:update_X}
	Consider a sample $S = (P,N)$ of CGS. %concurrent game structure $C = \langle Q,I,k,\prop, \pi, \sigma, d,\delta \rangle$. 
	Consider some subset $R \subseteq Q^2$. \iffalse We let $\mathsf{Upd}_{\neg}(S) \subseteq Q^2$ be equal to:
	\begin{equation*}
		\mathsf{Upd}_{\neg}(S) := \{ (q',q) \in Q^2 \setminus S \mid (q,q') \in S \}
	\end{equation*}
	\fi
	We let $\mathsf{Upd}(R) \subseteq Q^2$ be equal to:
	\begin{align*}
		\mathsf{Upd}(R) := \{ & (q,q'),(q',q) \in Q^2 \mid \exists A \subseteq \mathsf{RelAg}(q,q'),\; \exists \alpha \in \mathsf{Act}_{A}(q),\; \\ & \forall \alpha' \in \mathsf{Act}_{A}(q'), 
		\exists t' \in \mathsf{Succ}(q',\alpha'), \forall t \in \mathsf{Succ}(q,\alpha), \; (t,t') \in R \}
	\end{align*}
\end{definition}

%bound could be improved to n(n-1) intead of n^2
We can now iteratively compute the set $\mathsf{Distinguish}(S)$ with approximations until we reach a fixed point. Note that this is similar to what is done \cite[Secion 3]{DBLP:journals/tcs/BrowneCG88} with CTL formulas and Kripke structures, except that we additionally have to handle the coalitions of agents. 
\begin{lemma}
	\label{lem:charac_distinguish}
	Consider a sample $S = (P,N)$ of CGS. %concurrent game structure $C = \langle Q,I,k,\prop, \pi, \sigma, d,\delta \rangle$. 
	We let $n := |Q|^2$. We define $(R_i)_{0 \leq i \leq n^2} \subseteq (Q^2)^{n^2+1}$ as follows:
	\begin{itemize}
		\item $R_0 := \{ (q,q'),(q',q) \in Q^2 \mid \pi(q) \neq \pi(q') \}$;
		\item for all $0 \leq i \leq n^2-1$, $R_{i+1} := R_i \cup \mathsf{Upd}(R_i)$. 
	\end{itemize}
	We let $R := R_{n^2}$. Then, we have:
	\begin{equation*}
		R = \mathsf{Distinguish}(S)
	\end{equation*}
\end{lemma}
\begin{proof}
	First of all, let us show that $\mathsf{Upd}(R) \subseteq R$. Clearly, for all $0 \leq i \leq n^2-1$, we have $R_i \subseteq R_{i+1}$. Therefore, since $|Q^2| = n^2$, if $R_{n^2} \neq Q^2$, then there is some $0 \leq k \leq n^2-1$ such that $R_k = R_{k+1}$ and $\mathsf{Upd}(R_k) \subseteq R_k$. In such a case, for all $k \leq i \leq n^2$, we have $R_i = R_k$ and $\mathsf{Upd}(R_i) \subseteq R_i$. Hence, we do have $\mathsf{Upd}(R) \subseteq R$. 
	
	By definition, for all $(q,q') \in \mathsf{Distinguish}(S)$, we also have $(q',q) \in \mathsf{Distinguish}(S)$. This is also the case for $R$. Indeed, for all $(q,q') \in R_0$, $(q',q) \in R_0$. Furthermore, for all subsets $T \subseteq Q^2$, we have that for all $(q,q') \in \mathsf{Upd}(T)$, $(q',q) \in \mathsf{Upd}(T)$. 
	
	For all pairs $(t,t') \in \mathsf{Distinguish}(S)$, we let $\Phi_{(t,t')}$ be an ATL-$\lX$ formula such that $t \models \Phi_{(t,t')}$ and $t' \not \models \Phi_{(t,t')}$, which exists since we can use the negation in ATL-$\lX$ formulas. 
	
	Now, let us show that $R \subseteq \mathsf{Distinguish}(S)$. We show by induction on $0 \leq i \leq n^2$ that $R_i \subseteq \mathsf{Distinguish}(S)$. Consider any $(q,q') \in R_0$. Consider a proposition $p \in \prop$ such that $p \in \pi(q)$ iff $p \notin \pi(q')$, which exists by definition of $R_0$. Then, the ATL-$\lX$ formula $\Phi := p$ distinguishes the states $q$ and $q'$. Hence, $R_0 \subseteq \mathsf{Distinguish}(S)$. 
	
	Assume now that $R_i \subseteq \mathsf{Distinguish}(S)$ for some $0 \leq i \leq n^2-1$. Consider some $(q,q') \in \mathsf{Upd}(R_i)$ for which there is a coalition of agents $A \subseteq \mathsf{RelAg}(q,q')$ and a tuple of actions $\alpha \in \mathsf{Act}_{A}(q)$ such that for all tuples of actions $\alpha' \in \mathsf{Act}_{A}(q')$, there exists $t' \in \mathsf{Succ}(q',\alpha')$ such that for all $t \in \mathsf{Succ}(q,\alpha)$, we have $(t,t') \in R_i$. For all tuples of actions $\alpha' \in \mathsf{Act}_{A}(q')$, we let $t_{\alpha'} \in \mathsf{Succ}(q',\alpha')$ be such that for all $t \in \mathsf{Succ}(q,\alpha)$, we have $(t,t') \in R_i \subseteq \mathsf{Distinguish}(S)$. For all tuples of actions $\alpha' \in \mathsf{Act}_{A}(q')$, we let:
	\begin{equation*}
		\Phi_{\alpha'} := \bigvee_{t \in \mathsf{Succ}(q,\alpha)} \Phi_{(t,t_{\alpha'})}
	\end{equation*}
	
	We also let:
	\begin{equation*}
		\Phi_{\alpha} := \bigwedge_{\alpha' \in \mathsf{Act}_{A}(q')} \Phi_{\alpha'}
	\end{equation*}
	Finally, we let $\Phi := \fanBr{A} \lX \Phi_{\alpha}$. Note that $\Phi$ is an ATL-$\lX$ formula since, for all $(t,t') \in \mathsf{Distinguish}(S)$, $\Phi_{(t,t')}$ is an ATL-$\lX$ formula. We claim that $q \models \Phi$ and $q' \not\models \Phi$. Indeed:
	\begin{itemize}
		\item Let $t \in \mathsf{Succ}(q,\alpha)$. Consider any $\alpha' \in \mathsf{Act}_{A}(q')$. We have $t \models \Phi_{(t,t_{\alpha'})}$. Hence, $t \models \Phi_{\alpha'}$. Since this holds for all $\alpha' \in \mathsf{Act}_{A}(q')$, it follows that $t \models \Phi_{\alpha}$. This holds for all $t \in \mathsf{Succ}(q,\alpha)$. %Hence, $q \in \mathsf{Pre}_A(\mathsf{SAT}_C(\Phi_{\alpha}))$. 
		Thus, $q \models \Phi$. 
		\item Consider any $\alpha' \in \mathsf{Act}_{A}(q')$. For all $t \in \mathsf{Succ}(q,\alpha)$, we have $t_{\alpha'} \not\models \Phi_{(t,t_{\alpha'})}$. Therefore, $t_{\alpha'} \not\models \Phi_{\alpha'}$. Thus, $t_{\alpha'} \not\models \Phi_{\alpha}$. That is, for all $\alpha' \in \mathsf{Act}_{A}(q')$, there is some $t_{\alpha'} \in \mathsf{Succ}(q',\alpha')$ such that $t_{\alpha'} \not\models \Phi_{\alpha}$. %Hence, $q' \notin \mathsf{Pre}_A(\mathsf{SAT}_C(\Phi_{\alpha}))$. 
		%This holds for all $\alpha' \in \mathsf{Act}_{A}(q')$. 
		Thus, $q' \not\models \Phi$. 
	\end{itemize}
	Therefore, $(q,q'),(q',q) \in \mathsf{Distinguish}(S)$. It follows that $R_{i+1} \subseteq \mathsf{Distinguish}(S)$. In fact, $R_i \subseteq \mathsf{Distinguish}(S)$ for all $0 \leq i \leq n^2$. Hence, $R \subseteq \mathsf{Distinguish}(S)$.
	
	Let us now show that $\mathsf{Distinguish}(S) \subseteq R$. We show by induction on ATL-$\lX$ formulas $\Phi$ the property $\mathcal{P}(\Phi)$: for all $(q,q') \in Q^2$, if $q \models \Phi$ and $q' \not\models \Phi$, then $(q,q') \in R$. 
	
	Consider an ATL-$\lX$ formula $\Phi = p \in \prop$. Let $(q,q') \in Q^2$, such that $q \models \Phi$ and $q' \not\models \Phi$. Then, $p \in \pi(q)$ and $p \notin \pi(q')$. Therefore, $\pi(q) \neq \pi(q')$. Thus, $(q,q') \in R_0 \subseteq S$. Hence, $\mathcal{P}(\Phi)$ holds. Then:
	\begin{itemize}
		\item Consider an ATL-$\lX$ formula $\Phi = \neg \Phi'$. Assume that $\mathcal{P}(\Phi')$ holds. Consider any $(q,q') \in Q^2$, such that $q \models \Phi$ and $q' \not\models \Phi$. Then, $q' \models \Phi'$ and $q \not\models \Phi'$. By $\mathcal{P}(\Phi')$, it follows that $(q',q) \in S$, and therefore we also have $(q,q') \in S$. Hence, $\mathcal{P}(\Phi)$ holds.
		\item Consider an ATL-$\lX$ formula $\Phi = \Phi_1 \wedge \Phi_2$. Assume that $\mathcal{P}(\Phi_1)$ and $\mathcal{P}(\Phi_2)$ hold. Consider any $(q,q') \in Q^2$, such that $q \models \Phi$ and $q' \not\models \Phi$. Then, let $\Phi' \in \{ \Phi_1,\Phi_2\}$ be such that $q' \not\models \Phi'$. We have $q \models \Phi'$. Hence, by $\mathcal{P}(\Phi')$, we have $(q,q') \in S$. Hence, $\mathcal{P}(\Phi)$ holds.
		\item Consider an ATL-$\lX$ formula $\Phi = \Phi_1 \lor \Phi_2$. Assume that $\mathcal{P}(\Phi_1)$ and $\mathcal{P}(\Phi_2)$ hold. Consider any $(q,q') \in Q^2$, such that $q \models \Phi$ and $q' \not\models \Phi$. Then, let $\Phi' \in \{ \Phi_1,\Phi_2\}$ be such that $q \models \Phi'$. We have $q' \not\models \Phi'$. Hence, by $\mathcal{P}(\Phi')$, we have $(q,q') \in S$. Hence, $\mathcal{P}(\Phi)$ holds.
		\item Consider an ATL-$\lX$ formula $\Phi = \fanBr{A} \lX \Phi'$ for a coalition of agents $A \subseteq \Ag$. Assume that $\mathcal{P}(\Phi')$ holds. Consider any $(q,q') \in Q^2$, such that $q \models \Phi$ and $q' \not\models \Phi$. Then:
		\begin{itemize}
			\item there is a tuple of actions $\alpha \in \mathsf{Act}_{A}(q)$ such that, for all $t \in \mathsf{Succ}(q,\alpha)$, we have $t \models \Phi'$;
			\item for all tuples of actions $\alpha' \in \mathsf{Act}_{A}(q')$, there is some $t_{\alpha'} \in \mathsf{Succ}(q',\alpha')$, such that $t_{\alpha'} \not\models \Phi'$.
		\end{itemize}
		We let $\bar{A} := A \cap (\mathsf{RelAg}(q,q'))$ and for all $u \in \{ q,q'\}$, for all tuples of actions $\alpha \in \mathsf{Act}_{A}(u)$, we let $\bar{\alpha} \in \mathsf{Act}_{\bar{A}}(u)$ be such that, for all $a \in \bar{A}$, we have $\bar{\alpha}_a = \alpha_a$. Then, by definition of the set $\mathsf{RelAg}(q,q')$, we have $\mathsf{Succ}(q,\alpha) = \mathsf{Succ}(q,\bar{\alpha})$. Hence, the two above points are equivalent to:
		\begin{itemize}
			\item there is a tuple of actions $\alpha \in \mathsf{Act}_{\bar{A}}(q)$ such that, for all $t \in \mathsf{Succ}(q,\alpha)$, we have $t \models \Phi'$;
			\item for all tuples of actions $\alpha' \in \mathsf{Act}_{\bar{A}}(q')$, there is some $t_{\alpha'} \in \mathsf{Succ}(q',\alpha')$, such that $t_{\alpha'} \not\models \Phi'$.
		\end{itemize}
		Therefore, by $\mathcal{P}(\Phi')$, for all $t \in \mathsf{Succ}(q,\alpha)$ and $\alpha' \in \mathsf{Act}_{A}(q')$, we have $(t,t_{\alpha'}) \in R$. Overall, we have: for the coalition of agents $\bar{A} \subseteq \mathsf{RelAg}(q,q')$ and for the tuple of actions $\alpha \in \mathsf{Act}_{\bar{A}}(q)$, for all tuples of actions $\alpha' \in \mathsf{Act}_{\bar{A}}(q')$, the state $t_{\alpha'} \in \mathsf{Succ}(q',\alpha')$ is such that, for all $t \in \mathsf{Succ}(q,\alpha)$, we have $(t,t_{\alpha'}) \in R$. Therefore, $(q,q') \in \mathsf{Upd}(R) \subseteq R$. Hence, $\mathcal{P}(\Phi)$ holds.
	\end{itemize}
	Overall, $\mathcal{P}(\Phi)$ holds for ATL-$\lX$ formulas $\Phi$. Therefore, $\mathsf{Distinguish}(S) \subseteq R$.
\end{proof}
There now only remains to note that, given a set $R$, we can compute in polynomial time the set $\mathsf{Dist}_{\lX}(R)$. It then follows, from Lemma~\ref{lem:charac_distinguish} that we can compute in polynomial time the set $\mathsf{Distinguish}(\sample)$. We provide explicit algorithms in Appendix~\ref{subsec:algorithms}, with explicit bounds. This allows to establish Lemma~\ref{lem:separate_ATL_wedge}, from which directly follows Theorem~\ref{thm:decide_polynomial_time}.

\subsubsection{Separability for any fragment of ATL}
Let us now consider the case of an arbitrary fragment of ATL. Now, the above-described algorithm a priori does not work. Among other things, one of the issues is that, possibly, % since it may be that 
we may not be able to use the operator $\lX$%, Lemma~\ref{lem:defi_distinguish} may not hold% (since the items 2. and 3. of Lemma~\ref{lem:CNS_distinguish} correspond to the use of the operator $\lX$)
. %To circumvent this issue
Therefore, given a sample $\sample$, instead of iteratively constructing a subset of $\mathsf{Distinguish}(\sample) \subseteq Q^2$ of distinguishable pairs of states, we iteratively compute a subset $\mathsf{Acc}(\sample) \subseteq 2^Q$ of sets of states that 
can be accepted with a formula in the fragment considered, while the complement set is rejected. We can then prove a definition and lemma similar to Definition~\ref{def:update_X} and Lemma~\ref{lem:charac_distinguish} except that we additionally have to be able to handle the operators $\lG$ and $\lU$, %, which is done by defining, given any $R \subseteq 2^Q$, the sets $\mathsf{Dist}_{\lU}(R)$ and $\mathsf{Dist}_{\lG}(R)$. Note that these sets are defined using 
which is done by using Algorithms~\ref{alg:algo_U},~\ref{alg:algo_G}. 

The goal of subsection is to show the theorem below.
\begin{theorem}
	\label{thm:bound_size_formulas}
	Consider a sample $\sample% = (P,N)
	$ %of CGS 
	and a fragment ATL' of ATL. We can decide in time exponential in $|\sample|$ if the sample $\sample$ is separable with ATL' formulas. %Furthermore, i
	If it is, then there exists an ATL' formula consistent with $\sample$ %, then there is one 
	of size at most $2^{|Q|}$. 
\end{theorem}

To establish this theorem, we are actually to prove the lemma below.
\begin{lemma}
	\label{lem:decide_exponential_time_ATLn}
	Consider any subset of operators $\mathsf{O} \subseteq \mathsf{OP}$. and a sample $S = (P,N)$ of CGS. We can decide in time $O(2^{k+2n} \cdot n^2 \cdot M \cdot r)$ if there is an $\mathbf{ATL}(\mathsf{O})$-formula consistent with $S$. Furthermore, if it is the case, then there is an $\mathbf{ATL}(\mathsf{O})$-formula of size at most $2^n$ that is consistent with $S$.
\end{lemma}

The complexity result extends \cite[Theorem 3]{DBLP:journals/pacmpl/KrogmeierM23} to the case of ATL formulas and CGS. %There is also a slight improvement complexity wise, since 
For the bound on the size of the formulas, this extends \cite[Corollary 1]{DBLP:journals/corr/abs-2402-06366} to any fragment of ATL, although it is not clear whether our bound and the one from \cite[Corollary 1]{DBLP:journals/corr/abs-2402-06366} are comparable or not.

First of all, for any subset of operators $\mathsf{O} \subseteq \mathsf{OP}$, let us define formally the set $\mathsf{Acc}(S,\mathsf{O}) \subseteq 2^Q$ of accepted sets of states by $\mathbf{ATL}(\mathsf{O})$-formulas.
\begin{definition}
	Given any sample $S = (P,N)$ of CGS and ATL formula $\Phi$, we let $\SAT(\Phi) \subseteq Q$ denote the set of states that satisfy the formula $\Phi$: $\SAT(\Phi) := \{ q \in Q \mid q \models \Phi \}$. Then, for any subset of operators $\mathsf{O} \subseteq \mathsf{OP}$, we define the set $\mathsf{Acc}(S,\mathsf{O}) \subseteq 2^Q$ as follows:
	\begin{equation*}
		\mathsf{Acc}(S,\mathsf{O}) := \{ \SAT(\Phi) \mid \Phi \in \mathbf{ATL}(\mathsf{O}) \} \subseteq 2^Q
	\end{equation*}
\end{definition}

Our goal is to able to compute the set $\mathsf{Acc}(S,\mathsf{O})$, To do so, as in the previous subsection, we are going to iteratively compute an approximation of that set until we reach a fixed point. We define below the update function that we will consider. 
\begin{definition}
	Consider a sample $S = (P,N)$ of CGS and a coalition of agents $A \subseteq \Ag$. For all subsets $T \subseteq Q$, we let:
	\begin{itemize}
		\item $\mathsf{Upd}_{\neg}(S,T) := Q \setminus T$;
		\item $\mathsf{Upd}_{\lX}(S,A,T) := \mathsf{Pre}_A(T)%\{ q \in Q \mid \exists \alpha \in \mathsf{Act}_A(q),\; \mathsf{Succ}(q,\alpha) \subseteq T \}
		$;
		\item $\mathsf{Upd}_{\lG}(S,A,T) := \mathsf{ComputeSAT}_{\lG}(S,A,T)$ which refers to the output of Algorithm~\ref{alg:algo_G} on $(S,A,T)$;
	\end{itemize}
	
	For all pairs of subsets $T,T' \subseteq Q$, we let:
	\begin{itemize}
		\item $\mathsf{Upd}_{\wedge}(S,T,T') := T \cap T'$;
		\item $\mathsf{Upd}_{\lU}(S,A,T,T') := \mathsf{ComputeSAT}_{\lU}(S,A,T,T')$, which  refers to the output of Algorithm~\ref{alg:algo_U} on $(S,A,T,T')$;
	\end{itemize}
	
	Then, for any subset of operators $\mathsf{O} \subseteq \mathsf{OP}$ and for all $R \subseteq 2^Q$, we let $\mathsf{Upd}_{\mathsf{O}}(R) \subseteq 2^Q$ be equal to:
	\begin{equation*}
		\mathsf{Upd}_{\mathsf{O}}(R) := \bigcup_{A \subseteq \Ag} \left(\bigcup_{T \in R} \bigcup_{o \in \mathsf{O} \cap \{ \neg,\lX,\lG \}} \mathsf{Upd}_{o}(S,A,T) \cup \bigcup_{T,T' \in R} \bigcup_{o \in \mathsf{O} \cap \{ \wedge,\lU \}} \mathsf{Upd}_{o}(S,A,T,T') \right)
	\end{equation*}
\end{definition}

%Let us show that 
As intended, this definition matches the semantics of the operators $\neg,\wedge,\lX,\lG,\lU$. 
\begin{lemma}
	\label{lem:proper_behavior}
	Consider a sample $S = (P,N)$ of CGS and any ATL formula $\Phi$. 
	\begin{itemize}
		\item If $\Phi = \neg \Phi'$, then $\SAT(\Phi) = \mathsf{Upd}_{\neg}(S,\SAT(\Phi'))$.
		\item If $\Phi = \fanBr{A} \lX \Phi'$, then $\SAT(\Phi) = \mathsf{Upd}_{\lX}(S,A,\SAT(\Phi'))$.
		\item If $\Phi = \fanBr{A} \lG \Phi'$, then $\SAT(\Phi) = \mathsf{Upd}_{\lG}(S,A,\SAT(\Phi'))$.
		\item If $\Phi = \Phi_1 \wedge \Phi_2$, then $\SAT(\Phi) = \mathsf{Upd}_{\wedge}(S,\SAT(\Phi_1),\SAT(\Phi_2))$.
		\item If $\Phi = \fanBr{A} \Phi_1 \lU \Phi_2$, then $\SAT(\Phi) = \mathsf{Upd}_{\lU}(S,A,\SAT(\Phi_1,\SAT(\Phi_2))$.
	\end{itemize}
\end{lemma}
\begin{proof}
	The case of the operators $\neg,\wedge$ and $\lX$ comes directly from the definition of the semantics of these operators. The case of the operators $\lG$ and $\lU$ comes from the fixed point computation algorithm~\cite[Fig. 3]{DBLP:journals/jacm/AlurHK02} (as mentioned in the main part of the paper).
\end{proof}

We can now define the iterative approximations of the set $\mathsf{Acc}(S,\mathsf{O})$.
\begin{definition}
	\label{def_approx_acc}
	Consider a sample $S = (P,N)$ of CGS. We let $n := |Q|$. For any subset of operators $\mathsf{O} \subseteq \mathsf{OP}$, we define $(\mathsf{Apr}_i(S,\mathsf{O}))_{1 \leq i \leq 2^n} \subseteq (2^Q)^{2^n}$ as follows:
	\begin{itemize}
		\item $\mathsf{Apr}_1(S,\mathsf{O}) := \cup_{b \in \prop}\{ q \in Q \mid b \in \pi(q) \} = \cup_{b \in \prop} \SAT(b)$;
		\item for all $1 \leq i \leq 2^n-1$, $\mathsf{Apr}_{i+1}(S,\mathsf{O}) := \mathsf{Apr}_i(S,\mathsf{O}) \cup \mathsf{Upd}_{\mathsf{O}}(\mathsf{Apr}_i(S,\mathsf{O}))$. 
	\end{itemize}
	We let $\mathsf{Apr}(S,\mathsf{O}) := \mathsf{Apr}_{2^n}(S,\mathsf{O})$. 
\end{definition}

Our goal is now to show that $\mathsf{Apr}(S,\mathsf{O}) = \mathsf{Acc}(S,\mathsf{O})$. This will be done in two different lemmas. We start with the easy inclusion.
\begin{lemma}
	\label{lem:acc1}
	Given any sample $S = (P,N)$ of CGS and any subset of operators $\mathsf{O} \subseteq \mathsf{OP}$, we have: 
	\begin{equation*}
		\mathsf{Acc}(S,\mathsf{O}) \subseteq \mathsf{Apr}(S,\mathsf{O})
	\end{equation*}
\end{lemma}
We proceed similarly than for the proof of Lemma~\ref{lem:charac_distinguish}.
\begin{proof}
	Let us first show that $\mathsf{Upd}_{\mathsf{O}}(\mathsf{Apr}(S,\mathsf{O})) \subseteq \mathsf{Apr}(S,\mathsf{O})$. Clearly, for all $1 \leq i \leq 2^n-1$, we have $\mathsf{Apr}_i(S,\mathsf{O})\subseteq \mathsf{Apr}_{i+1}(S,\mathsf{O})$. Therefore, since $|2^Q| = 2^n$, if $\mathsf{Apr}_{2^n}(S,\mathsf{O}) \neq 2^Q$, then there is some $1 \leq k \leq 2^n-1$ such that $\mathsf{Apr}_k(S,\mathsf{O}) = \mathsf{Apr}_{k+1}(S,\mathsf{O})$ and $\mathsf{Upd}_{\mathsf{O}}(\mathsf{Apr}_k(S,\mathsf{O})) \subseteq \mathsf{Apr}_k(S,\mathsf{O})$. In such a case, for all $k \leq i \leq 2^n$, we have $\mathsf{Apr}_i(S,\mathsf{O}) = \mathsf{Apr}_k(S,\mathsf{O})$ and $\mathsf{Upd}_{\mathsf{O}}(\mathsf{Apr}_i(S,\mathsf{O})) \subseteq \mathsf{Apr}_i(S,\mathsf{O})$. Hence, we do have $\mathsf{Upd}_{\mathsf{O}}(\mathsf{Apr}(S,\mathsf{O})) \subseteq \mathsf{Apr}(S,\mathsf{O})$.
	
	We now show by induction on $\mathbf{ATL}(\mathsf{O})$-formulas $\Phi$ the property $\mathcal{P}(\Phi)$: $\SAT(\Phi) \in \mathsf{Apr}(S,\mathsf{O})$. 
	
	Consider an $\mathbf{ATL}(\mathsf{O})$-formula $\Phi \in \prop$. By definition, we have $\SAT(\Phi) \in \mathsf{Apr}_1(S,\mathsf{O}) \subseteq \mathsf{Apr}(S,\mathsf{O})$. Hence, $\mathcal{P}(\Phi)$ holds. 
	
	Consider any unary operator $u \in \{ \neg,\lX,\lG\}$. If $u \in \mathsf{O}$, then for any $\mathbf{ATL}(\mathsf{O})$-formula $\Phi = (\fanBr{A}) \; u \; \Phi'$ (where $\fanBr{A}$ appears, for some coalition of agents $A \subseteq \Ag$, iff $u \neq \neg$), if $\SAT(\Phi') \in \mathsf{Apr}(S,\mathsf{O})$, then by Lemma~\ref{lem:proper_behavior}, we have $\SAT(\Phi) \in \mathsf{Upd}_{\mathsf{O}}(\mathsf{Apr}(S,\mathsf{O})) \subseteq \mathsf{Apr}(S,\mathsf{O})$. That is, $\mathcal{P}(\Phi')$ implies $\mathcal{P}(\Phi)$.
	
	Similarly, consider any binary operator $b \in \{ \wedge,\lU\}$. If $b \in \mathsf{O}$, then for any $\mathbf{ATL}(\mathsf{O})$-formula $\Phi = (\fanBr{A}) \Phi_1 \; b \; \Phi_2$ (where $\fanBr{A}$ appears, for some coalition of agents $A \subseteq \Ag$, iff $b \neq \wedge$), if $\SAT(\Phi_1),\SAT(\Phi_2) \in \mathsf{Apr}(S,\mathsf{O})$, then by Lemma~\ref{lem:proper_behavior}, we have $\SAT(\Phi) \in \mathsf{Upd}_{\mathsf{O}}(\mathsf{Apr}(S,\mathsf{O})) \subseteq \mathsf{Apr}(S,\mathsf{O})$. That is, $\mathcal{P}(\Phi_1)$ and $\mathcal{P}(\Phi_2)$ imply $\mathcal{P}(\Phi)$.
	
	Therefore, $\mathcal{P}(\Phi)$ holds for all $\mathbf{ATL}(\mathsf{O})$-formulas $\Phi$. Therefore, $\mathsf{Acc}(S,\mathsf{O}) \subseteq \mathsf{Apr}(S,\mathsf{O})$.
\end{proof}

We now state and prove a lemma establishing the reverse inclusion. We additionally establish in this lemma a bound on the size of the $\mathbf{ATL}(\mathsf{O})$-formulas that are sufficient to consider% that we need to consider
.
\begin{lemma}
	\label{lem:acc2}
	Consider a sample $S = (P,N)$ of CGS and a subset of operators $\mathsf{O} \subseteq \mathsf{OP}$. For all $T \in \mathsf{Apr}(S,\mathsf{O})$, there is an $\mathbf{ATL}(\mathsf{O})$-formula $\Phi$ of size at most $2^n$ such that $\SAT(\Phi) = T$.
	
	In particular, this implies $\mathsf{Apr}(S,\mathsf{O}) \subseteq \mathsf{Acc}(S,\mathsf{O})$. 
\end{lemma}
\begin{proof}
	%We define inductively on $\mathsf{Apr}(S,\mathsf{O}) = \cup_{1 \leq i \leq 2^n} \mathsf{Apr}_i(S,\mathsf{O})$ a function $f: \mathsf{Apr}(S,\mathsf{O}) \rightarrow \mathbf{ATL}(\mathsf{O})$ such that that satisfies the property $\mathcal{P}(i)$: for all $T \in \mathsf{Apr}_i(S,\mathsf{O})$,  $\mathsf{SubF}(f(T)) \subseteq f(\mathsf{Apr}_i(S,\mathsf{O}))$.
	
	We define a function $f: \mathsf{Apr}(S,\mathsf{O}) \rightarrow \mathbf{ATL}(\mathsf{O})$ such that, for all $T \in \mathsf{Apr}(S,\mathsf{O})$, we have $\SAT(f(T)) = T$. We define this function inductively on $1 \leq i \leq 2^n$ such that it satisfies the property $\mathcal{P}(i)$: for all $T \in \mathsf{Apr}_i(S,\mathsf{O})$,  $\mathsf{SubF}(f(T)) \subseteq f(\mathsf{Apr}_i(S,\mathsf{O}))$.
	
	Consider any set $T \in \mathsf{Apr}_1(S,\mathsf{O})$. By definition, there is some proposition $p \in \prop$ such that $T = \SAT(p)$. We let $f(T) := p$, which ensures that the property $\mathcal{P}(1)$ holds. 
	
	Assume now that the function $f$ is defined on $\mathsf{Apr}_i(S,\mathsf{O})$ for some $1 \leq i \leq 2^n-1$ and that $\mathcal{P}(i)$ holds. Let us define $f$ on $\mathsf{Apr}_{i+1}(S,\mathsf{O})$. Let $T \in \mathsf{Apr}_{i+1}(S,\mathsf{O}) \setminus \mathsf{Apr}_i(S,\mathsf{O})$. There are several (quite similar) cases:
	\begin{itemize}
		\item Assume that $\neg \in \mathsf{O}$ and that $T = \mathsf{Upd}_{\neg}(S,T')$ for some $T' \in \mathsf{Apr}_i(S,\mathsf{O})$% and that $\mathcal{P}(T')$ holds
		. Then, we let $f(T) := \neg f(T') \in \mathbf{ATL}(\mathsf{O})$. By Lemma~\ref{lem:proper_behavior} and $\mathcal{P}(i)$, we have $\SAT(f(T)) = \mathsf{Upd}_{\neg}(S,\SAT(f(T'))) = \mathsf{Upd}_{\neg}(S,T') = T$. Furthermore, $\mathsf{SubF}(f(T)) = \{f(T)\} \cup \mathsf{SubF}(f(T')) \subseteq f(\mathsf{Apr}_{i+1}(S,\mathsf{O}))$. %Hence, $\mathcal{P}(T)$ holds. 
		\item Let $\bullet \in \{\lX,\lG\}$. Assume that $\bullet \in \mathsf{O}$ and that $T = \mathsf{Upd}_{\bullet}(S,A,T')$ for some $T' \in \mathsf{Apr}_i(S,\mathsf{O})$ and $A \subseteq \Ag$% and that $\mathcal{P}(T')$ holds
		. Then, we let $f(T) := \fanBr{A} \bullet f(T') \in \mathbf{ATL}(\mathsf{O})$. By Lemma~\ref{lem:proper_behavior} and $\mathcal{P}(i)$, we have $\SAT(f(T)) = \mathsf{Upd}_{\bullet}(S,A,\SAT(f(T'))) = \mathsf{Upd}_{\bullet}(S,A,T') = T$. Furthermore, $\mathsf{SubF}(f(T)) = \{f(T)\} \cup \mathsf{SubF}(f(T')) \subseteq f(\mathsf{Apr}_{i+1}(S,\mathsf{O}))$. %Hence, $\mathcal{P}(T)$ holds.
		\item Assume that $\wedge \in \mathsf{O}$ and that $T = \mathsf{Upd}_{\wedge}(S,T_1,T_2)$ for some $T_1,T_2 \in \mathsf{Apr}_i(S,\mathsf{O})$% and that $\mathcal{P}(T_1)$ and $\mathcal{P}(T_2)$ hold
		. Then, we let $f(T) := f(T_1) \wedge f(T_2) \in \mathbf{ATL}(\mathsf{O})$. By Lemma~\ref{lem:proper_behavior} and $\mathcal{P}(i)$, we have $\SAT(f(T)) = \mathsf{Upd}_{\wedge}(S,\SAT(f(T_1)),\SAT(f(T_2))) = \mathsf{Upd}_{\wedge}(S,T_1,T_2) = T$. Furthermore, $\mathsf{SubF}(f(T)) = \{f(T)\} \cup \mathsf{SubF}(f(T_1)) \cup \mathsf{SubF}(f(T_2)) \subseteq f(\mathsf{Apr}_{i+1}(S,\mathsf{O}))$. %Hence, $\mathcal{P}(T)$ holds. 
		\item Assume that $\lU \in \mathsf{O}$ and that $T = \mathsf{Upd}_{\lU}(S,A,T_1,T_2)$ for some $T_1,T_2 \in \mathsf{Apr}_i(S,\mathsf{O})$ and $A \subseteq \Ag$% and that $\mathcal{P}(T_1)$ and $\mathcal{P}(T_2)$ hold
		. Then, we let $f(T) := \fanBr{A} f(T_1) \lU f(T_2) \in \mathbf{ATL}(\mathsf{O})$. By Lemma~\ref{lem:proper_behavior} and $\mathcal{P}(i)$, we have $\SAT(f(T)) = \mathsf{Upd}_{\lU}(S,A,\SAT(f(T_1)),\SAT(f(T_2))) = \mathsf{Upd}_{\lU}(S,A,T_1,T_2) = T$. Furthermore, $\mathsf{SubF}(f(T)) = \{f(T)\} \cup \mathsf{SubF}(f(T_1)) \cup \mathsf{SubF}(f(T_2)) \subseteq f(\mathsf{Apr}_{i+1}(S,\mathsf{O}))$. %Hence, $\mathcal{P}(T)$ holds. 
	\end{itemize}
	Therefore, $\mathcal{P}(i+1)$ holds. In fact, $\mathcal{P}(i)$ holds for all $1 \leq i \leq 2^n$. 
	
	Hence, for all $T \in \mathsf{Apr}(S,\mathsf{O}) = \mathsf{Apr}_{2^n}(S,\mathsf{O})$, there is an $\mathbf{ATL}(\mathsf{O})$-formula $\Phi$ such that $\SAT(\Phi) = T$ and $\mathsf{SubF}(\Phi) \subseteq f(\mathsf{Apr}_{2^n}(S,\mathsf{O}))$. Therefore, $|\Phi| = |\mathsf{SubF}(\Phi)| \leq |\mathsf{Apr}_{2^n}(S,\mathsf{O})| \leq |2^Q| = 2^n$. 
\end{proof}

The exponential bound on the size of formulas is established in Lemma~\ref{lem:acc2}. We can then deduce Lemma~\ref{lem:decide_exponential_time_ATLn} by carefully looking at the complexity of computing the set $\mathsf{Apr}_i(S,\mathsf{O})$. This is done in Appendix~\ref{subsec:complexity_exponential}. Then, Theorem~\ref{thm:bound_size_formulas} follows directly from Lemma~\ref{lem:decide_exponential_time_ATLn}.

\section{Experimental Evaluation}
\label{sec:implem}

\begin{figure}[t]
			\begin{tikzpicture}
			\begin{axis}[%[log origin=infty] 
			height=42mm,
			width=50mm,
			xmin=20, ymin=-10, ymax=2000,
			enlarge x limits=true, enlarge y limits=true,
			xlabel = {Number of examples},
			ylabel = {Runtime (in secs)},
			xtick={20,40,60,80,100,120},%
			ytick={0,1000,2000},
			yticklabels = {0,1K,2K},
			ytickten= {0, 1, 2, 3},%
			%extra x ticks={2000}, extra x tick labels={\strut TO},%
			extra y ticks={2400}, %extra y tick labels=\empty%, 
			extra y tick labels={\strut TO},
			xminorticks=false,
			yminorticks=false,
			xlabel near ticks,
			ylabel near ticks,
			%xtick style={draw=major ticks},
			%ytick style={draw=none},
			label style={font=\footnotesize},
			x tick label style={font={\strut\tiny}},
			y tick label style={font={\strut\tiny}},		
			x label style = {yshift=1mm},
			legend columns = 1,
			legend style={title={CTL learning}, at={(3.17,1.18)}, font=\scriptsize, draw=none, fill=none},
			title = {CTL learning},
			title style = {font=\footnotesize, inner sep=-4pt}
			]
			
			\addlegendimage{empty legend}
			\addlegendentry{\hspace{-2mm}\textbf{CTL formulas}}
			
			\addplot[color=red,mark=x,opacity=0.7]
			plot coordinates {
				(20,10.48)
				(40,39.29)
				(60,82.13)
				(80,142.38)
				(100,216.02)
				(120,315.88)
			};
			\addlegendentry{$\lA\lG(\lA\lF(p))$}
			\addplot[mark=*,green,opacity =0.7] plot coordinates {
				(20,110.39)
				(40,305.97)
				(60,634.72)
				(80,801)
				(100,1456.74)
				(120,1415.79)
			};
			\addlegendentry{$\lA\lG(\neg p\!\lor\!\neg q)$}
			
			\addplot[mark=+,orange,opacity =0.7] plot coordinates {
				(20,101.77)
				(40,243.21)
				(60,544.60)
				(80,822.465)
				(100,1328.61)
				(120,1430.26)
			};
			\addlegendentry{$\lA\lG(p\!\rightarrow\!\lA\lF(q))$}
			 
			\addplot[mark=diamond,blue,opacity =0.7] plot coordinates {
				(20,11.315)
				(40,340.74)
				(60,740.52)
				(80,1020.16)
				(100,1472.84)
				(120,1590.76)
			};
			\addlegendentry{$\lA\lG(\lA\lF(p)\!\land\!\lA\lF(q))$}

			%				\addplot[mark=+,green,opacity =0.7] plot coordinates {
			%					(20,10.62)
			%					(40,38.68)
			%					(60,85.08)
			%					(80,147.78)
			%					(100,238.15)
			%					(120,341.37)
			%				};
			%				\addlegendentry{$\lA\lG(\lE\lF(q))$}
			
%			\addplot[mark=*,green,opacity =0.7] plot coordinates {
%				(20,95.01)
%				(40,206.38)
%				(60,450.86)
%				(80,767.88)
%				(100,1126.43)
%				(120,1372.21)
%			};
%			\addlegendentry{$\lA\lG(p\!\lor\!\lA\lX(\lnot q))$}
0			\end{axis}
			\hspace{35mm}
			\begin{axis}[%[log origin=infty] 
			height=42mm,
			width=50mm,
			xmin=10, ymin=0, ymax=2000,
			enlarge x limits=true, enlarge y limits=true,
			xlabel = {Number of examples},
			ylabel = {},
			xtick={10,20,30,40,50,60},%
			ytick={0,1000,2000},
			%ytickten= {0, 1, 2, 3},%
			%extra x ticks={2000}, extra x tick labels={\strut TO},%
			%extra y ticks={2400}, %extra y tick labels=\empty%, 
			%extra y tick labels={\strut TO},
			xminorticks=false,
			yminorticks=false,
			xlabel near ticks,
			ylabel near ticks,
			%xtick style={draw=major ticks},
			%ytick style={draw=none},
			label style={font=\footnotesize},
			x tick label style={font={\strut\tiny}},
			yticklabels = {},
			y tick label style={font={\strut\tiny}},		
			x label style = {yshift=1mm},
			legend columns = 1,
			legend style={at={(1.65,0.35)},font=\scriptsize,
				anchor=north, draw=none, fill=none},
			title = {ATL learning},
			title style = {font=\footnotesize, inner sep=-4pt}
			]
						
			\addlegendimage{empty legend}
			\addlegendentry{\hspace{-3mm}\textbf{ATL formulas}}
			
%			\addplot[color=olive,mark=x,opacity =0.8]
%			plot coordinates {
%				(10, 5.5)
%				(20, 21.8)
%				(30, 47.69)
%				(40, 84.12)
%				(50, 131.62)
%				(60, 189.72)
%			};
%			\addlegendentry{$\fanBr{1}\lX(p)$}
			
			\addplot[mark=+,olive,opacity =0.7] plot coordinates {
								(10, 4.9)
								(20, 22.72)
								(30, 48.80)
								(40, 82.75)
								(50, 131.186)
								(60, 186.65)
							};
			\addlegendentry{$\fanBr{0}\lF(q)$}
			
			\addplot[mark=star,magenta,opacity =0.7] plot coordinates {
				(10, 45.74)
				(20, 182.86)
				(30, 415.181)
				(40, 748.00)
				(50, 1155.71)
				(60, 1484.91)
			};
			\addlegendentry{$\fanBr{}\lG(p\!\rightarrow\!\fanBr{1}\lG(p))$}

			\addplot[mark=triangle,black,opacity =0.7] plot coordinates {
				(10, 41.59)
				(20, 392.17)
				(30, 835.11)
				(40, 1424.98)
				(50, 1586.71)
				(60, 1752.46)
				
			};
			\addlegendentry{$\fanBr{}\lG(p\rightarrow \fanBr{0,\!1}\lF(q))$}
			
				\addplot[mark=*,teal,opacity =0.7] plot coordinates {
				(10, 90.92)
				(20, 1012.32)
				(30, 1690.90)
				(40, 1698.69)
				(50, 1893.51)
				(60, 1924.78)
			};
			\addlegendentry{$\fanBr{}\lG((p\!\wedge\!\neg q)\!\rightarrow\!\fanBr{1}\lG(p))$}

			%			\addplot[mark=triangle,black,opacity =0.7] plot coordinates {
			%				(10, 107.31)
			%				(20, 391.87)
			%				(30, 803.79)
			%				(40, 1425.39)
			%				(50, 1613.98)
			%				(60, 1788.75)
			%			};
			%			\addlegendentry{$\fanBr{}\lG(p\!\rightarrow\!\fanBr{1}\lX(q))$ }
				
			\end{axis}
			
			\end{tikzpicture}

	\caption{Runtime of CTL and ATL learning algorithms on samples with varying number of examples (considering structures of size $\leq 20$).}
	\label{fig:runtime_results}
\end{figure}
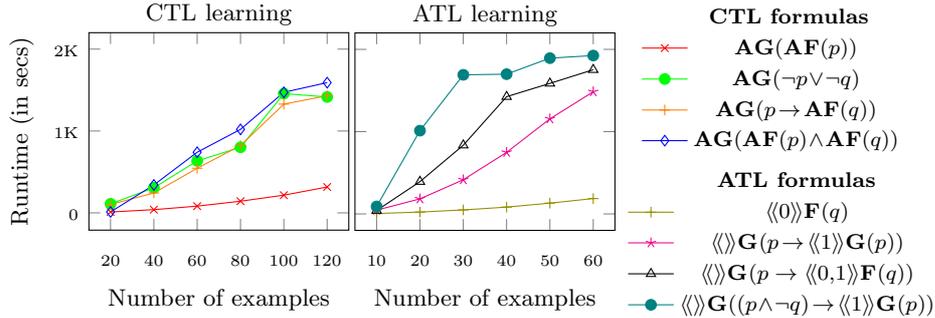

To test the ability of our learning algorithms (from Section~\ref{sec:sat-based}), we implement them in an open-source prototype\footnote{\url{https://github.com/rajarshi008/learning-branching-time}}.
The prototype is developed in Python3 and utilizes the PySMT library~\cite{pysmt2015}, providing us access to an array of SAT solvers.
Moreover, it offers configurable parameters for learning, including options to specify the model type (KS or CGS) and the formula type (CTL or ATL).
%\renewcommand{\arraystretch}{1.2}
%\begin{table}[h]
%	\vspace{-3mm}
%	\centering
%	\caption{Examples of CTL and ATL formulas used for generating benchmarks}
%	\label{tab:TL-patterns}
%	\setlength{\tabcolsep}{8pt}
%	%\resizebox{\linewidth}{!}{%
%	\begin{tabular}{|c|c|}
%		\hline
%		CTL formulas & ATL formulas \\
%		\hline
%		$\lA\lG(\lA\lF(p))$ & $\fanBr{1}\lX(p)$ \\
%		$\lA\lG(p \rightarrow \lA\lF(q))$ & $\fanBr{}\lG(p\rightarrow \fanBr{1}\lG(p))$ \\
%		$\lA\lG(\neg p \lor \neg q)$ & $\fanBr{}\lG(p\rightarrow \fanBr{0,\!1}\lF(q))$  \\
%		$\lA\lG(\lA\lF(p) \land \lA\lF(q))$ &  $\fanBr{}\lG((p\wedge \neg q)\rightarrow \fanBr{1}\lG(p))$ \\
%		\hline
%	\end{tabular}
%	\vspace{-1mm}
%\end{table}

To test various aspects of our algorithms, we rely on synthetic benchmarks generated from a set of chosen formulas as is common in the literature~\cite{flie,CamachoM19}.
We first identified some common CTL and ATL formulas from standard sources: the book on model-checking by Baier and Katoen~\cite[Example 6.3]{model-checking-book} for CTL and the seminal work by Alur et al.~\cite[Examples 3.1,3.2]{AlurATL} for ATL.
Fig.~(\ref{fig:runtime_results}) displays a selection of formulas (see Appendix~\ref{subsec:formula-list} for the full list); we used abstract propositions and players in the formulas to maintain uniformity.
The CTL formulas reflect properties occurring in distributed systems, such as mutual exclusion, request-response progress, etc.
The ATL formulas describe various properties of a train-controller (multi-agent) system.

We generated samples following a random generation technique.
To construct a structure $C$ for a sample, we iteratively added states to $C$ ensuring its connectedness; assigned random propositions to each state of $C$; and added random edges, labeled with action tuples if generating a CGS.
For CGSs, we focussed on generating turn-based games, i.e., games where only one player has actions in each state, due to their prevalence in verification and synthesis.
We split the randomly generated structures into $P$ and $N$ based on the selected formula.

Overall, from six CTL formulas and six ATL formulas, we generated two benchmark suites, the first one consisting of 144 samples of KSs and the second one consisting of 144 samples of CGSs, respectively.
The number of examples in the suites ranges from 20 to 120 and 10 to 60, respectively.
Moreover, the sizes of the KSs range from 1 to 40, while the size of the CGSs range from 1 to 20. 

All tests were run on Intel Xeon Gold 6142 CPU (at 2.6~GHz) using up to 10~GB of RAM using MathSAT solver~\cite{DBLP:conf/tacas/CimattiGSS13}\footnote{MathSAT~\cite{DBLP:conf/tacas/CimattiGSS13} performed better than Z3~\cite{MouraB08} and Boolector~\cite{DBLP:conf/tacas/BrummayerB09} in our experiments.}, with a timeout of 2400 seconds.

How effective are the algorithms in learning CTL/ATL formulas?
To answer this, we ran the CTL and ATL learning algorithms on the first and second benchmark suites, respectively.
Fig.~\ref{fig:runtime_results} depicts runtimes of both algorithms for samples generated from a selection of formulas (full results are available in Appendix~\ref{subsec:full-experiments}).
Both algorithms exhibit reasonable runtime performance for small samples.
As the samples size increases, the runtime also increases, more prominently for the larger formulas in our selection.
For the smaller formulas, however, the runtime of the algorithms remains reasonable even when the samples size increases.

\begin{figure}[t]
	\centering
	\begin{subfigure}[b]{0.4\textwidth}
		\centering
		\begin{tikzpicture}
		\begin{loglogaxis}[
		height=40mm,
		width=40mm,
		xmin=1e-1, ymin=1e-1,
		xmax=3600,ymax=3600,
		enlarge x limits=true, enlarge y limits=true,
		xlabel = {Runtime with $\lU$},
		ylabel = {Runtime without $\lU$},
		y label style={xshift=1mm,yshift=-2mm},
		xtickten={-1, 0, 1, 2},%
		ytickten={-1, 0, 1, 2},%
		extra x ticks={3600}, extra x tick labels={\strut TO},%
		extra y ticks={3600}, %extra y tick labels=\empty%, 
		extra y tick labels={\strut TO},
		xminorticks=false,
		yminorticks=false,
		xlabel near ticks,
		ylabel near ticks,
		title = {CTL learning},
		title style = {font=\footnotesize, inner sep=-4pt},
		%xtick style={draw=major ticks},
		%ytick style={draw=none},
		label style={font=\footnotesize},
		x tick label style={font={\strut\tiny}},
		y tick label style={font={\strut\tiny}},		
		x label style = {yshift=1mm},
		]
		
		\addplot[
		color=black,
		scatter=false,
		only marks,
		mark=x,
		mark size=1.2,
		]
		table [col sep=comma,x={Total_Time_U},y={Total_Time_WU}] {data/CTL_comp_runtimes.csv};
		
		% Diagonal
		\draw[black!25] (rel axis cs:0, 0) -- (rel axis cs:1, 1);
		
		\end{loglogaxis}
		\end{tikzpicture}
		\caption{Improvement by dropping $\lU$}
		\label{fig:improve1}
	\end{subfigure}
	\hskip .5cm
	\begin{subfigure}[b]{0.4\textwidth}
		\centering
		\begin{tikzpicture}
		\begin{loglogaxis}[
		height=40mm,
		width=40mm,
		xmin=1e-1, ymin=1e-1,
		xmax=3600,ymax=3600,
		enlarge x limits=true, enlarge y limits=true,
		xlabel = {Runtime with $\Omega_{q,i}^{\mathsf{Pre}}$},
		ylabel = {Runtime with $\bar{\Omega}_{q,i}^{\mathsf{Pre}}$},
		xtickten={-1, 0, 1, 2},%
		ytickten={-1, 0, 1, 2},
		extra x ticks={3600}, extra x tick labels={\strut TO},%
		extra y ticks={3600}, %extra y tick labels=\empty%,
		extra y tick labels={\strut TO}, 
		extra y tick labels={\strut TO},
		xminorticks=false,
		yminorticks=false,
		xlabel near ticks,
		ylabel near ticks,
		%xtick style={draw=major ticks},
		%ytick style={draw=none},
		label style={font=\footnotesize},
		x tick label style={font={\strut\tiny}},		
		x label style = {yshift=1mm},
		y tick label style={font={\strut\tiny}},
		title = {ATL learning},
		title style = {font=\footnotesize, inner sep=-4pt}
		]
		
		\addplot[
		scatter=false,
		only marks,
		mark=x,
		mark size=1.2,
		]
		table [col sep=comma ,x={Total_Time},y={Total_Time_TB}] {data/ATL_comp_runtimes.csv};
		%			
		% Diagonal
		\draw[black!25] (rel axis cs:0, 0) -- (rel axis cs:1, 1);
		
		\end{loglogaxis}
		\end{tikzpicture}
		\caption{Turn-based games improvement}
		\label{fig:improve2}
	\end{subfigure}
	\caption{Runtime improvement with optimized encodings for CTL and ATL learning on first and second benchmark suites, respectively. The runtimes are in seconds and \textquotedblleft TO\textquotedblright{} represents timeout.}
	\label{fig:improvements}
\end{figure}
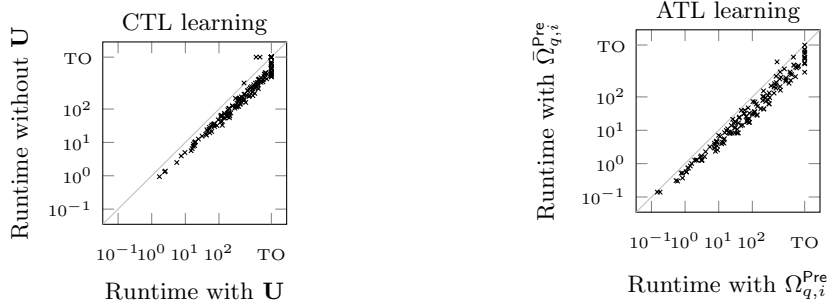

Which parts of the SAT encoding contribute significantly towards the runtime?
To understand this, we profiled both algorithms to identify the constraints responsible for significant runtime increases.
Notably, generating constraints for the $\lU$ operators (i.e., $\lE\lU$, $\lA\lU$ for CTL and $\fanBr{\cdot}\lU$ for ATL) turned out to be the most time-consuming among all operators.
Consequently, we compared runtime performance with and without these operators. (Note that we included $\lF$ and $\lG$ operators in both cases.) 
Learning without the $\lU$ operator is often justified, since several properties do not require the $\lU$ operator~\cite{scarlet}.
Fig.~\ref{fig:improve1} depicts this comparison for CTL learning; the average runtime improvement is 46\%.

Moreover, computing the constraints for $\Omega_{q,i}^{\mathsf{Pre}}(y)$ happens to be expensive due to the nesting of conjunction and disjunction of propositional formulas.
As a result, we noticed an improvement when we used the following optimized encoding designed specifically for turn-based games:
\[\bar{\Omega}_{q,i}^{\mathsf{Pre}}(y_S) \coloneqq [A_{i, \sigma(q)} \Rightarrow \bigvee_{\alpha\in \mathsf{Act}_\Ag(q)} y_S(\delta(q,\alpha)] \wedge [\neg A_{i, \sigma(q)} \Rightarrow \bigwedge_{\alpha\in\mathsf{Act}_\Ag(q)} y_S(\delta(q,\alpha)],\] where $\sigma(q)$ denotes the Agent owning the state $q$.
Fig.~\ref{fig:improve2} illustrates this improvement; the average runtime improvement is 52\%.

Overall, we demonstrated that our learning algorithms can successfully infer several CTL and ATL formulas from samples of varying sizes.

\section{Conclusion}
In this work, we address the passive learning problem for CTL and ATL, proposing a constraint-solving technique that encodes search and model-checking of formulas in propositional logic. 
We additionally investigate the separability of ATL formulas and develop decision procedures for it. 
Our experimental evaluations validate the efficacy of our algorithms in inferring useful CTL/ATL formulas.

As future work, we like to explore the computational hardness of the passive learning problems for CTL/ATL, similar to~\cite{DBLP:conf/icgi/FijalkowL21}.
We also plan to improve our prototype by adding heuristics, such as the ones discussed in~\cite{DBLP:journals/corr/abs-2402-06366}.
Finally, we like to lift our techniques to probabilistic logics such as PCTL~\cite{DBLP:journals/fac/HanssonJ94} and PATL~\cite{DBLP:conf/fskd/ChenL07a}.

\bibliographystyle{splncs04}
\bibliography{refs}

\appendix
%!TEX root = LearnATLarXiv.tex
\appendix
\section{Structural conditions}
\label{appen:structural_cond}
As mentioned in the main part of the paper, we impose structural constraints on the introduced variables to ensure that they do encode an ATL formula.
The constraints are similar to the ones proposed by Neider and Gavran~\cite{flie}. We define:
\begin{align*}
	\Omega_\Lambda & := \Big[ \bigwedge_{1\leq i \leq n} \bigvee_{\lambda \in \Lambda} x_{i,\lambda} \Big] \land \Big[\bigwedge_{1 \leq i \leq n} \bigwedge_{\lambda \neq \lambda' \in \Lambda} \lnot x_{i,\lambda} \lor \lnot x_{i,\lambda'}  \Big],%\label{eq:label-unique}
	\\
	\Omega_l & := \Big[ \bigwedge_{2\leq i \leq n} \bigvee_{1\leq j \leq i} l_{i,j} \Big] \land \Big[\bigwedge_{2 \leq i \leq n} \bigwedge_{1 \leq j < j' < i} \lnot l_{i,j} \lor \lnot l_{i,j'}  \Big],%\label{eq:left-unique}
	\\
	\Omega_r & := \Big[ \bigwedge_{2\leq i \leq n} \bigvee_{1\leq j \leq i} r_{i,j} \Big] \land \Big[\bigwedge_{2 \leq i \leq n} \bigwedge_{1 \leq j < j' < i} \lnot r_{i,j} \lor \lnot r_{i,j'}  \Big],%\label{eq:right-unique}
	\\
	\Omega_\prop & := \bigvee_{p\in\prop} x_{1,p}%\label{eq:only-prop}
\end{align*}

%where $\Lambda = \{\neg,\lor,\fanBr{\cdot}\lX,\fanBr{\cdot}\lU,\fanBr{\cdot}\lG \}$. 
The formula $\Omega_\Lambda$ ensures that each node is labeled by exactly one operator or one proposition. The formulas $\Omega_l$ and $\Omega_r$ enforce that each node (except for Node 1) has exactly one left and exactly one right child, respectively.
Finally, the formula $\Omega_\prop$ ensures that Node~$1$ is labeled by a proposition. 

Based on a satisfying valuation $v$ of $\Omega^{\mathsf{str}}$, one can straightforwardly construct a unique ATL formula as follows: label each Node~$i$ with the unique operator $\lambda$ for which $v(x_{i,\lambda})=1$, taking into account the variables $(A_{i,a})_{a \in \Ag}$ if $\lambda \in \{\fanBr{\cdot}\lX,\fanBr{\cdot}\lG,\fanBr{\cdot}\lU\}$, and mark the left and right children of the node with the unique Nodes $j$ and $j'$ for which $v(l_{i,j})=1$ and $v(r_{i,j'})=1$, respectively.

\section{Correctness of the encoding}
%\label{sec:appen_correctness}
\label{subsec:appen_correctness}
Consider a concurrent game structure $C = \langle Q,I,k,\prop, \pi, \sigma, d,\delta \rangle$ and some integer $n \in \N^+$. We consider the formulas $\propformula^{\mathsf{str}}_n$ and $\propformula^{\mathsf{sem}}_n$. We let $X_n(C)$ denote the variables occurring in these formulas. We have:
\begin{align*}
	X_n(C) = & X^\mathsf{Var}_n \cup X^{\mathsf{Struct}}_C
\end{align*}
with
\begin{align*}
	X^\mathsf{Var}_n := \; & \{ x_{i,\lambda} \mid 1 \leq i \leq n,\; \lambda \in \Lambda \} \cup \{ A_{i,a} \mid 1 \leq i \leq n,\; a \in \Ag \} \cup \\
	& \{ l_{i,j},r_{i,j} \mid 1 \leq i \leq n,\; 1 \leq j \leq i-1 \}
\end{align*}
and 
\begin{align*}
	X^{\mathsf{Struct}}_{C} := & \; \{ y^C_{i,q} \mid 1 \leq i \leq n,\; q \in Q \} \cup \\
	& \; \{ y^C_{i,q,k} \mid 1 \leq i \leq n,\; q \in Q,\; 1 \leq k \leq |Q|+1 \}
\end{align*}

The proof of Proposition~\ref{prop:pass-learning-correctness} is rather long and somewhat tedious. However, note that there is not much difficulty involved. The length of the proof is due to the several cases of all the operators to handle.

First of all, we define formally the ATL formula that a valuation $v: X^\mathsf{Var}_n \rightarrow \B$ satisfying the $\propformula^{\mathsf{str}}_n$ encodes.
\begin{definition}
	\label{def:formula_from_valuation}
	Consider a valuation $v: X^\mathsf{Var}_n \rightarrow \B$ satisfying the formula $\propformula^{\mathsf{str}}_n$. We define $\msf{Op}: \llbracket 1,n \rrbracket \rightarrow \prop \cup \msf{Op}(\prop,\Ag)_{\mathsf{Un}} \cup \msf{Op}(\prop,\Ag)_{\mathsf{Bin}}$ with $\msf{Op}(\prop,\Ag)_{\mathsf{Un}} := \{\neg,\fanBr{A} \lX,\fanBr{A} \lG \mid A \subseteq \Ag \}$ and $\msf{Op}(\prop,\Ag)_{\mathsf{Bin}} := \{\wedge,\fanBr{A} \lU \mid A \subseteq \Ag \}$. Let $i \in \llbracket 1,n \rrbracket$. Since $v \models \propformula^{\mathsf{str}}_n$, there is a unique $\lambda \in \Lambda$ such that $v(x_{i,\lambda}) = 1$. Then:
	\begin{itemize}
		\item If $\lambda \in \prop \cup \{ \neg,\wedge \}$, we let $\msf{Op}(i) := \lambda$;
		\item If $\lambda = \fanBr{\cdot} \bullet$ for some $\bullet \in \{ \lX,\lG,\lU \}$, we let $\msf{Op}(i) := \fanBr{A_i} \bullet$ with $A_i \subseteq \Ag$ such that for all $a \in \Ag$, we have $a \in A_i$ if and only if $v(A_{i,a}) = 1$.
	\end{itemize}

	We also let $r: \llbracket 2,n \rrbracket \rightarrow \llbracket 1,n \rrbracket$ be such that, for all $i \in \llbracket 2,n \rrbracket$, $r(i) \in \llbracket 1,i-1 \rrbracket$ is such that $v(r_{i,r(i)}) = 1$. Note that there is exactly one such $r(i) \in \llbracket 1,n \rrbracket$ because $v \models \propformula^{\mathsf{str}}_n$. We define similarly $l: \llbracket 2,n \rrbracket \rightarrow \llbracket 1,n \rrbracket$. 
	\iffalse
	We also define the set $\mathsf{Defined}^v \subseteq \llbracket 1,n \rrbracket$ as the smallest subset of $\llbracket 1,n \rrbracket$ such that:
	\begin{enumerate}
		\item $n \in \mathsf{Defined}^v$;
		\item $\{ l(i) \mid i \in \mathsf{Defined}^v,\; \msf{Op}(i) \in \msf{Op}(\prop,\Ag)_{\mathsf{Un}} \} \subseteq \mathsf{Defined}^v$;
		%for all $i \in \mathsf{Defined}^v$, if $\msf{Op}(i) \in \{\neg,\fanBr{A} \lX,\fanBr{A} \lG \mid A \subseteq \Ag \}$, then $l(i) \in \mathsf{Defined}^v$;
		\item $\{ l(i),r(i) \mid i \in \mathsf{Defined}^v,\; \msf{Op}(i) \in \msf{Op}(\prop,\Ag)_{\mathsf{Bin}} \} \subseteq \mathsf{Defined}^v$
		%for all $i \in \mathsf{Defined}^v$, if $\msf{Op}(i) \in \{\wedge,\fanBr{A} \lU \mid A \subseteq \Ag \}$, then $l(i),r(i) \in \mathsf{Defined}^v$.
	\end{enumerate}
	\fi
	
	Finally, we define the function $\Phi^v_n: \llbracket1,n \rrbracket \rightarrow ATL$ as follows. For all $i \in \llbracket1,n \rrbracket$:
	\begin{itemize}
		\item If $\msf{Op}(i) \in \prop$, then $\Phi^v_n(i) := \msf{Op}(i)$;
		\item If $\msf{Op}(i) \in %\{\neg,\fanBr{A} \lX,\fanBr{A} \lG \mid A \subseteq \Ag \}
		\msf{Op}(\prop,\Ag)_{\mathsf{Un}}$, then $\Phi^v_n(i) := \msf{Op}(i) \; \Phi^v_n(l(i))$;
		\item If $\msf{Op}(i) \in %\{\wedge,\fanBr{A} \lU \mid A \subseteq \Ag \}
		\msf{Op}(\prop,\Ag)_{\mathsf{Bin}}$, then $\Phi^v_n(i) := \Phi^v_n(l(i)) \; \msf{Op}(i) \; \Phi^v_n(r(i))$.
	\end{itemize}
	
	One can show by induction (on the tree structure induced by the children relations given by the functions $l$ and $r$) that we do have that 
	%	$l$ and $r$) that, indeed, 
	for all $i \in \llbracket1,n \rrbracket$, $\Phi^v_n(i)$ is an ATL formula.
	
	We can finally define the ATL formula $\Phi^v$ that $v$ encodes: $\Phi^v := \Phi^v_n(n)$.
\end{definition}

\begin{remark}
	A valuation satisfying the formula $\propformula^{\mathsf{str}}_n$ encodes a single ATL formula. 
	On the other hand, given an ATL formula $\Phi$, there exist (several) valuations satisfying the formula $\propformula^{\mathsf{str}}_n$ that encode the formula $\Phi$. We do not give a proof of this statement here, a similar statement is proved in~\cite[Lemma 1]{flie}. We define such a valuation by considering a syntax DAG of $\Phi$. 
\end{remark}

Let us now show the formula $\propformula^{\mathsf{sem}}_n$ does capture the semantics of ATL formula. 
\begin{lemma}
	\label{lem:central_lemma}
	Consider a valuation $v: X^\mathsf{Var}_n \rightarrow \B$ satisfying the formula $\propformula^{\mathsf{str}}_n$. For all $q \in Q$, we have:
	\begin{equation*}
		q \models \Phi^v \text{ iff } v(y_{n,q}^C) = 1
	\end{equation*}
\end{lemma}

We are going to prove this lemma by induction on the tree structure of the set $\llbracket1,n \rrbracket$. Before we proceed to the proof of this lemma, we state and prove intermediate lemmas handling the temporal and strategic ATL operators on their own.

First of all, we tackle the formula $\Omega_{q,i}^{\mathsf{Pre}}(y)$ in the lemma below.
\begin{lemma}
	\label{lem:pre}
	Consider a valuation $v: X_n(C) \rightarrow \B$ satisfying the formula $\propformula^{\mathsf{str}}_n \wedge \propformula^{\mathsf{sem}}_n$. Let $q \in Q$ and $i \in \llbracket 1,n \rrbracket$. For all predicates $Y \subseteq Q$, defined with $(v(y_q))_{q \in Q}$, we have (recall that $A_i$ is defined in Definition~\ref{def:formula_from_valuation}):
	\begin{equation*}
		v \models \Omega_{q,i}^{\mathsf{Pre}}(y) \text{ iff } q \in \mathsf{Pre}_{A_i}(Y)
	\end{equation*} 
\end{lemma}
\begin{proof}	
	Let $\alpha \in \mathsf{Act}_{\Ag}(q)$. Let $\Omega_{\alpha}$ denote the conjunction in $\Omega_{q,i}^{\mathsf{Pre}}(y)$ evaluated with the tuple of actions $\alpha \in \mathsf{Act}_{\Ag}(q)$. 
	%Assume that $v \models \Omega_{q,i}^{\mathsf{Pre}}(y)$. 
	%Let $\alpha \in \mathsf{Act}_{\Ag}(q)$ satisfying the disjunction of $\Omega_{q,i}^{\mathsf{Pre}}(y)$. 
	We let $\alpha_{i} \in \mathsf{Act}_{A_i}$ be such that, for all $a \in A_i$, we have $\alpha_{i}(a) := \alpha(a)$. We claim that $v \models \Omega_\alpha$ if and only if for all $q' \in \mathsf{Succ}(q,\alpha_i)$, we have $q' \in Y$ (i.e. $v(y_{q'}) = 1$). We denote by $\mathsf{Act}_{\Ag}(\alpha_i)$ the set of tuples of actions that coincide with $\alpha_i$ on $A_i$. Then, $v \models \Omega_\alpha$ iff for all $\alpha' \in \mathsf{Act}_{\Ag}$ such that, for all $a \in \Ag$, if $v(A_{i,a}) = 1$, then $\alpha(a) = \alpha'(a)$, we have $v(y_{\delta(q,\alpha')}) = 1$. By definition of $A_i$, this is equivalent to for all $\alpha' \in \mathsf{Act}_{\Ag}(\alpha_i)$, we have $v(y_{\delta(q,\alpha')}) = 1$. This is equivalent to, for all tuple of actions $\alpha_i' \in \mathsf{Act}_{\Ag \setminus A_i}$, we have $v(y_{\delta(q,(\alpha_i,\alpha_i'))}) = 1$. That is, for all $q' \in \mathsf{Succ}(q,\alpha_i)$, we have $q' \in Y$. We have established that $v \models \Omega_\alpha$ if and only if $\mathsf{Succ}(q,\alpha_i) \subseteq Y$ for all tuples of actions $\alpha \in \mathsf{Act}_{\Ag}(q)$. The lemma follows. 
	%. Let $q' \in \mathsf{Succ}(q,\alpha_i)$. There is some $\alpha_i' \in \mathsf{Act}_{\Ag \setminus A_i}$ such that $q' = \delta(q,\alpha')$ for $\alpha' := (\alpha_i,\alpha_i') \in \mathsf{Act}_{\Ag}$. That is, $\alpha'$ is a tuple of actions that coincides with $\alpha_i$ and $A_i$ and that coincides with $\alpha_i'$ on $\Ag \setminus \alpha_i$. By definition, for all $a \in \Ag$, if $v(A_{i,a}) = 1$, then $a \in A_i$ and $\alpha'(a) = \alpha_i(a)$. Therefore, since $v \models \Omega_{q,i}^{\mathsf{Pre}}(y)$, it follows that $v(y_{q'}) = v(y_{\delta(q,\alpha')}) = 1$. 
\end{proof}

This will allow us to formally tackle the case of the operator $\lX$ in the proof of Lemma~\ref{lem:central_lemma}. Let us consider the slightly more involved case of the operator $\lU$. 

\iffalse
The soundness of the semantics of the formula $\Omega^{\mathsf{sem}}_{\lX}$ follows.
\begin{lemma}
	Consider a valuation $v: X_n(C) \rightarrow \B$ satisfying the formula $\propformula^{\mathsf{str}}_n \wedge \propformula^{\mathsf{sem}}_n$. Let $q \in Q$ and $i \in \msf{Defined}^v$. If $\msf{Op}(i) = \fanBr{A_i} \lX$, then:
	\begin{equation*}
		q \models \Phi^v(i) \text{ if and only if }v(y_{i,q}^C) = 1
	\end{equation*} 
\end{lemma}
\begin{proof}
	
\end{proof}
\fi

%To show that our encoding indeed captures the meaning of the variables $y^{C}_{i,s,k}$, we present the following lemmas.
\begin{lemma}\label{lem:U-correctness}
	Consider a valuation $v: X_n(C) \rightarrow \B$ satisfying the formula $\propformula^{\mathsf{str}}_n \wedge \propformula^{\mathsf{sem}}_n$. Let $i \in \msf{Defined}^v$. Assume that $\msf{Op}(i) = \fanBr{A_i} \lU$ and that, for all $q \in Q$, we have $v(y_{q,l(i)}^C) = 1$ if and only if $q \models \Phi^v(l(i))$ and $v(y_{q,r(i)}^C) = 1$ if and only if $q \models \Phi^v(r(i))$. Then, for all $q \in Q$, we have:
	\begin{equation*}
		q \models \Phi^v(i) \text{ if and only if }v(y_{i,q}^C) = 1
	\end{equation*} 
	%Let $v$ be an assignment such that $v(x_{i,\lE\lU})=1$, $v(l_{i,j})=1$ and $v(r_{i,j'})=1$.
	%Then, for any state $s\in S$ in $M$ and $k\in \{1,\dots,|S|+1\}$, we have: $v(y^{M}_{i,s,k})=1$ if and only if there exists a path prefix $\pi = s_0 s_1 s_2\ldots s_t$ in $M$ where $s=s_0$ and $0\leq t < k$, such that $v(y^{M}_{j,\overline{s}})=1$ for all $\overline{s} \in \{s_0, s_1, \ldots, s_{t-1}\}$ and $v(y^{M}_{j',s_t})=1$.
\end{lemma}
\begin{proof}
	Recall that we have:
	\begin{equation*}
		\Phi^v(i) = \fanBr{A_i} \; \Phi^v(l(i)) \; \lU \; \Phi^v(r(i))
	\end{equation*}

	\iffalse
	Let us define inductively sets that capture the steps taken by Algorithm~\ref{alg:algo_U} to compute $\mathsf{SAT}_C(\Phi^v(i))$. We define the sets $(S_k)_{0 \leq k \leq |Q|}$ as follows:
	\begin{itemize}
		\item $S_0 := \mathsf{SAT}_C(\Phi^v(r(i)))$.
		\item for all $0 \leq k \leq |Q|-1$, $S_{k+1} := S_k \cup \{ q \in Q \mid q \in \mathsf{SAT}_C(\Phi^v(l(i))) \cap  \mathsf{Pre}_{A_i}(S)\}$ 
	\end{itemize}
	Then, $\mathsf{SAT}_C(\Phi^v(i)) = S_{|Q|}$. 
	\fi
	
	Consider Algorithm~\ref{alg:algo_U}. For all $0 \leq j \leq |Q|$, we let $S_j \subseteq Q$ denote the set $S$ when exiting the \emph{for} loop when $k = j$. When $j = 0$, this corresponds to initialization of the set $S$ before the \emph{for} loop is ever entered. We show by induction on $0 \leq k \leq |Q|$ the property $\mathcal{P}(k)$: for all states $q \in Q$, we have $v(y^{C}_{i,q,k}) = 1$ if and only if $q \in S_k$. 
	\begin{itemize}
		\item Let $q \in Q$. We have $S_0 = \mathsf{SAT}_C(\Phi^v(r(i)))$, i.e. $q \in S_0$ iff $q \models \Phi^v(r(i))$. Since $v \models \propformula^{\mathsf{sem}}_{\lU}$, we have $v(y_{i,q,0}^C) = v(y_{r(i),q}^C)$. Furthermore, by assumption, $v(y_{q,r(i)}^C) = 1$ if and only if $q \models \Phi^v(r(i))$. Therefore, $q \in S_0$ if and only if  $v(y_{i,q,0}^C) = 1$. Hence, $\mathcal{P}(0)$ holds.
		\item Assume now that $\mathcal{P}(k)$ holds for some $0 \leq k \leq |Q|-1$. Let $q \in Q$. %Let us show that $q\in S_{k+1} \setminus S_k$ if and only if $v(y^{C}_{i,q,k}) = 0$ and $v(y^{C}_{i,q,k+1}) =1$. 
		Assume that $q \in S_{k+1}$. If $q \in S_k$, then by $\mathcal{P}(k)$, we have $v(y^{C}_{i,q,k}) = 1$. Since $v \models \propformula^{\mathsf{sem}}_{\lU}$, this implies $v(y^{C}_{i,q,k+1}) =1$. Assume now that $q \in S_{k+1} \setminus S_k$. In that case, we have $q \in \mathsf{SAT}_C(\Phi^v(l(i)))$, i.e. $q \models \Phi^v(l(i))$. Therefore, by assumption, $v(y_{q,l(i)}^C) = 1$. In addition, we also have $q \in \mathsf{Pre}_{A_i}(S_k)$. Therefore, by Lemma~\ref{lem:pre}, this implies $v \models \Omega_{q,i}^{\mathsf{Pre}}(y)$ for a collection of variables $(y_q)_{q \in Q}$ such that, for all $q' \in Q$, we have $v(y_{q'}) = 1$ if and only if $q' \in S_k$. By $\mathcal{P}(k)$, this is verified by the collection $(y_{i,q,k}^C)_{q\in Q}$. Therefore, $v \models \Omega_{q,i}^{\mathsf{Pre}}(y_{i,\cdot,k}^C)$ and since $v \models \propformula^{\mathsf{sem}}_{\lU}$, we have $v(y^{C}_{i,q,k+1}) =1$. 
		
		Assume now $v(y^{C}_{i,q,k+1}) = 1$. Since $v \models \propformula^{\mathsf{sem}}_{\lU}$, we have $v(y^{C}_{i,q,k}) = 1$ or $v \models y_{q,l(i)}^C \wedge \Omega_{q,i}^{\mathsf{Pre}}(y_{i,\cdot,k}^C)$. In the first case, this implies $q \in S_k$, by $\mathcal{P}(k)$. In the second case, this implies both $v(y_{q,l(i)}^C) =1$, i.e. %$q \models \Phi^v(l(i))$ (by assumption) or 
		$q \in \SAT_C(\Phi^v(l(i)))$, and by Lemma~\ref{lem:pre} and $\mathcal{P}(k)$, $q \in \mathsf{Pre}_{A_i}(S_k)$. Therefore, in both cases, we have $q \in S_{k+1}$. Thus, $\mathcal{P}(k+1)$ holds.
	\end{itemize}
	In fact, $\mathcal{P}(k)$ holds for all $0 \leq k \leq |Q|$. Therefore, by $\mathcal{P}(|Q|)$ and since $v \models \propformula^{\mathsf{sem}}_{\lU}$, for all $q \in Q$, we have $q \in S_{|Q|} = \SAT_C(\Phi^v(i))$, i.e. $q \models \Phi^v(i)$, if and only if $v(y^{C}_{i,q}) = v(y^{C}_{i,q,|Q|}) = 1$. 
	\iffalse
	The proof requires to proceed via induction on the parameter $k$.
	\begin{itemize}
		\item As the base case, let $k=1$. 
		Using the constraint~\ref{eq:EU}, we have $v(y^{M}_{i,s,1})=v(y^{M}_{j',s})$. 
		Thus, $v(y^{M}_{i,s,1})=1$ if and only if there is a trivial path $\pi=s$ (with parameter $t<1$) such that $v(y^{M}_{j',s})=1$.
		
		\item As the induction hypothesis, let for any $s\in S$ and $k\leq K$ the following hold: $v(y^{M}_{i,s,K})=1$ if and only if there is a path prefix $\pi=s s_1 s_2 \ldots s_t$, $t< K$ such that $v(y^{M}_{j,\overline{s}})=1$ for all $\overline{s} \in \{s, s_1, \ldots, s_{t-1}\}$ and $v(y^{M}_{j',s_t})=1$.
		Now, using the constraint~\ref{eq:EU}, $v(y^{M}_{i,s,K+1})=1$ if and only if (i) $v(y^{M}_{i,s,K})=1$, meaning, there is a path prefix  $\pi=s s_1 s_2 \ldots s_t$, $t<K<K+1$ such that $v(y^{M}_{j,\overline{s}})=1$ for all $\overline{s} \in \{s, s_1, \ldots, s_{t-1}\}$ and $v(y^{M}_{j',s_t})=1$; or (ii) $v(y^{M}_{j',s})=1$ and for some $s'\in\post(s)$, $v(y^{M}_{i,s',K})=1$, meaning, there is a path prefix $\pi=s s' \ldots s_t$, $t<K+1$ such that $v(y^{M}_{j,s'})=1$ for all $\overline{s} \in \{s, s_1, \ldots, s_{t-1}\}$ and $v(y^{M}_{j',s_t})=1$.
	\end{itemize}
	\fi
\end{proof}

Let us now tackle the similar case of the operator $\lG$. 
\begin{lemma}\label{lem:G-correctness}
	Consider a valuation $v: X_n(C) \rightarrow \B$ satisfying the formula $\propformula^{\mathsf{str}}_n \wedge \propformula^{\mathsf{sem}}_n$. Let $i \in \msf{Defined}^v$. Assume that $\msf{Op}(i) = \fanBr{A_i} \lG$ and that, for all $q \in Q$, we have $v(y_{q,l(i)}^C) = 1$ if and only if $q \models \Phi^v(l(i))$. Then, for all $q \in Q$, we have:
	\begin{equation*}
		q \models \Phi^v(i) \text{ if and only if }v(y_{i,q}^C) = 1
	\end{equation*} 
	%Let $v$ be an assignment such that $v(x_{i,\lE\lU})=1$, $v(l_{i,j})=1$ and $v(r_{i,j'})=1$.
	%Then, for any state $s\in S$ in $M$ and $k\in \{1,\dots,|S|+1\}$, we have: $v(y^{M}_{i,s,k})=1$ if and only if there exists a path prefix $\pi = s_0 s_1 s_2\ldots s_t$ in $M$ where $s=s_0$ and $0\leq t < k$, such that $v(y^{M}_{j,\overline{s}})=1$ for all $\overline{s} \in \{s_0, s_1, \ldots, s_{t-1}\}$ and $v(y^{M}_{j',s_t})=1$.
\end{lemma}
The proof of this lemma is very close to the proof of Lemma~\ref{lem:U-correctness}.
\begin{proof}
	Recall that we have:
	\begin{equation*}
		\Phi^v(i) = \fanBr{A_i} \; \lG \; \Phi^v(l(i))
	\end{equation*}
	
	Consider Algorithm~\ref{alg:algo_G}. For all $0 \leq j \leq |Q|$, we let $S_j \subseteq Q$ denote the set $S$ when exiting the \emph{for} loop when $k = j$. When $j = 0$, this corresponds to initialization of the set $S$ before the \emph{for} loop is ever entered. We show by induction on $0 \leq k \leq |Q|$ the property $\mathcal{P}(k)$: for all states $q \in Q$, we have $v(y^{C}_{i,q,k}) = 1$ if and only if $q \in S_k$. 
	\begin{itemize}
		\item Let $q \in Q$. We have $S_0 = \mathsf{SAT}_C(\Phi^v(l(i)))$, i.e. $q \in S_0$ iff $q \models \Phi^v(l(i))$. Since $v \models \propformula^{\mathsf{sem}}_{\lG}$, we have $v(y_{i,q,0}^C) = v(y_{l(i),q}^C)$. Furthermore, by assumption, $v(y_{q,l(i)}^C) = 1$ if and only if $q \models \Phi^v(l(i))$. Therefore, $q \in S_0$ if and only if  $v(y_{i,q,0}^C) = 1$. Hence, $\mathcal{P}(0)$ holds.
		\item Assume now that $\mathcal{P}(k)$ holds for some $0 \leq k \leq |Q|-1$. Let $q \in Q$. %Let us show that $q\in S_{k+1} \setminus S_k$ if and only if $v(y^{C}_{i,q,k}) = 0$ and $v(y^{C}_{i,q,k+1}) =1$. 
		Assume that $q \in S_{k+1}$. In that case, we have $q \in S_k$, and hence by $\mathcal{P}(k)$, we have $v(y^{C}_{i,q,k}) = 1$. %Since $v \models \propformula^{\mathsf{sem}}_{\lG}$, this implies $v(y^{C}_{i,q,k+1}) =1$. 
		In addition, %Assume now that $q \in S_{k+1} \setminus S_k$. In that case, we have $q \in \mathsf{SAT}_C(\Phi^v(l(i)))$, i.e. $q \models \Phi^v(l(i))$. Therefore, by assumption, $v(y_{q,l(i)}^C) = 1$. In addition, 
		we also have $q \in \mathsf{Pre}_{A_i}(S_k)$. Therefore, by Lemma~\ref{lem:pre}, this implies $v \models \Omega_{q,i}^{\mathsf{Pre}}(y)$ for a collection of variables $(y_q)_{q \in Q}$ such that, for all $q' \in Q$, we have $v(y_{q'}) = 1$ if and only if $q' \in S_k$. By $\mathcal{P}(k)$, this is verified by the collection $(y_{i,q,k}^C)_{q\in Q}$. Therefore, $v \models \Omega_{q,i}^{\mathsf{Pre}}(y_{i,\cdot,k}^C)$. Since $v \models \propformula^{\mathsf{sem}}_{\lG}$, it follows that $v(y^{C}_{i,q,k+1}) = 1$. 
		
		Assume now $v(y^{C}_{i,q,k+1}) = 1$. Since $v \models \propformula^{\mathsf{sem}}_{\lG}$, we have $v(y^{C}_{i,q,k}) = 1$ and $v \models \Omega_{q,i}^{\mathsf{Pre}}(y_{i,\cdot,k}^C)$. Hence, by $\mathcal{P}(k)$, we have $q \in S_k$. In addition, by Lemma~\ref{lem:pre} and $\mathcal{P}(k)$, we have $q \in \mathsf{Pre}_{A_i}(S_k)$. Therefore, $q \in S_{k+1}$. Thus, $\mathcal{P}(k+1)$ holds.
	\end{itemize}
	In fact, $\mathcal{P}(k)$ holds for all $0 \leq k \leq |Q|$. By $\mathcal{P}(|Q|)$ and since $v \models \propformula^{\mathsf{sem}}_{\lG}$, for all $q \in Q$, we have $q \in S_{|Q|} = \SAT_C(\Phi^v(i))$, i.e. $q \models \Phi^v(i)$, if and only if $v(y^{C}_{i,q}) = v(y^{C}_{i,q,|Q|}) = 1$. 
\end{proof}

\iffalse
\begin{lemma}\label{lem:EG-correctness}
	Let $v$ be an assignment such that $v(x_{i,\lE\lG})=1$ and $v(l_{i,j})=1$.
	Then, for any state $s\in S$ in $M$ and $k\in \{1,\dots,|S|+1\}$, we have: $v(y^{M}_{i,s,k})=1$ if and only if there exists a path prefix $\pi = s_0 s_1 s_2\ldots s_t$ in $M$ where $s=s_0$ and $t \geq k-1$, such that $v(y^{M}_{j,\overline{s}})=1$ for all $\overline{s} \in \{s_0, s_1, \ldots, s_{t}\}$.
\end{lemma}
\begin{proof}
	The proof requires an induction on the parameter $k$ similar to the previous proof.
	\begin{itemize}
		\item As the base case, let $k=1$. 
		Using the constraint~\ref{eq:EG}, we have $v(y^{M}_{i,s,1})=v(y^{M}_{j,s})$. 
		Thus, $v(y^{M}_{i,s,1})=1$ if and only if there is a trivial path $\pi=s$ (with parameter $t\geq 0$) such that $v(y^{M}_{j,s})=1$.
		
		\item As the induction hypothesis, let for any $s\in S$ and $k\leq K$ the following hold: $v(y^{M}_{i,s,K})=1$ if and only if there is a path prefix $\pi=s s_1 s_2 \ldots s_t$, $t\geq K-1$ such that $v(y^{M}_{j,\overline{s}})=1$ for all $\overline{s} \in \{s, s_1, \ldots, s_{t}\}$.
		Now, using the constraint~\ref{eq:EG}, $v(y^{M}_{i,s,K+1})=1$ if and only if  $v(y^{M}_{j,s})=1$ and for some $s'\in\post(s)$, $v(y^{M}_{i,s',K})=1$, meaning, there is a path prefix $\pi=s s' \ldots s_t$, $t\geq K$ such that $v(y^{M}_{j,s'})=1$ for all $\overline{s} \in \{s, s_1, \ldots, s_{t}\}$.
	\end{itemize}
\end{proof}
\fi

We can now proceed to the proof of Lemma~\ref{lem:central_lemma}.
\begin{proof}
	Let us define a tree structure (with possibly muiltiple roots) on $\llbracket 1,n \rrbracket$ has a tree structure. For all $i \in \llbracket1,n \rrbracket$:
	\begin{itemize}
		\item If $\msf{Op}(i) \in \prop$, then $i$ is a leaf;
		\item If $\msf{Op}(i) \in \msf{Op}(\prop,\Ag)_{\mathsf{Un}}$, then $i$ has one child $l(i)$;
		\item If $\msf{Op}(i) \in \msf{Op}(\prop,\Ag)_{\mathsf{Bin}}$, then $i$ has two children $l(i),r(i)$.
	\end{itemize}
	On this tree structure, let us prove by induction the property $\mathcal{P}(i)$ for $i \in \llbracket1,n \rrbracket$: for all $q \in Q$, we have $q \models \Phi^v(i) \text{ iff } v(y_{i,q}^C) = 1$.
	
	\begin{itemize}
		\item Let $i \in \llbracket1,n \rrbracket$ be a leaf. In that case, we have $\Phi^v(i) = \mathsf{Op}(i) \in \prop$. Since $v \models \Omega_\prop^{\mathsf{sem}}$, for all $q \in Q$, we have $v(y_{i,q}^C) = 1$ if and only if $\Phi^v(i) \in \pi(q)$, i.e. $q \models \Phi^v(i)$. Hence, $\mathcal{P}(i)$ holds.
		\item Let $i \in \llbracket1,n \rrbracket$ be with a single child $l(i) \in \llbracket1,n \rrbracket$. Assume that $\mathcal{P}(l(i))$ holds. There are three cases:
		\begin{itemize}
			\item Assume that $\mathsf{Op}(i) = \neg$. Since $v \models \Omega_\neg^{\mathsf{sem}}$, we have $v(y_{i,q}^C) = 1 - v(y_{l(i),q}^C)$. Therefore, by $\mathcal{P}(l(i))$, for all $q \in Q$, we have $q \models \Phi^v(i) = \neg \Phi^v(l(i))$ iff $v(y_{l(i),q}^C) = 0$ iff $v(y_{i,q}^C) = 1$. 
			\item Assume that $\mathsf{Op}(i) = \fanBr{A_i} \lX$. Since $v \models \Omega_{\lX}^{\mathsf{sem}}$, for all $q \in Q$, we have $v(y_{i,q}^C) = 1$ iff $v \models \Omega_{q,i}^{\mathsf{Pre}}(y_{l(i)}^C)$ iff (by Lemma~\ref{lem:pre} and $\mathcal{P}(l(i))$) $q \in \mathsf{Pre}_{A_i}(\SAT_C(\Phi^v(l(i))))$ iff $q \models \fanBr{A_i} \lX \Phi^v(l(i)) = \Phi^v(i)$.
			\item Assume that $\mathsf{Op}(i) = \fanBr{A_i} \lG$. By Lemma~\ref{lem:G-correctness} and $\mathcal{P}(l(i))$, for all $q \in Q$, we have $q \models \Phi^v(i)$ iff $v(y_{i,q}^C) = 1$.  
		\end{itemize}
		Hence, $\mathcal{P}(i)$ holds.
		\item Let $i \in \llbracket1,n \rrbracket$ be with two children $l(i),r(i) \in \llbracket1,n \rrbracket$. Assume that $\mathcal{P}(l(i))$ and $\mathcal{P}(r(i))$ hold. There are two cases:
		\begin{itemize}
			\item Assume that $\mathsf{Op}(i) = \wedge$. Since $v \models \Omega_\wedge^{\mathsf{sem}}$, it follows that $v(y_{i,q}^C) = \min(v(y_{l(i),q}^C),v(y_{r(i),q}^C))$. Therefore, by $\mathcal{P}(l(i))$ and $\mathcal{P}(r(i))$, for all $q \in Q$, we have $q \models \Phi^v(i) = \Phi^v(l(i)) \wedge \Phi^v(r(i))$ iff $v(y_{l(i),q}^C) = 1$ and $v(y_{r(i),q}^C) = 1$ iff $v(y_{i,q}^C) = 1$. 
			\item Assume that $\mathsf{Op}(i) = \fanBr{A_i} \lU$. By Lemma~\ref{lem:U-correctness} and $\mathcal{P}(l(i))$ and $\mathcal{P}(r(i))$, for all $q \in Q$, we have $q \models \Phi^v(i)$ iff $v(y_{i,q}^C) = 1$.  
		\end{itemize}
		Hence, $\mathcal{P}(i)$ holds.
	\end{itemize}
	Therefore, $\mathcal{P}(i)$ holds for all $i \in \llbracket1,n \rrbracket$. The lemma is then given by $\mathcal{P}(n)$ since $\Phi^v = \Phi^v(n)$. 
\end{proof}

We can finally prove Proposition~\ref{prop:pass-learning-correctness}. 
\begin{proof}
	Assume that there is an ATL formula $\Phi$ of size $n$ consistent with $\sample$. Consider a valuation $v: X^\mathsf{Var}_n \rightarrow \B$ such that $\Phi = \Phi^v$. Let $C \in P \cup N$. Let us show that we can define $v$ on $X_C^{\mathsf{Struct}}$ such that $v \models \Omega_{n,C}^{\mathsf{sem}}$. As for the proof of Lemma~\ref{lem:central_lemma}, we define $v$ by induction on the tree structure of the set $\llbracket1,n \rrbracket$. 	%Let $i \in \llbracket1,n \rrbracket$. 
	\begin{itemize}
		\item Let $i \in \llbracket1,n \rrbracket$ be a leaf.  %$\mathsf{Op}(i) \in \prop$, 
		For all $q \in Q$, we set $v(y_{i,q}^C) := 1$ if and only if $p \in \pi(q)$. That way, the formula $\Omega^{\mathsf{sem}}_{\prop}$ is satisfied by $v$ for the index $i$. For all $0 \leq k \leq |Q|$, $v(y_{i,q,k}^C)$ is defined arbitrarily.
		\item Let $i \in \llbracket1,n \rrbracket$ be with a single child $l(i) \in \llbracket1,i-1 \rrbracket$. Assume that $v$ is defined on  $X_{C,l(i)}^{\mathsf{Struct}}$. There are three cases:
		\begin{itemize}
			\item Assume that $\mathsf{Op}(i) = \neg$. For all $q \in Q$, we let $v(y_{i,q}^C) := 1 - v(y_{l(i),q}^C)$. That way, the formula $\Omega^{\mathsf{sem}}_{\neg}$ is satisfied by $v$ for the index $i$. For all $0 \leq k \leq |Q|$, $v(y_{i,q,k}^C)$ is defined arbitrarily.
			\item Assume that $\mathsf{Op}(i) = \fanBr{A_i} \lX$. For all $q \in Q$, we set $v(y_{i,q}^C) = 1$ iff $v \models \Omega_{q,i}^{\mathsf{Pre}}(y_{l(i)}^C)$. That way, the formula $\Omega^{\mathsf{sem}}_{\lX}$ is satisfied by $v$ for the index $i$. For all $0 \leq k \leq |Q|$, $v(y_{i,q,k}^C)$ is defined arbitrarily.
			\item Assume that $\mathsf{Op}(i) = \fanBr{A_i} \lG$. Let $q \in Q$. We set $v(y^C_{i,q,0}) := v(y^C_{l(i),q})$. Furthermore, for all $0 \leq k \leq |Q|-1$, if $v \models \Omega_{q,i}^{\mathsf{Pre}}(y^C_{i,\cdot,k})$, we set $v(y^C_{i,q,k+1}) := v(y^C_{i,q,k})$, otherwise $v(y^C_{i,q,k+1}) := 0$. Finally, we set $v(y^C_{i,q}) := v(y^C_{i,q,|Q|})$. That way, the formula $\Omega^{\mathsf{sem}}_{\lG}$ is satisfied by $v$ for the index $i$.   
		\end{itemize}
		\item Let $i \in \llbracket1,n \rrbracket$ be with two children $l(i),r(i) \in \llbracket1,i-1 \rrbracket$. Assume that $v$ is defined on $X_{C,l(i)}^{\mathsf{Struct}}$ and $X_{C,r(i)}^{\mathsf{Struct}}$. There are two cases:
		\begin{itemize}
			\item Assume that $\mathsf{Op}(i) = \wedge$. For all $q \in Q$, we set $v(y_{i,q}^C) := \min(v(y_{l(i),q}^C),v(y_{r(i),q}^C))$. That way, the formula $\Omega^{\mathsf{sem}}_{\wedge}$ is satisfied by $v$ for the index $i$. For all $0 \leq k \leq |Q|$, $v(y_{i,q,k}^C)$ is defined arbitrarily.
			\item Assume that $\mathsf{Op}(i) = \fanBr{A_i} \lU$. We set $v(y^C_{i,q,0}) := v(y^C_{r(i),q})$. Furthermore, for all $0 \leq k \leq |Q|-1$, if $v \not\models \Omega_{q,i}^{\mathsf{Pre}}(y^C_{i,\cdot,k})$, we set $v(y^C_{i,q,k+1}) := v(y^C_{i,q,k})$, otherwise $v(y^C_{i,q,k+1}) := \max(v(y^C_{i,q,k}),v(y^C_{l(i),q}))$. Finally, we set $v(y^C_{i,q}) := v(y^C_{i,q,|Q|})$. That way, the formula $\Omega^{\mathsf{sem}}_{\lU}$ is satisfied by $v$ for the index $i$.   
		\end{itemize}
	\end{itemize} 
	
	We have shown that we can define the valuation $v$ on $X_C^{\mathsf{Struct}}$ such that $v \models \Omega_{n,C}^{\mathsf{sem}}$. For all $C \in P \cup N$, we define on $X_C^{\mathsf{Struct}}$ such that $v \models \Omega_{n,C}^{\mathsf{sem}}$. Then, since $\Phi$ is consistent with $S$, for all positive structures $C \in P$ and initial states $q \in I$, we have $s \models \Phi$. Therefore, by Lemma~\ref{lem:central_lemma}, $v(y^C_{n,s}) = 1$. On the other hand, for all negative structures $C \in N$, there is an initial state $q \in I$ such that we have $s \not\models \Phi$. Therefore, by Lemma~\ref{lem:central_lemma}, $v(y^C_{n,s}) = 0$. Hence, we have $v \models \Omega^{\mathsf{con}}_n$. That is, $v \models \Omega^{\S}_n$. 
	
	Now, let us assume that $\Omega^{\S}_n$ is satisfiable. Consider a valuation $v$ such that $v \models \Omega^{\S}_n$. Since $v \models \Omega_n^{\mathsf{str}}$, there is an ATL formula $\Phi^v$ of size at most $n$ that this valuation encodes. Consider a positive (resp. negative) structure $C \in P$. Since $v \models \Omega^{\mathsf{con}}_n$ and by Lemma~\ref{lem:central_lemma} (since $v \models \Omega_n^{\mathsf{sem}}$), we have $C \models \Phi$ (resp. $C \not\models \Phi$). Hence the valuation $v$ is consistent with $S$.
\end{proof}

\section{Deciding separability}
\label{sec:appen_separability}
\subsection{Proof of Lemma~\ref{lem:ATLn_ok}}
\label{proof:lem_ATLn_ok}
The proof of this lemma is not complicated, but is a little long due the need to handle all operators, and especially the operators $\lG$ and $\lU$. 
\begin{lemma}
	\label{lem:equiv_U_n_X}
	Let $n \in \N$. Consider a set of agents $\Ag$, a coalition of agents $A \subseteq \Ag$ and an ATL formula $\Phi = \fanBr{A} \Phi_1 \lU \Phi_2$. 
	
	For all $k \in \N$, we define the ATL formula $\Phi_{\lX,k}(A,\Phi_1,\Phi_2)$. We set:
	\begin{itemize}
		\item $\Phi_{\lX,0}(A,\Phi_1,\Phi_2) := \Phi_2$;
		\item for all $k \in \N$, $\Phi_{\lX,k+1}(A,\Phi_1,\Phi_2) := \Phi_{\lX,k}(A,\Phi_1,\Phi_2) \lor (\Phi_1 \wedge \fanBr{A}\lX \Phi_{\lX,k}(A,\Phi_1,\Phi_2))$.
	\end{itemize}

	Let $n \in \N$. For all CGS $C$ with $n$ states, for all states $q \in C$ and for all $k \geq n$, we have $q \models \Phi$ if and only if $q \models \Phi_{\lX,k}(A,\Phi_1,\Phi_2)$. 
\end{lemma}
\begin{proof}
	Consider a CGS $C$ with $n$ states. As for the proof of Lemma~\ref{lem:U-correctness}, let us consider Algorithm~\ref{alg:algo_U}. For all $0 \leq j \leq |Q|$, we let $S_j \subseteq Q$ denote the set $S$ when exiting the \emph{for} loop when $k = j$. When $j = 0$, this corresponds to initialization of the set $S$ before the \emph{for} loop is ever entered. We show by induction on $0 \leq k \leq n$ the property $\mathcal{P}(k)$: $S_k = \SAT_C(\Phi_{\lX,k}(A,\Phi_1,\Phi_2))$.
	
	For $k = 0$, we have $\Phi_{\lX,k}(A,\Phi_1,\Phi_2) = \Phi_2$, hence $\mathcal{P}(0)$ straightforwardly holds. Assume now that $\mathcal{P}(k)$ holds for some $k \in \N$. Let $q \in Q$. We have:
	\begin{align*}
		q \in S_{k+1} & \text{ iff } q \in S_k \text{ or }q \in \SAT_C(\Phi_1) \cap \mathsf{Pre}_A(S_k) \\
		& \text{ iff } q \models \Phi_{\lX,k}(A,\Phi_1,\Phi_2) \text{ or }q \models \Phi_{\lX,k}(A,\Phi_1,\Phi_2) \text{ and } q \models \fanBr{A}\lX \Phi_{\lX,k}(A,\Phi_1,\Phi_2) \\
		& \text{ iff } q \models \Phi_{\lX,k+1}(A,\Phi_1,\Phi_2)
	\end{align*}
	Thus, $\mathcal{P}(k+1)$ follows. In fact, $\mathcal{P}(k)$ holds for all $0 \leq k \leq n$. Since $S_n = \SAT_C(\Phi)$, it follows that we have $\SAT_C(\Phi_{\lX,n}(A,\Phi_1,\Phi_2)) = \SAT_C(\Phi)$. 
	
	Let us show that $\SAT_C(\Phi_1) \cap \mathsf{Pre}_A(S_n) \subseteq S_n = \SAT_C(\Phi)$. There are two cases:
	\begin{itemize}
		\item Either $S_n = Q$, and in that case the inclusion is obvious. 
		\item Or $S_n \neq Q$, and in that case, since for all $0 \leq k \leq n-1$, we have $S_k \subseteq S_{k+1}$, and $|Q| = n$, it follows that there is some $k \leq n-1$ such that $S_k = S_{k+1}$. That is, $\SAT_C(\Phi_1) \cap \mathsf{Pre}_A(S_k) \subseteq S_k$. This straightforwardly implies that, for all $k \leq i \leq n$, we have $S_k = S_i$ and $\SAT_C(\Phi_1) \cap \mathsf{Pre}_A(S_i) \subseteq S_i$. 
	\end{itemize}

	Now, consider any $k \in \N$ and assume that $\SAT_C(\Phi_{\lX,k}(A,\Phi_1,\Phi_2)) = \SAT_C(\Phi)$. Let us show that $\SAT_C(\Phi_{\lX,k+1}(A,\Phi_1,\Phi_2)) = \SAT_C(\Phi)$. Clearly, $\Phi_{\lX,k}(A,\Phi_1,\Phi_2) \implies \Phi_{\lX,k+1}(A,\Phi_1,\Phi_2)$, hence $\SAT_C(\Phi) %= \SAT_C(\Phi_{\lX,k}(A,\Phi_1,\Phi_2)) 
	\subseteq \SAT_C(\Phi_{\lX,k+1}(A,\Phi_1,\Phi_2))$. On the other hand, consider any $q \in \SAT_C(\Phi_{\lX,k+1}(A,\Phi_1,\Phi_2))$. If $q \notin \SAT_C(\Phi)$, then %$q \models \Phi_1 \wedge \fanBr{A}\lX \Phi_{\lX,k}(A,\Phi_1,\Phi_2)$, i.e. 
	$q \in \SAT_C(\Phi_1) \cap \mathsf{Pre}_A(\SAT_C(%\fanBr{A}\lX \Phi_{\lX,k}(A,\Phi_1,\Phi_2)
	\Phi))% = \SAT_C(\Phi_1) \cap \mathsf{Pre}_A(\SAT_C(\Phi)) \subseteq \SAT_C(\Phi)
	$. 
	
	Hence, we do have: $\SAT_C(\Phi_{\lX,k+1}(A,\Phi_1,\Phi_2)) = \SAT_C(\Phi)$. Therefore, for all $k \geq n$, we have $\SAT_C(\Phi_{\lX,k}(A,\Phi_1,\Phi_2)) = \SAT_C(\Phi)$. The lemma follows.
\end{proof}

We state a lemma analogous to Lemma~\ref{lem:equiv_U_n_X} for the operator $\lG$.
\begin{lemma}
	\label{lem:equiv_G_n_X}
	Consider a set of agents $\Ag$, a coalition of agents $A \subseteq \Ag$ and an ATL formula $\Phi = \fanBr{A} \lG \Phi'$. 
	
	For all $k \in \N$, we define the ATL formula $\Phi_{\lX,k}(A,\Phi')$. We set:
	\begin{itemize}
		\item $\Phi_{\lX,0}(A,\Phi') := \Phi'$;
		\item for all $k \in \N$, $\Phi_{\lX,k+1}(A,\Phi') := \Phi_{\lX,k}(A,\Phi') \wedge \fanBr{A}\lX \Phi_{\lX,k}(A,\Phi')$.
	\end{itemize}
	
	Let $n \in \N$. For all CGS $C$ with $n$ states, for all states $q \in C$ and for all $k \geq n$, we have $q \models \Phi$ if and only if $q \models \Phi_{\lX,k}(A,\Phi')$. 
\end{lemma}
We proceed very similarly than the proof of Lemma~\ref{lem:equiv_U_n_X}.
\begin{proof}
	Consider a CGS $C$ with $n$ states. As for the proof of Lemma~\ref{lem:G-correctness}, let us consider Algorithm~\ref{alg:algo_G}. For all $0 \leq j \leq |Q|$, we let $S_j \subseteq Q$ denote the set $S$ when exiting the \emph{for} loop when $k = j$. When $j = 0$, this corresponds to initialization of the set $S$ before the \emph{for} loop is ever entered. We show by induction on $0 \leq k \leq n$ the property $\mathcal{P}(k)$: $S_k = \SAT_C(\Phi_{\lX,k}(A,\Phi'))$.
	
	For $k = 0$, we have $\Phi_{\lX,k}(A,\Phi') = \Phi'$, hence $\mathcal{P}(0)$ straightforwardly holds. Assume now that $\mathcal{P}(k)$ holds for some $k \in \N$. Let $q \in Q$. We have:
	\begin{align*}
		q \in S_{k+1} & \text{ iff } q \in S_k \text{ and }q \in \mathsf{Pre}_A(S_k) \\
		& \text{ iff } q \models \Phi_{\lX,k}(A,\Phi') \text{ and } q \models \fanBr{A}\lX \Phi_{\lX,k}(A,\Phi') \\
		& \text{ iff } q \models \Phi_{\lX,k+1}(A,\Phi')
	\end{align*}
	Thus, $\mathcal{P}(k+1)$ follows. In fact, $\mathcal{P}(k)$ holds for all $0 \leq k \leq n$. Since $S_n = \SAT_C(\Phi)$, it follows that we have $\SAT_C(\Phi_{\lX,n}(A,\Phi')) = \SAT_C(\Phi)$. 
	
	Let us show that $\SAT_C(\Phi) \subseteq \mathsf{Pre}_A(\SAT_C(\Phi))$. We have $S_n = \SAT_C(\Phi)$. There are two cases:
	\begin{itemize}
		\item Either $S_n = \emptyset$, and in that case the inclusion is obvious. 
		\item Or $S_n \neq \emptyset$, and in that case, since for all $0 \leq k \leq n-1$, we have $S_k \supseteq S_{k+1}$, and $|Q| = n$, it follows that there is some $k \leq n-1$ such that $S_k = S_{k+1}$ and $S_k \subseteq \mathsf{Pre}_A(S_k)$. This straightforwardly implies that, for all $k \leq i \leq n$, we have $S_k = S_i$ and $S_i \subseteq \mathsf{Pre}_A(S_i)$. 
	\end{itemize}
	
	Now, consider any $k \in \N$ and assume that $\SAT_C(\Phi_{\lX,k}(A,\Phi')) = \SAT_C(\Phi)$. Let us show that $\SAT_C(\Phi_{\lX,k+1}(A,\Phi')) = \SAT_C(\Phi)$. Clearly, $\Phi_{\lX,k+1}(A,\Phi') \implies \Phi_{\lX,k}(A,\Phi')$, hence $\SAT_C(\Phi_{\lX,k+1}(A,\Phi_1,\Phi_2)) \subseteq \SAT_C(\Phi)$. On the other hand, consider any $q \in \SAT_C(\Phi) \subseteq \mathsf{Pre}_A(\SAT_C(\Phi)) = \mathsf{Pre}_A(\SAT_C(\Phi_{\lX,k}(A,\Phi')))$. Therefore, we have $q \in \SAT_C(\Phi_{\lX,k}(A,\Phi')) \cap \mathsf{Pre}_A(\SAT_C(\Phi_{\lX,k}(A,\Phi'))) = \SAT_C(\Phi_{\lX,k+1}(A,\Phi'))$. 
	
	Hence, we do have: $\SAT_C(\Phi_{\lX,k+1}(A,\Phi')) = \SAT_C(\Phi)$. Therefore, for all $k \geq n$, we have $\SAT_C(\Phi_{\lX,k}(A,\Phi')) = \SAT_C(\Phi)$. The lemma follows.
\end{proof}

We can now proceed to the proof of Lemma~\ref{lem:ATLn_ok}.
\begin{proof}
	Let $n \in \N$. 
	%Consider a sample $S = (P,N)$. Let $n : \max_{C \in P \cup N} |Q_C|$ where $Q_C$ denotes the set of states of the CGS $C$. 
	Let us prove by induction on ATL formulas $\Phi$ the property $\mathcal{P}(\Phi)$: there exists an %if $\Phi$ is consistent with the sample $S$, then there is an 
	ATL-$\lX$ formula $\Phi'$ that is equivalent $\Phi$ on CGS with at most $n$ states.%also consistent with the sample $S$. 
	
	This straightforwardly holds for ATL formulas $\Phi \in \prop$ since such formulas are also ATL-$\lX$ formulas. Then:
	\begin{itemize}
		\item Let $\Phi = \neg \Phi_1$ be an ATL formula and assume that $\mathcal{P}(\Phi_1)$ holds. Consider an ATL-$\lX$ formulas $\Phi_1'$ equivalent to $\Phi_1$ on CGS with at most $n$ states. Then, $\Phi' := \neg \Phi_1'$ is an ATL-$\lX$ formula equivalent to $\Phi$ on CGS with at most $n$ states. Hence, $\mathcal{P}(\Phi)$ holds.
		\item Let $\Phi = \fanBr{A} \lX \Phi_1$ be an ATL formula for a coalition of agents $A \subseteq \Ag$ and assume that $\mathcal{P}(\Phi_1)$ holds. Consider an ATL-$\lX$ formulas $\Phi_1'$ equivalent to $\Phi_1$ on CGS with at most $n$ states. Then, $\Phi' := \fanBr{A} \lX \Phi_1'$ is an ATL-$\lX$ formula equivalent to $\Phi$ on CGS with at most $n$ states. Hence, $\mathcal{P}(\Phi)$ holds.
		\item Let $\Phi = \fanBr{A} \lG \Phi_1$ be an ATL formula for a coalition of agents $A \subseteq \Ag$ and assume that $\mathcal{P}(\Phi_1)$ holds. Consider an ATL-$\lX$ formulas $\Phi_1'$ equivalent to $\Phi_1$ on CGS with at most $n$ states. Then, let $\Phi' := \Phi_{\lX,n}(A,\Phi_1')$ be the ATL-$\lX$ formula from Lemma~\ref{lem:equiv_G_n_X}. By Lemma~\ref{lem:equiv_G_n_X}, $\Phi$ and $\Phi'$ are equivalent on CGS with at most $n$ states. Hence, $\mathcal{P}(\Phi)$ holds.
		\item Let $\Phi = \Phi_1 \wedge \Phi_2$ be an ATL formula and assume that $\mathcal{P}(\Phi_1)$ and $\mathcal{P}(\Phi_2)$ hold. Consider two ATL-$\lX$ formulas $\Phi_1',\Phi_2'$ equivalent to $\Phi_1,\Phi_2$ respectively on CGS with at most $n$ states. Then, $\Phi' := \Phi_1' \wedge \Phi_2'$ is an ATL-$\lX$ formula equivalent to $\Phi$ on CGS with at most $n$ states. Hence, $\mathcal{P}(\Phi)$ holds.
		\item Let $\Phi = \fanBr{A} \Phi_1 \lU \Phi_2$ be an ATL formula for a coalition of agents $A \subseteq \Ag$ and assume that $\mathcal{P}(\Phi_1)$ and $\mathcal{P}(\Phi_2)$ hold. Consider two ATL-$\lX$ formulas $\Phi_1',\Phi_2'$ equivalent to $\Phi_1,\Phi_2$ respectively on CGS with at most $n$ states. Then, let $\Phi' := \Phi_{\lX,n}(A,\Phi_1',\Phi_2')$ be the ATL-$\lX$ formula from Lemma~\ref{lem:equiv_U_n_X}. By Lemma~\ref{lem:equiv_U_n_X}, $\Phi$ and $\Phi'$ are equivalent on CGS with at most $n$ states. Hence, $\mathcal{P}(\Phi)$ holds.
	\end{itemize}
	Overall, the property $\mathcal{P}(\Phi)$ holds for all ATL formula $\Phi$. The lemma follows.
\end{proof}

\iffalse
Now, our goal is to compute the set that we define below. 
\begin{definition}
	Consider a sample $S = (P,N)$ os CGS. %concurrent game structure $C = \langle Q,I,k,\prop, \pi, \sigma, d,\delta \rangle$. 
	We let $\mathsf{Distinguish}(S) \subseteq Q^2$ denote the set:
	\begin{equation*}
		\mathsf{Distinguish}(S) := \{ (q,q') \in Q^2 \mid \exists \Phi \in \text{ ATL-$\lX$ },\; q \models \Phi \Leftrightarrow q' \not\models \Phi \}
	\end{equation*}
\end{definition}
\fi

\subsection{Algorithm to decide the separability with $\neg,\wedge,\lX$}
\label{subsec:algorithms}
\begin{algorithm}[t]
	\caption{Given a sample $S$, and $R \subseteq Q^2$, computes the set $\mathsf{Upd}(R) \subseteq Q^2$}
	\label{algo:compute_update}
	\textbf{Input}: Sample $S$ of CGS and $R \subseteq Q^2$
	\begin{algorithmic}[1]
		\State $\mathsf{Upd}(R) \gets \emptyset$
		\For{$(q,q') \in Q^2$}
		\State $\mathsf{ToAdd} \gets \mathsf{False}$
		\State $\mathsf{RelAg} \gets \emptyset$
		\For{$a \in \Ag$}
		\If{$d(q,a) \cdot d(q',a) \geq 2$}
		\State $\mathsf{RelAg} \gets \mathsf{RelAg} \cup \{a \}$
		\EndIf
		\EndFor
		\For{$A \subseteq \mathsf{RelAg}$}
		\For{$\alpha \in \mathsf{Act}_A(q)$}
		\State $\mathsf{CheckTuplePos} \gets \mathsf{True}$
		\For{$\alpha' \in \mathsf{Act}_{A}(q')$}
		\State $\mathsf{CheckTupleNeg} \gets \mathsf{False}$
		\For{$q_{\alpha'} \in \mathsf{Succ}(q',\alpha')$}
		\State $\mathsf{CheckState} \gets \mathsf{True}$
		\For{$q_{\alpha} \in \mathsf{Succ}(q,\alpha)$}
		\If{$(q_{\alpha},q_{\alpha'}) \notin R$}
		\State $\mathsf{CheckState} \gets \mathsf{False}$%; $\mathsf{break}$
		\EndIf
		\EndFor
		\If{$\mathsf{CheckState}$}
		\State $\mathsf{CheckTupleNeg} \gets \mathsf{True}$%; $\mathsf{break}$
		\EndIf
		\EndFor
		\If{$\lnot \mathsf{CheckTupleNeg}$}
		\State $\mathsf{CheckTuplePos} \gets \mathsf{False}$%; $\mathsf{break}$
		\EndIf
		\EndFor
		\If{$\mathsf{CheckTuplePos}$}
		\State $\mathsf{ToAdd} \gets \mathsf{True}$%; $\mathsf{break}$
		\EndIf
		\EndFor
		\EndFor
		\If{$\mathsf{ToAdd}$}
		\State $\mathsf{Upd}(R) \gets \mathsf{Upd}(R) \cup \{ (q,q'),(q',q)\}$
		\EndIf
		\EndFor
		\Return $\mathsf{Upd}(R)$
	\end{algorithmic}
\end{algorithm}

Our goal is now to design a polynomial time algorithm that decides if there exists a formula that is consistent with $S$. First of all, we design Algorithm~\ref{algo:compute_update} that computes the set $\mathsf{Upd}(R) \subseteq Q^2$ from a set $R \subseteq Q^2$. This algorithm satisfies the lemma below.
\begin{lemma}
	\label{lem:compute_upd}
	Given a sample $S = (P,N)$ of CGS and a subset $R \subseteq Q^2$ of pairs of states as input, Algorithm~\ref{algo:compute_update} returns the set $\mathsf{Upd}(R) \subseteq Q^2$. Furthermore, its complexity is in $O(n^4 \cdot r^6 + n^2 \cdot k)$. 
\end{lemma}
\begin{proof}
	The steps taken by Algorithm~\ref{algo:compute_update} follow closely the definition of the set $\mathsf{Upd}(R)$. Let us explain what every step of the computation does, along with the complexity of the multiple \emph{for} loops:
	\begin{itemize}
		\item Line 2-38: this \emph{for} loop looks over all pairs of states to add, and is entered $n^2$ times.
		\item Line 4-8: computes the set $\mathsf{RelAg}(q) \cup \mathsf{RelAg}(q')$, takes $O(k)$ steps. 
		\item Line 10-34: this \emph{for} loop looks over coalitions of agents. 
		
		It is entered  $|2^{\mathsf{RelAg}(q) \cup \mathsf{RelAg}(q')}| \leq |2^{\mathsf{RelAg}(q)}| \cdot |2^{\mathsf{RelAg}(q')}|$ times, with $|2^{\mathsf{RelAg}(q)}| \leq ||q|| \leq r$ and $|2^{\mathsf{RelAg}(q')}| \leq ||q'|| \leq r$. Therefore, this loop is entered at most $r^2$ times.
		\item Line 11-33: this \emph{for} loop looks over tuples of actions $\alpha \in \mathsf{Act}_A(q)$% for the coalition $A$ at state $q$
		. The boolean $\mathsf{CheckTuplePos}$ is initially set to true. If a tuple of actions $\alpha' \in \mathsf{Act}_A(q)$ % for the coalition $A$ at the state $q'$ 
		that makes the sets $\mathsf{Succ}(q,\alpha)$ and $\mathsf{Succ}(q,\alpha)$ undistinguishable w.r.t. $R$ is found, then $\mathsf{CheckTuplePos}$ is set to false. In that case, the algorithm, checks for another tuple of actions $\alpha \in \mathsf{Act}_A(q)$% for the coalition $A$ at state $q$
		. If it is not the case, then we have found an appropriate coalition of agents and tuples of actions. Therefore, $(q,q') \in \mathsf{Upd}(R)$, and the boolean $\mathsf{ToAdd}$ is set to true. 
		
		This loop is entered at most $||q|| \leq r$ times.
		\item Line 13-29: this \emph{for} loop looks over tuples of actions $\alpha' \in \mathsf{Act}_A(q')$% for the coalition $A$ at state $q'$
		. The boolean $\mathsf{CheckTupleNeg}$ is initially set to false. If a successor state $q_{\alpha'} \in \mathsf{Succ}(q,\alpha')$ that can be distinguished, w.r.t. $R$, from all the states in $\mathsf{Succ}(q,\alpha)$ is found, then $\mathsf{CheckTupleNeg}$ is set to true. In that case, the algorithm checks that it is also the case for the other tuples of actions $\alpha' \in \mathsf{Act}_A(q')$ %for the coalition $A$ at state $q'$
		. If it is not the case, the tuple of actions $\alpha$ does not work, $\mathsf{CheckTuplePos}$ is set to false, and the algorithm looks for tuple of actions $\alpha$. 
		
		This loop is entered at most $||q'|| \leq r$ times.
		\item Line 15-25: this \emph{for} loop looks over possible successor states $q_{\alpha'} \in \mathsf{Succ}(q,\alpha')$. The boolean $\mathsf{CheckState}$ is initially set to true. If a successor state $q_{\alpha} \in \mathsf{Succ}(q,\alpha)$ that cannot be distinguished, w.r.t. $R$, from $q_{\alpha'}$ is found, then $\mathsf{CheckState}$ is set to false. In that case, the algorithm looks for other successor states $q_{\alpha'} \in \mathsf{Succ}(q,\alpha')$. If it is not the case, then it is possible to distinguish $q_{\alpha'}$ %the sets of successors $\mathsf{Succ}(q,\alpha)$ and 
		the sets of successors $\mathsf{Succ}(q,\alpha)$ w.r.t. $R$. Therefore, the boolean $\mathsf{CheckStateNeg}$ is set to true. 
		
		This loop is entered at most $||q'|| \leq r$ times.
		\item Line 17-21: this \emph{for} loop looks over possible successor states $q_{\alpha} \in \mathsf{Succ}(q,\alpha)$. 
		
		This loop is entered at most $||q|| \leq r$ times. Furthermore, checking that $(q_\alpha,q_{\alpha'}) \notin R$ takes time $O(|R|)$ with $|R| \leq n^2$.
	\end{itemize}
	Overall, this algorithm runs in time $O(n^2 \cdot (k + r^2 \cdot r \cdot r \cdot r \cdot r \cdot n^2) = O(n^2 \cdot k + n^4  \cdot r^6)$. 
\end{proof}

\begin{algorithm}[t]
	\caption{Given a sample $S$, computes the set $\mathsf{Distinguish}(S)$}
	\label{algo:compute_distinguish}
	\textbf{Input}: Sample $S$ of CGS
	\begin{algorithmic}[1]
		\State $R \gets \emptyset$
		\State $n \gets |Q|$
		\For{$(q,q') \in Q^2$}
		\If{$\pi(q) \neq \pi(q')$}
		\State $R \gets R \cup \{ (q,q'),(q',q) \}$
		\EndIf
		\EndFor
		\For{$1 \leq i \leq n^2$}
		\State $R \gets R \cup \mathsf{Upd}(R)$
		\EndFor 
		\Return $R$
	\end{algorithmic}
\end{algorithm}

Algorithm~\ref{algo:compute_distinguish} computes the $\mathsf{Distinguish}(S) \subseteq Q^2$ given a sample $S = (P,N)$ of CGS. It satisfies the lemma below. 
\begin{lemma}
	\label{lem:compute_dist}
	Given a sample $S = (P,N)$ of CGS, Algorithm~\ref{algo:compute_distinguish} returns the set $\mathsf{Distinguish}(S) \subseteq Q^2$. Its complexity is in $O(n^6 \cdot r^6 + n^4 \cdot k + n^2 \cdot p)$. 
\end{lemma}
\begin{proof}
	The steps taken by Algorithm~\ref{algo:compute_distinguish} follow closely the characterization of the set $\mathsf{Distinguish}(S)$ from Lemma~\ref{lem:charac_distinguish}. 
	
	Complexity wise, the \emph{for} loop of line 3-7 is entered $n^2$ times. Furthermore, checking that $\pi(q) \neq \pi(q')$ takes time $O(p)$. In addition, the computation of Line 9 takes time $O(n^4 \cdot r^6 + n^2 \cdot k)$, by Lemma~\ref{lem:compute_upd}, with the for loop of Line 8-10 entered $n^2$ times. The lemma follows.
\end{proof}

\begin{algorithm}[t]
	\caption{Given a sample $S$, return yes iff there is an ATL formula that is consistent with $S$}
	\label{algo:decide_consistent}
	\textbf{Input}: Sample $S$ of CGS
	\begin{algorithmic}[1]
		\State $D \gets \mathsf{Distinguish}(S)$
		\State $n \gets |Q|$
		\For{$C_N \in N$}
		\State $\mathsf{CheckNeg} \gets \mathsf{False}$
		\For{$q_{C_N} \in I_{C_N}$}
		\State $\mathsf{CheckState} \gets \mathsf{True}$
		\For{$C_P \in P$}
		\For{$q_{C_P} \in I_{C_P}$}
		\If{$(q_{C_P},q_{C_N}) \notin D$}
		\State $\mathsf{CheckState} \gets \mathsf{False}$
		\EndIf
		\EndFor
		\EndFor
		\If{$\mathsf{CheckState}$}
		\State $\mathsf{CheckNeg} \gets \mathsf{True}$
		\EndIf
		\EndFor
		\If{$\neg \mathsf{CheckNeg}$}
		\Return No
		\EndIf
		\EndFor
		\Return Yes
	\end{algorithmic}
\end{algorithm}

Finally, Algorithm~\ref{algo:decide_consistent} decides if there exists an ATL formula consistent with a sample $S = (P,N)$ of CGS. It satisfies the lemma below. 
\begin{lemma}
	\label{lem:final}
	Given a sample $S = (P,N)$ of CGS, Algorithm~\ref{algo:decide_consistent} decides if there exists an ATL-$\lX$ formula consistent with $S$. Its complexity is in $O(n^6 \cdot r^6 + n^4 \cdot k + n^2 \cdot p)$. 
\end{lemma}
\begin{proof}
	The steps taken by Algorithm~\ref{algo:compute_distinguish} follow the characterization from Lemma~\ref{lem:decide_from_distinguish}. 
	
	Complexity wise, the \emph{for} loop of line 3-21 checks for at most $O(n^2)$ pairs of states, with the check of Line 9 takins time at most $|D|$ with $|D| \leq n^2$. Hence, the complexity of \emph{for} loop is in $O(n^4)$, which is less that the time taken to compute the set $\mathsf{Distinguish}(S)$ at Line 1, which is $O(n^6 \cdot r^6 + n^4 \cdot k + n^2 \cdot p)$ by Lemma~\ref{lem:compute_dist}. 
\end{proof}

%Theorem~\ref{thm:decide_polynomial_time_ATLn} is now a direct consequence of Lemmas~\ref{lem:final} and~\ref{lem:true_ATLn_ok}. Theorem~\ref{thm:decide_polynomial_time} in the main part of the paper directly follows from Theorem~\ref{thm:decide_polynomial_time_ATLn}. 

\subsection{Complexity of deciding the separability of any fragment of ATL}
\label{subsec:complexity_exponential}
We focus here on the complexity of effectively computing the set $\mathsf{Apr}(S,\mathsf{O})$. We start with the complexity of computing the update function.
\begin{lemma}
	\label{lem:comple:compute:upd}
	Consider a sample $S = (P,N)$ of CGS and a subset of operators $\mathsf{O} \subseteq \mathsf{OP}$. For all $R \subseteq 2^Q$, the set $	\mathsf{Upd}_{\mathsf{O}}(R)$ can be computed in time $O(2^{k+2n} \cdot n^2 \cdot M \cdot r)$. 
\end{lemma}
\begin{proof}
	Consider any coalition of agents $A \subseteq \Ag$. Let $T \in R$.
	\begin{itemize}
		\item The set $\mathsf{Upd}_{\neg}(S,T)$ can be computed in time $O(n)$.
		\item To check that any state $q \in Q$ is in $\mathsf{Pre}_{A}(T)$, one can go through all possible tuples of actions $\alpha$ for the coalition $A$ at state $q$ and check if $\mathsf{Succ}(q,\alpha) \subseteq T$. This can be done in time $O(||q|| \cdot ||q|| \cdot |T|)$. Therefore, the set $\mathsf{Upd}_{\lX}(S,A,T)$ can be computed in time $O(M \cdot r \cdot n)$.
		\item As above, given any subset of states $S \subseteq Q$, for all $q \in Q$, checking that $q \in \mathsf{Pre}_A(S)$ can be done in time $O(||q||^2 \cdot |S|)$. Therefore, the set $\mathsf{Upd}_{\lG}(S,A,T)$ can be computed by Algorithm~\ref{alg:algo_G}, in time $O(n \cdot M \cdot r \cdot n)$.
	\end{itemize}
	Let $T,T' \in R$.
	\begin{itemize}
		\item The set $\mathsf{Upd}_{\wedge}(S,T,T')$ can be computed in time $O(n)$.
		\item As above, the set $\mathsf{Upd}_{\lU}(S,A,T,T')$ can be computed by Algorithm~\ref{alg:algo_U} in time $O(n \cdot M \cdot r \cdot n)$.
	\end{itemize}
	
	Therefore, computing the set $\mathsf{Upd}_{\mathsf{O}}(R)$ can be done in time $O(2^k \cdot 2^{3n} \cdot n^2 \cdot M \cdot r)$.
\end{proof}

The complexity of computing the set $\mathsf{Apr}(S,\mathsf{O})$ follows.
\begin{lemma}
	\label{lem:comp}
	Consider a sample $S = (P,N)$ of CGS and a subset of operators $\mathsf{O} \subseteq \mathsf{OP}$. For all $R \subseteq 2^Q$, the set $\mathsf{Apr}(S,\mathsf{O})$ can be computed in time $O(2^{k+2n} \cdot n^2 \cdot M \cdot r)$. 
\end{lemma}
\begin{proof}
	Computing the set $\mathsf{Apr}_1(S,\mathsf{O})$ can be done in time $O(p \cdot n \cdot 2^n)$. The lemma then follows by Lemma~\ref{lem:comple:compute:upd} %, that computing the set $\mathsf{Apr}_1(S,\mathsf{O})$ can be done in time $O(2^{k+3n} \cdot (n^2 \cdot M \cdot r + p))$.
	and the definition of the set $\mathsf{Apr}(S,\mathsf{O}) = \mathsf{Apr}_{2^n}(S,\mathsf{O})$.
\end{proof}

We can finally proceed to the proof of Lemma~\ref{lem:decide_exponential_time_ATLn}.
\begin{proof}
	To decide if there is an $\mathbf{ATL}(\mathsf{O})$-formula consistent with $S$, we first compute the set $\mathsf{Apr}(S,\mathsf{O}) = \mathsf{Acc}(S,\mathsf{O})$ (by Lemmas~\ref{lem:acc1} and~\ref{lem:acc2}) in time $O(2^{k+2n} \cdot n^2 \cdot M \cdot r)$ by Lemma~\ref{lem:comp}. Then, it amounts to check that there is some $T \in  \mathsf{Acc}(S,\mathsf{O})$ such that, for all $C_P \in P$ we have $I_{C_P} \subseteq T$ and for all $C_N \in N$ we have $I_{C_N} \not\subseteq T$. This can be done in time $O(2^n \cdot n^2)$.
	
	The bound on the size of the formulas that we need to consider directly follows from Lemmas~\ref{lem:acc2} and the fact that $\mathsf{Acc}(S,\mathsf{O}) = \mathsf{Apr}_{2^n}(S,\mathsf{O})$. 
\end{proof}

\section{Experimental results}
\label{sec:app-experiments}

\subsection{List of formulas}
\label{subsec:formula-list}

We present the entire list of formulas used in the experiments, which is an extension of the formulas presented in Figure~\ref{fig:runtime_results}.%Table~\ref{tab:TL-patterns}.

\begin{table}[h]
	\centering
	\setlength{\tabcolsep}{8pt}
	%\resizebox{\linewidth}{!}{%
	\begin{tabular}{|c|c|}
		\hline
		CTL formulas & ATL formulas \\
		\hline
		$\lA\lG(\lA\lF(p))$ & $\fanBr{1}\lX(p)$ \\
		$\lA\lG(\lE\lF(q))$ & $\fanBr{0}\lF(q)$ \\
		$\lA\lG(\neg p \lor \neg q)$ & $\fanBr{}\lG(p\rightarrow \fanBr{1}\lX(q))$ \\
		$\lA\lG(p \rightarrow \lA\lF(q))$ & $\fanBr{}\lG(p\rightarrow \fanBr{1}\lG(p))$ \\
		$\lA\lG(p \lor \lA\lX(\lnot q))$ & $\fanBr{}\lG(p\rightarrow \fanBr{0,\!1}\lF(q))$  \\
		$\lA\lG(\lA\lF(p) \land \lA\lF(q))$ &  $\fanBr{}\lG((p\wedge \neg q)\rightarrow \fanBr{1}\lG(p))$ \\
		\hline
	\end{tabular}
	\caption{Examples of CTL and ATL formulas used for generating benchmarks}
\end{table}

\subsection{The complete results for Figure~\ref{fig:runtime_results}.}
\label{subsec:full-experiments}

\begin{figure}[t]
	\begin{tikzpicture}
	\begin{axis}[%[log origin=infty] 
	height=42mm,
	width=50mm,
	xmin=20, ymin=-10, ymax=2000,
	enlarge x limits=true, enlarge y limits=true,
	xlabel = {Number of examples},
	ylabel = {Runtime (in secs)},
	xtick={20,40,60,80,100,120},%
	ytick={0,1000,2000},
	yticklabels = {0,1K,2K},
	ytickten= {0, 1, 2, 3},%
	%extra x ticks={2000}, extra x tick labels={\strut TO},%
	extra y ticks={2400}, %extra y tick labels=\empty%, 
	extra y tick labels={\strut TO},
	xminorticks=false,
	yminorticks=false,
	xlabel near ticks,
	ylabel near ticks,
	%xtick style={draw=major ticks},
	%ytick style={draw=none},
	label style={font=\footnotesize},
	x tick label style={font={\strut\tiny}},
	y tick label style={font={\strut\tiny}},		
	x label style = {yshift=1mm},
	legend columns = 1,
	legend style={title={CTL learning}, at={(3.17,1.5)}, font=\scriptsize, draw=none, fill=none},
	title = {CTL learning},
	title style = {font=\footnotesize, inner sep=-4pt}
	]
	
	\addlegendimage{empty legend}
	\addlegendentry{\hspace{-2mm}\textbf{CTL formulas}}
	
	\addplot[color=red,mark=x,opacity=0.7]
	plot coordinates {
		(20,10.48)
		(40,39.29)
		(60,82.13)
		(80,142.38)
		(100,216.02)
		(120,315.88)
	};
	\addlegendentry{$\lA\lG(\lA\lF(p))$}
	\addplot[mark=+,green,opacity =0.7] plot coordinates {
		(20,10.62)
		(40,38.68)
		(60,85.08)
		(80,147.78)
		(100,238.15)
		(120,341.37)
	};
	\addlegendentry{$\lA\lG(\lE\lF(q))$}
	
	\addplot[mark=*,yellow,opacity =0.7] plot coordinates {
		(20,110.39)
		(40,305.97)
		(60,634.72)
		(80,801)
		(100,1456.74)
		(120,1415.79)
	};
	\addlegendentry{$\lA\lG(\neg p\!\lor\!\neg q)$}
	
	\addplot[mark=+,orange,opacity =0.7] plot coordinates {
		(20,101.77)
		(40,243.21)
		(60,544.60)
		(80,822.465)
		(100,1328.61)
		(120,1430.26)
	};
	\addlegendentry{$\lA\lG(p\!\rightarrow\!\lA\lF(q))$}
	
	\addplot[mark=diamond,blue,opacity =0.7] plot coordinates {
		(20,11.315)
		(40,340.74)
		(60,740.52)
		(80,1020.16)
		(100,1472.84)
		(120,1590.76)
	};
	\addlegendentry{$\lA\lG(\lA\lF(p)\!\land\!\lA\lF(q))$}

	\addplot[mark=x,violet,opacity=0.9] plot coordinates {
		(20,95.01)
		(40,206.38)
		(60,450.86)
		(80,767.88)
		(100,1126.43)
		(120,1372.21)
	};
	\addlegendentry{$\lA\lG(p\!\lor\!\lA\lX(\lnot q))$}
	\end{axis}
	
	\hspace{35mm}
	
	\begin{axis}[%[log origin=infty] 
	height=42mm,
	width=50mm,
	xmin=10, ymin=0, ymax=2000,
	enlarge x limits=true, enlarge y limits=true,
	xlabel = {Number of examples},
	ylabel = {},
	xtick={10,20,30,40,50,60},%
	ytick={0,1000,2000},
	%ytickten= {0, 1, 2, 3},%
	%extra x ticks={2000}, extra x tick labels={\strut TO},%
	%extra y ticks={2400}, %extra y tick labels=\empty%, 
	%extra y tick labels={\strut TO},
	xminorticks=false,
	yminorticks=false,
	xlabel near ticks,
	ylabel near ticks,
	%xtick style={draw=major ticks},
	%ytick style={draw=none},
	label style={font=\footnotesize},
	x tick label style={font={\strut\tiny}},
	yticklabels = {},
	y tick label style={font={\strut\tiny}},		
	x label style = {yshift=1mm},
	legend columns = 1,
	legend style={at={(1.65,0.35)},font=\scriptsize,
		anchor=north, draw=none, fill=none},
	title = {ATL learning},
	title style = {font=\footnotesize, inner sep=-4pt}
	]
	
	\addlegendimage{empty legend}
	\addlegendentry{\hspace{-3mm}\textbf{ATL formulas}}
	
	\addplot[color=olive,mark=x,opacity =0.8]
	plot coordinates {
		(10, 5.5)
		(20, 21.8)
		(30, 47.69)
		(40, 84.12)
		(50, 131.62)
		(60, 189.72)
	};
	\addlegendentry{$\fanBr{1}\lX(p)$}
	
	\addplot[mark=+,blue,opacity =0.7] plot coordinates {
		(10, 4.9)
		(20, 22.72)
		(30, 48.80)
		(40, 82.75)
		(50, 131.186)
		(60, 186.65)
	};
	\addlegendentry{$\fanBr{0}\lF(q)$}
	
	\addplot[mark=triangle,red,opacity=0.7] plot coordinates {
		(10, 107.31)
		(20, 391.87)
		(30, 803.79)
		(40, 1425.39)
		(50, 1613.98)
		(60, 1788.75)
	};
	\addlegendentry{$\fanBr{}\lG(p\!\rightarrow\!\fanBr{1}\lX(q))$ }
	
	\addplot[mark=star,magenta,opacity =0.7] plot coordinates {
		(10, 45.74)
		(20, 182.86)
		(30, 415.181)
		(40, 748.00)
		(50, 1155.71)
		(60, 1484.91)
	};
	\addlegendentry{$\fanBr{}\lG(p\!\rightarrow\!\fanBr{1}\lG(p))$}

	\addplot[mark=triangle,black,opacity=0.7] plot coordinates {
		(10, 41.59)
		(20, 392.17)
		(30, 835.11)
		(40, 1424.98)
		(50, 1586.71)
		(60, 1752.46)
		
	};
	\addlegendentry{$\fanBr{}\lG(p\rightarrow \fanBr{0,\!1}\lF(q))$}
	
	\addplot[mark=*,teal,opacity=0.7] plot coordinates {
		(10, 90.92)
		(20, 1012.32)
		(30, 1690.90)
		(40, 1698.69)
		(50, 1893.51)
		(60, 1924.78)
	};
	\addlegendentry{$\fanBr{}\lG((p\!\wedge\!\neg q)\!\rightarrow\!\fanBr{1}\lG(p))$}

	\end{axis}
	
	\end{tikzpicture}
	
	\caption{Runtime of CTL and ATL learning algorithms on samples with varying number of examples (considering structures of size $\leq 20$).}
	\label{fig:app-runtime-comp}
\end{figure}
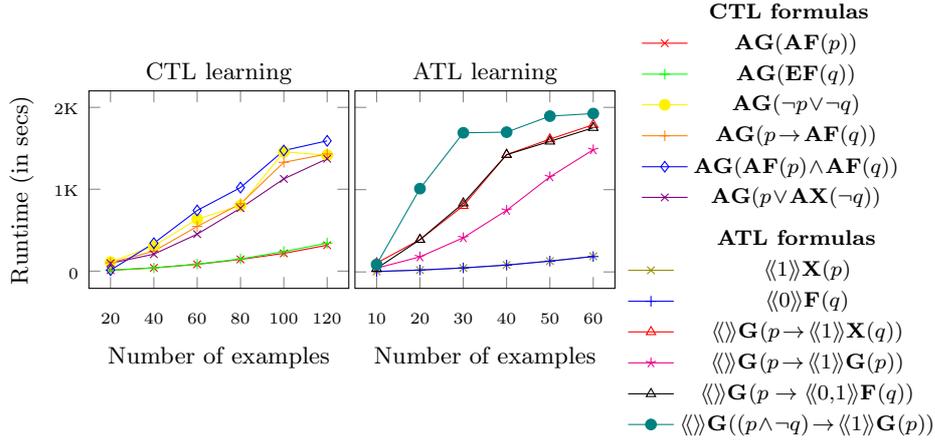
Figure~\ref{fig:app-runtime-comp} presents the entire results of the experiment conducted for Figure~\ref{fig:runtime_results}.
We omitted some formulas from Figure~\ref{fig:runtime_results} since the runtimes for some formulas overlap, especially if the formulas have similar sizes.

In the above figure, we presented the change in runtime with respect to sample sizes.
In Figure~\ref{fig:app-runtime-comp-structures}, we compare the change in runtime with respect to size of the structures in the sample.

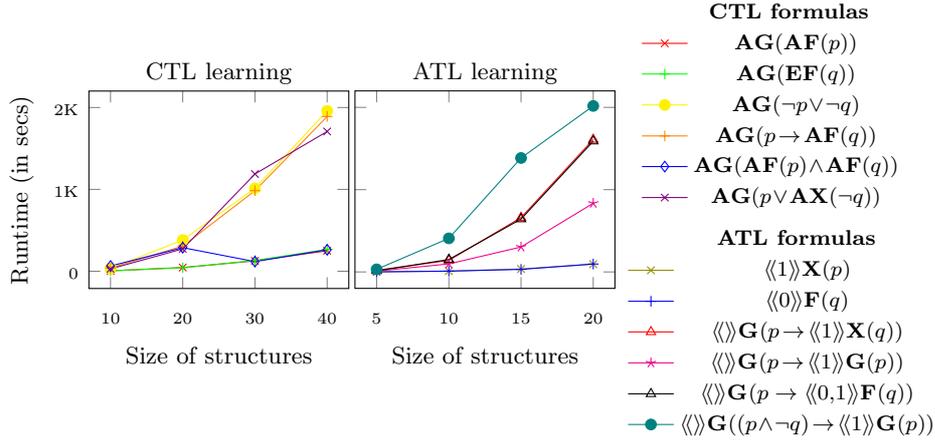
\begin{figure}[t]
	\begin{tikzpicture}
	\begin{axis}[%[log origin=infty] 
	height=42mm,
	width=50mm,
	xmin=10, ymin=-10, ymax=2000,
	enlarge x limits=true, enlarge y limits=true,
	xlabel = {Size of structures},
	ylabel = {Runtime (in secs)},
	xtick={10,20,30,40},%
	ytick={0,1000,2000},
	yticklabels = {0,1K,2K},
	ytickten= {0, 1, 2, 3},%
	%extra x ticks={2000}, extra x tick labels={\strut TO},%
	extra y ticks={2400}, %extra y tick labels=\empty%, 
	extra y tick labels={\strut TO},
	xminorticks=false,
	yminorticks=false,
	xlabel near ticks,
	ylabel near ticks,
	%xtick style={draw=major ticks},
	%ytick style={draw=none},
	label style={font=\footnotesize},
	x tick label style={font={\strut\tiny}},
	y tick label style={font={\strut\tiny}},		
	x label style = {yshift=1mm},
	legend columns = 1,
	legend style={title={CTL learning}, at={(3.17,1.5)}, font=\scriptsize, draw=none, fill=none},
	title = {CTL learning},
	title style = {font=\footnotesize, inner sep=-4pt}
	]
	
	\addlegendimage{empty legend}
	\addlegendentry{\hspace{-2mm}\textbf{CTL formulas}}
	
	\addplot[color=red,mark=x,opacity=0.7]
	plot coordinates {
		(10, 4.62)
		(20, 45.15)
		(30, 127.06)
		(40, 250.02)
	};
	\addlegendentry{$\lA\lG(\lA\lF(p))$}
	\addplot[mark=+,green,opacity =0.7] plot coordinates {
		(10, 6.05)
		(20, 43.24)
		(30, 130.96)
		(40, 267.31)
	};
	\addlegendentry{$\lA\lG(\lE\lF(q))$}
	
	\addplot[mark=*,yellow,opacity =0.7] plot coordinates {
		(10, 35.10)
		(20, 381.25)
		(30, 1012.44)
		(40, 1958.81)
	};
	\addlegendentry{$\lA\lG(\neg p\!\lor\!\neg q)$}
	
	\addplot[mark=+,orange,opacity =0.7] plot coordinates {
		(10, 38.45)
		(20, 306.54)
		(30, 983.69)
		(40, 1890.15)
	};
	\addlegendentry{$\lA\lG(p\!\rightarrow\!\lA\lF(q))$}
	
	\addplot[mark=diamond,blue,opacity =0.7] plot coordinates {
		(10, 64.17)
		(20, 287.88)
		(30, 117.62)
		(40, 261.76)
	};
	\addlegendentry{$\lA\lG(\lA\lF(p)\!\land\!\lA\lF(q))$}

	\addplot[mark=x,violet,opacity=0.9] plot coordinates {
		(10, 28.75)
		(20, 272.63)
		(30, 1187.76)
		(40, 1707.18)
	};
	\addlegendentry{$\lA\lG(p\!\lor\!\lA\lX(\lnot q))$}
	\end{axis}
	
	\hspace{35mm}
	
	\begin{axis}[%[log origin=infty] 
	height=42mm,
	width=50mm,
	xmin=5, ymin=0, ymax=2000,
	enlarge x limits=true, enlarge y limits=true,
	xlabel = {Size of structures},
	ylabel = {},
	xtick={5,10,15,20},%
	ytick={0,1000,2000},
	%ytickten= {0, 1, 2, 3},%
	%extra x ticks={2000}, extra x tick labels={\strut TO},%
	%extra y ticks={2400}, %extra y tick labels=\empty%, 
	%extra y tick labels={\strut TO},
	xminorticks=false,
	yminorticks=false,
	xlabel near ticks,
	ylabel near ticks,
	%xtick style={draw=major ticks},
	%ytick style={draw=none},
	label style={font=\footnotesize},
	x tick label style={font={\strut\tiny}},
	yticklabels = {},
	y tick label style={font={\strut\tiny}},		
	x label style = {yshift=1mm},
	legend columns = 1,
	legend style={at={(1.65,0.35)},font=\scriptsize,
		anchor=north, draw=none, fill=none},
	title = {ATL learning},
	title style = {font=\footnotesize, inner sep=-4pt}
	]
	
	\addlegendimage{empty legend}
	\addlegendentry{\hspace{-3mm}\textbf{ATL formulas}}
	
	\addplot[color=olive,mark=x,opacity =0.8]
	plot coordinates {
		(5, 1.14)
		(10, 10.21)
		(15,35.25)
		(20, 94.36)
	};
	\addlegendentry{$\fanBr{1}\lX(p)$}
	
	\addplot[mark=+,blue,opacity =0.7] plot coordinates {
		(5, 1.02)
		(10, 9.39)
		(15, 31.00)
		(20, 95.95)
	};
	\addlegendentry{$\fanBr{0}\lF(q)$}
	
	\addplot[mark=triangle,red,opacity=0.7] plot coordinates {
		(5, 19.61)
		(10, 139.99)
		(15, 659.93)
		(20, 1609.23)
	};
	\addlegendentry{$\fanBr{}\lG(p\!\rightarrow\!\fanBr{1}\lX(q))$ }
	
	\addplot[mark=star,magenta,opacity =0.7] plot coordinates {
		(5, 10.33)
		(10, 97.44)
		(15, 302.83)
		(20, 834.38)
	};
	\addlegendentry{$\fanBr{}\lG(p\!\rightarrow\!\fanBr{1}\lG(p))$}

	\addplot[mark=triangle,black,opacity=0.7] plot coordinates {
		(5, 12.77)
		(10, 150.81)
		(15, 642.34)
		(20, 1588.90)
		
	};
	\addlegendentry{$\fanBr{}\lG(p\rightarrow \fanBr{0,\!1}\lF(q))$}
	
	\addplot[mark=*,teal,opacity=0.7] plot coordinates {
		(5, 31.55)
		(10, 408.037)
		(15, 1384.81)
		(20, 2017.35)
	};
	\addlegendentry{$\fanBr{}\lG((p\!\wedge\!\neg q)\!\rightarrow\!\fanBr{1}\lG(p))$}

	\end{axis}
	
	\end{tikzpicture}
	
	\caption{Runtime of CTL and ATL learning algorithms on samples with varying size of structures (considering number of examples  $\leq 40$).}
	\label{fig:app-runtime-comp-structures}
\end{figure}

\end{document}